%% file: main.tex
\documentclass[journal]{IEEEtran}

\input{preamble}

\title{MIMO Channel Shaping and Rate Maximization\\Using Beyond-Diagonal RIS}
\author{
	\IEEEauthorblockN{
		Yang Zhao, \IEEEmembership{Member, IEEE,}
		Hongyu Li, \IEEEmembership{Member, IEEE,}\\
		Bruno Clerckx, \IEEEmembership{Fellow, IEEE,}
		and Massimo Franceschetti, \IEEEmembership{Fellow, IEEE}
	}

}

\begin{document}

\maketitle

\begin{abstract}
	This paper investigates the limits to which a passive \gls{ris} can reshape a point-to-point \gls{mimo} channel in terms of singular values and their functions (e.g., achievable rate and harvestable power) for improved wireless performance.
	We depart from the \gls{d} scattering model and adopt a \gls{bd} model that exploits element-wise connections for passive signal amplitude and phase manipulation.
	Specifically, analytical tight bounds are derived under typical \gls{ris} deployment scenarios to unveil the channel shaping potentials of \gls{bd}-\gls{ris} regarding communication \gls{dof}, singular value spread, power gain, and capacity.
	An efficient numerical method is then proposed to optimize \gls{bd}-\gls{ris} for any locally Lipschitz function of channel singular values, and showcased to characterize the achievable singular value region.
	As a side product, we tackle \gls{bd}-\gls{ris}-aided \gls{mimo} rate maximization problem by a local-optimal \gls{ao} approach and a low-complexity shaping approach.
	Results show that \gls{bd}-\gls{ris} significantly improves the dynamic range of channel singular values and the tradeoff in manipulating them, thus offering enhanced data rate, harvestable power, and physical-layer security.
	These advantages become more pronounced when the number of \gls{ris} elements, group size, or \gls{mimo} dimensions increase.
	Of particular interest, \gls{bd}-\gls{ris} is shown to activate multi-stream transmission and achieve the asymptotic \gls{dof} at much lower transmit power than \gls{d}-\gls{ris} thanks to its proficiency in channel shaping.
\end{abstract}


\glsresetall

\begin{section}{Introduction}
	Today we are witnessing a paradigm shift from connectivity to intelligence, where the wireless environment is no longer a chaotic medium but a conscious agent that can serve on demand.
	This is empowered by recent developments in \gls{ris}, a programmable surface that recycles and redistributes ambient electromagnetic waves for improved wireless performance.
	A typical \gls{ris} consists of numerous low-power sub-wavelength scattering elements, whose response can be engineered in real-time to manipulate the amplitude, phase, frequency, and polarization of the scattered waves \cite{Basar2019}.
	It enables full-duplex operation while featuring higher flexibility than reflectarrays, lower noise than relays, and greater scalability than multi-antenna transceivers.
	One popular \gls{ris} research topic is \emph{joint beamforming} design with transceivers for a specific performance measure, which has attracted significant interests in wireless communication \cite{Wu2019,Guo2020,Liu2022}, sensing \cite{He2022,Luo2022,Hua2023}, and power transfer \cite{Wu2020a,Feng2022,Zhao2022}.
	Although \gls{ris}-induced propagation paths suffers attenuation from double fading, passive beamforming at \gls{ris} offers better asymptotic behaviors than active beamforming at transceivers (e.g., second-order array gain and fourth-order harvested power \cite{Zhao2022}).
	Another \gls{ris} application is \emph{information modulation} by periodically switching its reflection pattern within the channel coherence time.
	This creates a free-ride message stream with dual benefits -- integrating with legacy transmitter for enhanced channel capacity \cite{Karasik2020,Ye2022} or serving as individual source for low-power uplink communication \cite{Liang2020,Zhao2024}.
	Different from above, \emph{channel shaping} exploits \gls{ris} as a stand-alone device to modify the inherent properties of the wireless environment, for example, compensate for the Doppler effect \cite{Basar2021}, flatten frequency-selective channels \cite{Arslan2022}, improve the channel rank \cite{Ozdogan2020a}, and introduce time diversity for multiple access schemes \cite{Yang2019,Chen2023}.
	This helps decouple joint beamforming problems into a channel shaping stage and a transceiver design stage, offering a modular and versatile solution for diverse wireless applications.
	At a specific time-frequency resource block, channel shaping metrics can be classified into the two categories below.
	\begin{itemize}
		\item \emph{Singular value:} The impact of \gls{ris} has been studied in terms of channel minimum singular value \cite{ElMossallamy2021}, effective rank \cite{Meng2023}, condition number \cite{Huang2023}, and \gls{dof} \cite{Chae2023}. Those are closely related to explicit performance measures but sensitive to minor perturbations of the channel matrix;
		\item \emph{Power:} The impact of \gls{ris} has been studied in terms of channel power gain \cite{Wu2019,Shen2020a,Nerini2023,Nerini2024,Santamaria2023} in point-to-point channels and leakage interference \cite{Santamaria2023a} in interference channels. Those second-order metrics are less informative in the \gls{mimo} context but easier to analyze and optimize.
	\end{itemize}

	Although above works offered inspiring glimpses into the channel shaping potential of passive \gls{ris}, they neither provided in-depth theoretical analysis nor characterized the achievable singular value region.
	Most works \cite{ElMossallamy2021,Meng2023,Huang2023,Chae2023,Wu2019,Santamaria2023a} have also been confined to the conventional \gls{d} architecture where each \gls{ris} element is connected to a dedicated impedance and functions independently of the others, namely, the wave impinging on one element is entirely scattered by itself.
	This architecture is modeled by a diagonal scattering matrix with unit-magnitude diagonal entries that ideally applies a phase shift to the incident signal.
	The idea was soon extended to \gls{bd}-\gls{ris} with group-connected architecture that connects elements within the same group via passive reconfigurable circuit components \cite{Shen2020a} that can be symmetric (e.g., capacitors and inductors) or asymmetric (e.g, ring hybrids and branch-line hybrids \cite{Ahn2006}).
	As such, the wave impinging on one element is able to propagate within the circuit and depart partially from the others of the same group.
	It can thus manipulate both amplitude and phase of the scattered wave while remaining globally passive.
	The main manufacturing complexity of \gls{bd}-\gls{ris} lies in the design and implementation of the circuit network.
	Fortunately, novel topologies such as tree- and forest-connections have been proposed to reduce the number of components for a flexible cost-performance tradeoff \cite{Nerini2024}.
	Other practical challenges such as channel estimation \cite{Li2024}, mutual coupling \cite{Li2023f}, wideband modelling \cite{Li2024a}, multi-sector coverage \cite{Li2023c}, and hardware implementation \cite{Tapie25e} have also been studied in recent literature.
	\gls{bd}-\gls{ris} has been proved to achieve higher spectral efficiency than \gls{d}-\gls{ris} and higher energy efficiency than active \gls{ris} and \gls{af} relay \cite{Santamaria2023,Fang2023,Zhou2023}.
	However, the interplay between \gls{bd}-\gls{ris} and \gls{mimo} is still at infancy stage and the potential remains largely unexplored.
	For example, the rate maximization problem \cite{Bartoli2023} has only been tacked in the special case where the direct channel is negligible and the \gls{bd}-\gls{ris} is fully-connected.
	Under those conditions, the mathematical modeling of \gls{bd}-\gls{ris} coincides with that of \gls{af} relay of unit power, although the operation mechanism and noise characteristics are clearly distinct.

	When it comes to signal processing, existing works have mainly invoked the {quasi-Newton} method \cite{Shen2020a}, the \gls{pdd} method \cite{Zhou2023}, and the generic (i.e., non-geodesic) \gls{rcg} method \cite{Li2023c} for the optimization of \gls{bd}-\gls{ris}.
	The first solves an unconstrained problem and projects the solution back to the feasible domain without optimality guarantee.
	The second alternates between the primal variables in the inner layer and the penalty coefficient in the outer layer.
	It is often used to tackle coupled constraints (e.g., \gls{sinr} thresholds under active and passive beamforming) and can be computationally expensive (e.g., $\mathcal{O}(N_{\mathrm{s}}^2)$ for \gls{d}-\gls{ris} and $\mathcal{O}(N_{\mathrm{s}}^4)$ for fully-connected \gls{bd}-\gls{ris}) \cite[Table~I]{Zhou2023}.
	The third applies the conjugate gradient method on generic Riemannian manifolds.
	Each iteration consists of an addition on the tangent space and a retraction to the feasible domain, which constitutes a zigzag path departing from and returning to the manifold.
	However, none of them effectively exploits the special structure of \gls{bd}-\gls{ris} for accelerated convergence.

	This paper is motivated by a fundamental question:
	\emph{What is the channel shaping capability, in terms of singular values and their functions, of a passive \gls{ris} in \gls{mimo} channels?}
	Unlike existing works that focus on specific performance metrics or deployment scenarios, we aim for an understanding of the theoretical shaping limits (via analysis) and the achievable shaping results (via optimization) that are broadly applicable across diverse wireless applications.
	We believe a comprehensive shaping answer can serve as a theoretical support/reference for the vast number of \gls{ris} research papers and real-world applications.
	Without a framework that identifies the fundamental shaping limits of \gls{ris}, the design of truly optimal and efficient architectures will remain elusive.
	The contributions of this paper are summarized below.

	First, we pioneer \gls{bd}-\gls{ris} study in general \gls{mimo} channels and interpret its shaping ability as branch matching and mode alignment.
	Branch matching refers to pairing and combining the branches (i.e., entries) of backward and forward channels corresponding to each group of the \gls{bd}-\gls{ris}.
	Mode alignment refers to aligning and ordering the modes (i.e., singular vectors) of the \gls{ris}-induced channels with those of the direct channel.
	The former arises uniquely from the off-diagonal entries of the \gls{bd}-\gls{ris} scattering matrix while the latter is enabled by its block-unitary transformation.

	Second, we provide an analytical answer to the shaping question under typical channel conditions.
	It is shown that \gls{bd}-\gls{ris} may achieve a larger or smaller communication \gls{dof} than \gls{d}-\gls{ris}.
	When the backward or forward channel is rank-deficient, we derive asymptotic bounds of individual singular values applying to \gls{d}- and \gls{bd}-\gls{ris}.
	When the direct channel is negligible, we recast the shaping question for fully-connected \gls{bd}-\gls{ris} as a well-studied linear algebra question and provide tight bounds (with closed-form scattering matrices) on channel singular values, power gain, and capacity.
	These results help us understand the fundamental limits of channel shaping and serve as a foundation for application-specific designs.


	Third, we provide a numerical \gls{bd}-\gls{ris} design framework for any locally Lipschitz function of channel singular values via a geodesic \gls{rcg} method.
	It compares favorably to generic manifold optimizers in that the updates are performed along the geodesics, namely the shortest paths on the manifold, for accelerated convergence.
	The method is then invoked for a Pareto problem to reveal the achievable channel singular value region, which generalizes most relevant metrics and provides an intuitive shaping benchmark.

	Fourth, we tackle \gls{bd}-\gls{ris}-aided \gls{mimo} rate maximization problem by a local-optimal \gls{ao} approach and a low-complexity shaping approach.
	The former iteratively updates the passive beamforming via geodesic \gls{rcg} and the active beamforming by eigenmode transmission, until convergence.
	The latter simply shapes the channel for maximum power gain then performs legacy transmission.

	Fifth, we validate the analytical bounds and the numerical methods by extensive simulation.
	It is concluded that:
	\begin{itemize}
		\item \gls{bd}-\gls{ris} can widen the dynamic range of channel singular values for enhanced rate, power, and physical-layer security;
		\item The shaping benefits of \gls{bd}-\gls{ris} over \gls{d}-\gls{ris} scale with the number of elements, group size, and \gls{mimo} dimensions;
		\item \gls{bd}-\gls{ris} can activate multi-stream transmission and achieve the asymptotic \gls{dof} at lower transmit power than \gls{d}-\gls{ris};
		\item The rate gap between the \gls{ao} and shaping approaches diminishes as the \gls{ris} evolves from \gls{d} to fully-connected \gls{bd};
		\item The proposed geodesic \gls{rcg} method is efficient and the optimization cost of practically-sized \gls{bd}-\gls{ris} remains low;
		\item The solutions are robust to channel estimation errors and extendable to symmetric constraint with minimal degradation.
	\end{itemize}

	\emph{Notation:}
	Italic, bold lower-case, and bold upper-case letters indicate scalars, vectors and matrices, respectively.
	$\jmath$ denotes the imaginary unit.
	$\mathbb{R}$ and $\mathbb{C}$ denote the set of real and complex numbers, respectively.
	$\mathbb{H}^{n \times n}$, $\mathbb{H}_+^{n \times n}$, $\mathbb{U}^{n \times n}$, and $\mathbb{P}^{n \times n}$ denote the set of $n \times n$ Hermitian, positive semi-definite, unitary, and permutation matrices, respectively.
	$\mathbf{0}$ and $\mathbf{I}$ are the zero and identity matrices with appropriate size, respectively.
	$\Re\{\cdot\}$ takes the real part of a complex number.
	$\mathbb{E}\{\cdot\}$ is the expectation operator.
	$\conv\{\cdot\}$ returns the convex hull of arguments.
	$\tr(\cdot)$ and $\det(\cdot)$ evaluate the trace and determinant of a square matrix, respectively.
	$\diag(\cdot)$ constructs a square matrix with arguments on the main (block) diagonal and zeros elsewhere.
	$\sv(\cdot)$, $\ran(\cdot)$, and $\ker(\cdot)$ evaluate the singular values, range, and kernel of a matrix, respectively.
	$\mathrm{vec}(\cdot)$ stacks the columns of a matrix as a vector.
	$\lvert \cdot \rvert$, $\lVert \cdot \rVert$, and $\lVert \cdot \rVert _\mathrm{F}$ denote the absolute value, Euclidean norm, and Frobenius norm, respectively.
	$\sigma_n(\cdot)$ and $\lambda_n(\cdot)$ are the $n$-th largest singular value and eigenvalue, respectively.
	$(\cdot)^*$, $(\cdot)^\mathsf{T}$, $(\cdot)^\mathsf{H}$, $(\cdot)^\dagger$, $(\cdot)^{\star}$ denote the conjugate, transpose, conjugate transpose (Hermitian), Moore-Penrose inverse, and stationary point, respectively.
	$[N]$ is a shortcut for $\{1,2,\ldots,N\}$.
	$(\cdot)_{[x:y]}$ is a shortcut for $(\cdot)_x,(\cdot)_{x+1},\ldots,(\cdot)_y$.
	$\odot$ denotes the Hadamard product.
	$\mathcal{O}(\cdot)$ is the big-O notation.
	$\mathcal{N}_{\mathbb{C}}(\mathbf{0}, \mathbf{\Sigma})$ is the multivariate \gls{cscg} distribution with mean $\mathbf{0}$ and covariance $\mathbf{\Sigma}$.
	$\sim$ means ``distributed as''.
\end{section}

\begin{section}{System Model}
	We model the \gls{bd}-\gls{ris} as an $N_\mathrm{S}$-port network that divides into $G$ individual groups, where group $g \in [G]$ contains $N_g$ scattering elements interconnected by real-time reconfigurable components \cite{Shen2020a} satisfying $N_\mathrm{S} = \sum_{g=1}^G N_g$.
	For the ease of analysis, we assume no mutual coupling and equal group size $N_g = L \triangleq N_\mathrm{S} / G, \ \forall g$.
	The overall scattering matrix of an asymmetric \gls{bd}-\gls{ris} is block-diagonal
	\begin{equation}
		\label{eq:ris}
		\mathbf{\Theta} = \diag(\mathbf{\Theta}_1,\ldots,\mathbf{\Theta}_G),
	\end{equation}
	where $\mathbf{\Theta}_g \in \mathbb{U}^{L \times L}$ is the $g$-th unitary block modeling the response of group $g$.
	\gls{d}-\gls{ris} can be seen an extreme case of \eqref{eq:ris} with group size $L=1$.
	Some viable architectures of \gls{bd}-\gls{ris} are illustrated in \cite[Fig.~3]{Shen2020a}, \cite[Fig.~5]{Li2023c}, \cite[Fig.~2]{Nerini2024} where the circuit topology have been modeled in the scattering matrix.

	Consider a \gls{bd}-\gls{ris}-aided \gls{mimo} point-to-point channel with $N_\mathrm{T}$ and $N_\mathrm{R}$ transmit and receive antennas, respectively, and $N_\mathrm{S}$ scattering elements at the \gls{bd}-\gls{ris}.
	This configuration is denoted as $N_\mathrm{T} \times N_\mathrm{S} \times N_\mathrm{R}$ throughout this paper.
	Let $\mathbf{H}_\mathrm{D} \in \mathbb{C}^{N_\mathrm{R} \times N_\mathrm{T}}$, $\mathbf{H}_\mathrm{B} \in \mathbb{C}^{N_\mathrm{R} \times N_\mathrm{S}}$, $\mathbf{H}_\mathrm{F} \in \mathbb{C}^{N_\mathrm{S} \times N_\mathrm{T}}$ denote the direct (i.e., transmitter-receiver), backward (i.e., \gls{ris}-receiver), and forward (i.e., transmitter-\gls{ris}) channels, respectively.
	The equivalent channel is
	\begin{equation}
		\label{eq:channel}
		\mathbf{H} = \mathbf{H}_\mathrm{D} + \mathbf{H}_\mathrm{B} \mathbf{\Theta} \mathbf{H}_\mathrm{F} = \mathbf{H}_\mathrm{D} + \sum_g \mathbf{H}_{\mathrm{B},g} \mathbf{\Theta}_g \mathbf{H}_{\mathrm{F},g},
	\end{equation}
	where $\mathbf{H}_{\mathrm{B},g} \in \mathbb{C}^{N_\mathrm{R} \times L}$ and $\mathbf{H}_{\mathrm{F},g} \in \mathbb{C}^{L \times N_\mathrm{T}}$ are the backward and forward channels associated with group $g$, corresponding to the $(g{-}1)L{+}1$ to $gL$ columns of $\mathbf{H}_\mathrm{B}$ and rows of $\mathbf{H}_\mathrm{F}$, respectively.
	Since unitary matrices constitute an algebraic group with respect to multiplication, we can decompose the scattering matrix of group $g$ as
	\begin{equation}
		\label{eq:ris_decompose_group}
		\mathbf{\Theta}_g = \mathbf{L}_g \mathbf{X}_g \mathbf{R}_g^\mathsf{H},
	\end{equation}
	where $\mathbf{L}_g, \mathbf{R}_g \in \mathbb{U}^{L \times L}$ are unitary matrices and $\mathbf{X}_g \in \mathbb{P}^{n \times n}$ is a permutation matrix.
	Let $\mathbf{H}_g \triangleq \mathbf{H}_{\mathrm{B},g} \mathbf{\Theta}_g \mathbf{H}_{\mathrm{F},g}$ be the indirect channel via group $g$ and $\mathbf{H}_{\mathrm{B/F},g} = \mathbf{U}_{\mathrm{B/F},g} \mathbf{\Sigma}_{\mathrm{B/F},g} \mathbf{V}_{\mathrm{B/F},g}^\mathsf{H}$ be the \gls{svd} of the backward and forward channels, respectively.
	The equivalent channel is
	\begin{equation}
		\label{eq:channel_svd}
		\mathbf{H} = \overbrace{\mathbf{H}_\mathrm{D} + \sum_g \mathbf{U}_{\mathrm{B},g} \mathbf{\Sigma}_{\mathrm{B},g} \underbrace{\mathbf{V}_{\mathrm{B},g}^\mathsf{H} \mathbf{L}_g \mathbf{X}_g \mathbf{R}_g^\mathsf{H} \mathbf{U}_{\mathrm{F},g}}_\text{backward-forward} \mathbf{\Sigma}_{\mathrm{F},g} \mathbf{V}_{\mathrm{F},g}^\mathsf{H}}^\text{direct-indirect}.
	\end{equation}

	\begin{remark}
		In \eqref{eq:channel_svd}, the \gls{bd}-\gls{ris} performs a blockwise unitary transformation to combine the backward-forward (intra-group, multiplicative) channels and direct-indirect (inter-group, additive) channels.
		These two attributes are refined respectively as:
		\begin{itemize}
			\item \emph{Branch matching:} To pair and combine the branches (i.e., entries) of $\mathbf{H}_{\mathrm{B},g}$ and $\mathbf{H}_{\mathrm{F},g}$ through $\mathbf{\Theta}_g$;
			\item \emph{Mode alignment:} To align and order the modes (i.e., singular vectors) of $\{\mathbf{H}_g\}_{g \in [G]}$ with those of $\mathbf{H}_\mathrm{D}$ through $\mathbf{\Theta}$.
		\end{itemize}
	\end{remark}

	\begin{example}[\gls{siso} channel gain maximization]
		\label{eg:siso}
		Denote the direct, backward, forward channels as $h_\mathrm{D}$, $\mathbf{h}_\mathrm{B} \in \mathbb{C}^{N_\mathrm{S} \times 1}$, and $\mathbf{h}_\mathrm{F}^\mathsf{H} \in \mathbb{C}^{1 \times N_\mathrm{S}}$, respectively.
		In this case, mode alignment boils down to phase matching and the optimal \gls{bd}-\gls{ris} structure is
		\begin{equation}
			\mathbf{\Theta}_{\textnormal{P-max},g}^\textnormal{SISO} = \frac{h_\mathrm{D}}{\lvert h_\mathrm{D} \rvert} \mathbf{V}_{\mathrm{B},g} \mathbf{U}_{\mathrm{F},g}^\mathsf{H}, \quad \forall g,
		\end{equation}
		where $\mathbf{V}_{\mathrm{B},g} = \bigl[\mathbf{h}_{\mathrm{B},g}/\lVert \mathbf{h}_{\mathrm{B},g} \rVert, \mathbf{N}_{\mathrm{B},g}\bigr] \in \mathbb{U}^{L \times L}$, $\mathbf{U}_{\mathrm{F},g} = \bigl[\mathbf{h}_{\mathrm{F},g}/\lVert \mathbf{h}_{\mathrm{F},g} \rVert, \mathbf{N}_{\mathrm{F},g}\bigr] \in \mathbb{U}^{L \times L}$, and $\mathbf{N}_{\mathrm{B/F},g} \in \mathbb{C}^{L \times (L-1)}$ are the orthonormal bases of kernels of $\mathbf{h}_{\mathrm{B/F},g}$.
		Evidently, any group size $L$ (including \gls{d}-\gls{ris} $L=1$ with empty kernels) suffices for perfect phase matching.
		The maximum channel gain still depends on $L$
		\begin{equation}
			\ \lvert h \rvert = \lvert h_\mathrm{D} \rvert + \sum_g \sum_l \lvert h_{\mathrm{B},g,\pi_{\mathrm{B},g}(l)} \rvert \lvert h_{\mathrm{F},g,\pi_{\mathrm{F},g}(l)} \rvert,
		\end{equation}
		where $h_{\mathrm{B/F},g,l}$ are the $l$-th entries of $\mathbf{h}_{\mathrm{B/F},g}$, and $\pi_{\mathrm{B/F},g}$ are permutations of $[L]$ sorting their magnitude in similar orders.
		That is, the maximum \gls{siso} channel gain is attained when each \gls{bd}-\gls{ris} group, apart from phase shifting, matches the $l$-th strongest backward and forward channel branches.
		Increasing $L$ improves the branch matching flexibility and boosts the channel gain.
	\end{example}

	Example~\ref{eg:siso} clarifies the difference between branch matching and mode alignment as well as their impacts on channel shaping.
	When it comes to \gls{mimo}, the advantage of \gls{bd}-\gls{ris} in branch matching improves since the number of available branches is proportional to $N_\mathrm{T}$ and $N_\mathrm{R}$. On the other hand, the limitation of \gls{d}-\gls{ris} in mode alignment intensifies since each element can only apply a scalar phase shift to the indirect channel of $\min(N_\mathrm{T}, N_\mathrm{S}, N_\mathrm{R})$ modes.
\end{section}

\begin{section}{Channel Shaping}
	In this section, we first provide an example demonstrating the \gls{mimo} channel shaping advantages of \gls{bd}-\gls{ris} over \gls{d}-\gls{ris}, then derive some analytical bounds on singular values, power gain, and capacity under specific channel conditions.
	Finally, we propose a numerical method to optimize the \gls{bd}-\gls{ris} for a broad class of singular value functions.

	\begin{example}[$2 \times 2 \times 2$ shaping]
		\label{eg:shaping_potential}
		Here \gls{d}-\gls{ris} and fully-connected \gls{bd}-\gls{ris} can be modeled by 2 and 4 independent angular parameters, respectively:
		\begin{equation*}
			\mathbf{\Theta}_\mathrm{D} = \diag(e^{\jmath \theta_1}, e^{\jmath \theta_2}), \quad
			\mathbf{\Theta}_\mathrm{BD} = e^{\jmath \phi} \begin{bmatrix}
				e^{\jmath \alpha} \cos \psi  & e^{\jmath \beta} \sin \psi   \\
				-e^{-\jmath \beta} \sin \psi & e^{-\jmath \alpha} \cos \psi
			\end{bmatrix},
		\end{equation*}
		We consider a special case where the \gls{bd}-\gls{ris} is symmetric (i.e., $\beta = \pi / 2$) and the direct channel is negligible such that $\phi$ has no impact on $\sv(\mathbf{H})$, since $\sv(e^{\jmath \phi} \mathbf{A}) = \sv(\mathbf{A})$.
		The singular value shaping capabilities of $\mathbf{\Theta}_\mathrm{D}$ and $\mathbf{\Theta}_\mathrm{BD}$ can thus be visualized over 2 tunable parameters.
		\begin{figure}
			\centering
			\includegraphics[width=0.7\columnwidth]{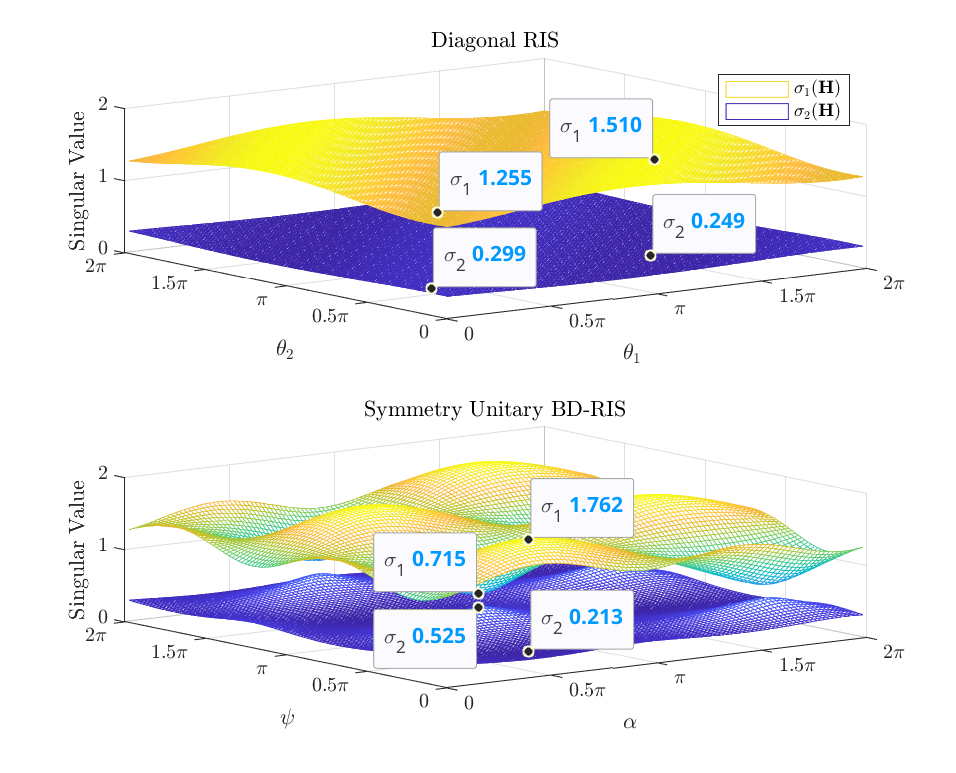}
			\caption{$2 \times 2 \times 2$ singular value shaping by \gls{d}-\gls{ris} and symmetric fully-connected \gls{bd}-\gls{ris} when the direct channel is negligible.
				The maximum and minimum of both singular values are marked explicitly on the plot.}
			\label{fg:shaping_potential}
		\end{figure}
		With an exhaustive grid search over $(\theta_1, \theta_2)$ and $(\alpha, \psi)$, Fig.~\ref{fg:shaping_potential} shows the achievable singular values of a specific channel instance
		$
			\mathbf{H}_\mathrm{B} =
			\begin{bsmallmatrix}
				-0.2059 + 0.5914 \jmath & -0.0909 + 0.5861 \jmath \\
				0.4131 + 0.2651 \jmath  & -0.1960 + 0.4650 \jmath
			\end{bsmallmatrix}, \
			\mathbf{H}_\mathrm{F} =
			\begin{bsmallmatrix}
				-0.6362 + 0.1332 \jmath & -0.1572 + 1.5538 \jmath \\
				0.0196 + 0.4011 \jmath  & -0.3170 - 0.2303 \jmath
			\end{bsmallmatrix}.
		$
		In this example, both singular values can be manipulated up to\footnote{The percentage for manipulating $\sigma_n(\mathbf{H})$ is calculated by $\eta_n^+ = \frac{\max \sigma_n(\mathbf{H}) - \mathrm{avg} \sigma_n(\mathbf{H})}{\mathrm{avg} \sigma_n(\mathbf{H})} \times 100\%$ and  $\eta_n^- = \frac{\min \sigma_n(\mathbf{H}) - \mathrm{avg} \sigma_n(\mathbf{H})}{\mathrm{avg} \sigma_n(\mathbf{H})} \times 100\%$.} $\pm 9\%$ by \gls{d}-\gls{ris} (using 2 reconfigurable components) and $\pm 42\%$ by symmetric fully-connected \gls{bd}-\gls{ris} (using 3 reconfigurable components).
	\end{example}

	Example~\ref{eg:shaping_potential} suggests that the physical interconnection of \gls{ris} elements, even if using symmetric circuit components, can create a ``cooperation effect'' that significantly enhances the dynamic range of channel singular values.
	This motivates the analytical and numerical shaping studies in Sections~\ref{sc:shaping_analytical} and \ref{sc:shaping_numerical}.

	\begin{subsection}{Analytical Shaping Bounds}
		\label{sc:shaping_analytical}
		\begin{definition}[\gls{dof}]
			\gls{dof} refers to the maximum number of streams that can be transmitted in parallel over a \gls{mimo} channel in the asymptotic high-\gls{snr} regime
			\begin{equation}
				\textnormal{DoF}(\mathbf{H}) = \lim_{\rho \to \infty} \frac{\log \det(\mathbf{I} + \rho \mathbf{H} \mathbf{H}^\mathsf{H})}{\log \rho},
			\end{equation}
			where $\rho$ is the \gls{snr}.
		\end{definition}

		\begin{definition}[Negligible direct channel]
			A direct channel is considered negligible when its contribution to the received signal is substantially weaker than that of the \gls{ris}-induced indirect channels. Mathematically, this can be defined as
			\begin{equation*}
				\frac{\lVert \mathbf{H}_\mathrm{D} \rVert _\mathrm{F}^2}{\lVert \sum_g \mathbf{H}_{\mathrm{B},g} \mathbf{\Theta}_g \mathbf{H}_{\mathrm{F},g} \rVert _\mathrm{F}^2} < \epsilon,
			\end{equation*}
			where $\epsilon$ is a small positive threshold.
			This can result from a very large number of \gls{ris} elements or physical obstacles in the propagation path.

		\end{definition}

		The main results of this subsection are summarized in the following Propositions and Corollaries.
		\begin{proposition}[\gls{dof}]
			\label{pp:dof}
			\gls{bd}-\gls{ris} may achieve a larger or smaller \gls{mimo} \gls{dof} than \gls{d}-\gls{ris}.
		\end{proposition}
		\begin{proof}
			Please refer to Appendix~\ref{ap:dof}.
		\end{proof}


		\begin{example}[\gls{dof} of $4 \times 4 \times 4$ shaping]
			\label{eg:dof}
			Consider a $4 \times 4 \times 4$ shaping with $\mathbf{H}_\mathrm{D} = \mathbf{0}$, $\mathbf{H}_\mathrm{B} =
				\begin{bsmallmatrix}
					1 & 1 & 0 & 0 \\
					0 & 0 & 0 & 0 \\
					0 & 0 & 1 & 0 \\
					0 & 0 & 0 & 0
				\end{bsmallmatrix}$, and $\mathbf{H}_\mathrm{F} = \diag(1, 1, 0, 0).$
			Evidently, any \gls{d}-\gls{ris} $\mathbf{\Theta}_\mathrm{D} = \diag(e^{\jmath \theta_1}, e^{\jmath \theta_2}, e^{\jmath \theta_3}, e^{\jmath \theta_4})$ results in
			\begin{equation*}
				\mathbf{H} =
				\begin{bsmallmatrix}
					\scriptscriptstyle e^{\jmath \theta_1} & \scriptscriptstyle e^{\jmath \theta_2} & 0 & 0 \\
					0 & 0 & 0 & 0 \\
					0 & 0 & 0 & 0 \\
					0 & 0 & 0 & 0
				\end{bsmallmatrix}
			\end{equation*}
			with 1 \gls{dof}.
			On the other hand, a fully-connected \gls{bd}-\gls{ris} can perfectly align or misalign the kernels of $\mathbf{H}_\mathrm{B}$ and $\mathbf{H}_\mathrm{F}$ using the closed-form solutions \eqref{eq:ris_dof_max} or \eqref{eq:ris_dof_min} in Appendix~\ref{ap:dof}. That is,
			$\mathbf{\Theta}_{\textnormal{DoF-max}}^{\textnormal{MIMO}} =
				\begin{bsmallmatrix}
					0 & \scriptscriptstyle{\frac{1}{\sqrt{2}}} & 0 & \scriptscriptstyle{-\frac{1}{\sqrt{2}}} \\
					0 & \scriptscriptstyle{\frac{1}{\sqrt{2}}} & 0 & \scriptscriptstyle{\frac{1}{\sqrt{2}}} \\
					-1 & 0 & 0 & 0 \\
					0 & 0 & 1 & 0
				\end{bsmallmatrix}$
			and
			$\mathbf{\Theta}_{\textnormal{DoF-min}}^{\textnormal{MIMO}} =
				\begin{bsmallmatrix}
					\scriptscriptstyle{-\frac{1}{\sqrt{2}}} & 0 & \scriptscriptstyle{\frac{1}{\sqrt{2}}} & 0\\
					\scriptscriptstyle{\frac{1}{\sqrt{2}}} & 0 & \scriptscriptstyle{\frac{1}{\sqrt{2}}} & 0 \\
					0 & 0 & 0 & 1 \\
					0 & -1 & 0 & 0
				\end{bsmallmatrix}$,
			which correspond to
			\begin{equation*}
				\mathbf{H} =
				\begin{bsmallmatrix}
					0 & \sqrt{2} & 0 & 0 \\
					0 & 0 & 0 & 0 \\
					-1 & 0 & 0 & 0 \\
					0 & 0 & 0 & 0
				\end{bsmallmatrix}, \quad
				\mathbf{H} = \mathbf{0},
			\end{equation*}
			and a \gls{dof} of 2 and 0, respectively.
		\end{example}

		Proposition~\ref{pp:dof} and Example~\ref{eg:dof} suggest that we can expect more parallel data streams or less crosstalk when shaping the channel with \gls{bd}-\gls{ris}.
		Increasing the \gls{dof} can improve the asymptotic rate performance for point-to-point transmission.
		Conversely, reducing the \gls{dof} can help orthogonalize channels in multi-user networks for the interest of interference alignment and physical layer security.
		Next, we progress to quantify the limits of singular value redistribution in rank-deficient channels.

		\begin{proposition}[Rank-deficient channel]
			\label{pp:rd}
			If the minimum rank of backward and forward channels is $k$ ($k \le N \triangleq \min(N_\mathrm{T}, N_\mathrm{R})$),
			then for \gls{d}-\gls{ris} or \gls{bd}-\gls{ris} of arbitrary number of elements, the $n$-th singular value of the equivalent channel is bounded above and below respectively by
			\begin{subequations}
				\label{iq:sv_rd}
				\begin{align}
					\sigma_n(\mathbf{H}) & \le \sigma_{n-k}(\mathbf{T}), &  & \text{if } n > k, \label{iq:sv_rd_max}         \\
					\sigma_n(\mathbf{H}) & \ge \sigma_n(\mathbf{T}),     &  & \text{if } n < N - k + 1, \label{iq:sv_rd_min}
				\end{align}
			\end{subequations}
			where $\mathbf{T}$ is any auxiliary matrix satisfying
			\begin{equation}
				\label{eq:auxiliary_rd}
				\mathbf{T} \mathbf{T}^\mathsf{H} =
				\begin{cases}
					\mathbf{H}_\mathrm{D} (\mathbf{I} - \mathbf{V}_\mathrm{F} \mathbf{V}_\mathrm{F}^\mathsf{H}) \mathbf{H}_\mathrm{D}^\mathsf{H}, & \text{if } \rank(\mathbf{H}_\mathrm{F}) = k, \\
					\mathbf{H}_\mathrm{D}^\mathsf{H} (\mathbf{I} - \mathbf{U}_\mathrm{B} \mathbf{U}_\mathrm{B}^\mathsf{H}) \mathbf{H}_\mathrm{D}, & \text{if } \rank(\mathbf{H}_\mathrm{B}) = k,
				\end{cases}
			\end{equation}
			and $\mathbf{V}_\mathrm{F}$ and $\mathbf{U}_\mathrm{B}$ are any right and left singular matrices of $\mathbf{H}_\mathrm{F}$ and $\mathbf{H}_\mathrm{B}$, respectively.
		\end{proposition}
		\begin{proof}
			Please refer to Appendix~\ref{ap:rank_deficient}.
		\end{proof}

		Inequality \eqref{iq:sv_rd_max} states that
		if $\mathbf{H}_\mathrm{B}$ and $\mathbf{H}_\mathrm{F}$ are at least rank $k$, then
		using a \gls{d}-\gls{ris} or \gls{bd}-\gls{ris} of \emph{sufficiently large} $N_\mathrm{S}$,
		the $n$-th singular value of $\mathbf{H}$ can be enlarged to the $(n-k)$-th singular value of $\mathbf{T}$, or suppressed to the $n$-th singular value of $\mathbf{T}$.
		Moreover, the first $k$ channel singular values are unbounded above\footnote{The energy conservation law $\sum_{n=1}^N \sigma_n^2(\mathbf{H}) \le 1$ still has to be respected in all cases.} while the last $k$ channel singular values can be suppressed to zero.
		A special case of \gls{los} channel is presented below\footnote{A similar eigenvalue result has been derived for \gls{d}-\gls{ris} only \cite{Semmler2023}.}.

		\begin{corollary}[\gls{los} channel]
			\label{co:los}
			If at least one of backward and forward channels is \gls{los}, then a \gls{d}-\gls{ris} or \gls{bd}-\gls{ris} can at most enlarge the $n$-th ($n \ge 2$) channel singular value to the $(n-1)$-th singular value of $\mathbf{T}$, or suppress the $n$-th channel singular value to the $n$-th singular value of $\mathbf{T}$.
			That is,
			\begin{equation}
				\label{iq:sv_los}
				\sigma_1(\mathbf{H}) \ge \sigma_1(\mathbf{T}) \ge {\sigma_2(\mathbf{H})} \ge \ldots \ge \sigma_{N-1}(\mathbf{T}) \ge {\sigma_N(\mathbf{H})} \ge \sigma_N(\mathbf{T}).
			\end{equation}
		\end{corollary}

		\begin{proof}
			This is a direct result of \eqref{iq:sv_rd} with $k = 1$.
		\end{proof}

		We emphasize that Proposition~\ref{pp:rd} and Corollary~\ref{co:los} apply to both \gls{d}- and \gls{bd}-\gls{ris} configurations regardless of the status of the direct channel.
		Out of $2N$ bounds in \eqref{iq:sv_rd} or \eqref{iq:sv_los}, $N$ of them can be \emph{simultaneously} tight as $N_\mathrm{S} \to \infty$, namely when the direct channel becomes negligible.
		For a finite $N_\mathrm{S}$, the \gls{ris} may prioritize a subset of those by aligning the corresponding modes.
		We will show by simulation that \gls{bd}-\gls{ris} outperforms \gls{d}-\gls{ris} on this purpose.
		Proposition~\ref{pp:rd} complements the \gls{dof} result in Proposition~\ref{pp:dof} by quantifying the dynamic range of extreme singular values in low-multipath scenarios.
		They reveal a diminishing return of increasing the number of \gls{bd}-\gls{ris} elements and group size in enhancing channel shaping capability. Therefore, the bounds can be used to guide practical \gls{ris} configurations, especially in millimeter-wave and terahertz systems under sparse propagation environment, for a balanced performance-complexity tradeoff.
		Next, we progress to quantify the limits of singular value redistribution when the direct channel is negligible.

		\begin{proposition}[Negligible direct channel]
			\label{pp:nd}
			If the direct channel is negligible, then a fully-connected \gls{bd}-\gls{ris} of arbitrary number of elements can manipulate the channel singular values up to
			\begin{equation}
				\sv(\mathbf{H}) = \sv(\mathbf{BF}),
			\end{equation}
			where $\mathbf{B}$ and $\mathbf{F}$ are any matrices satisfying $\sv(\mathbf{B})=\sv(\mathbf{H}_\mathrm{B})$ and $\sv(\mathbf{F})=\sv(\mathbf{H}_\mathrm{F})$.
		\end{proposition}

		\begin{proof}
			Please refer to Appendix~\ref{ap:nd}.
		\end{proof}

		Proposition~\ref{pp:nd} says that if the direct channel is negligible and the \gls{bd}-\gls{ris} is fully-connected, the only singular value bounds on the equivalent channel are those on the product of unitary-transformed backward and forward channels.
		It is \emph{not necessarily} an asymptotic result and does \emph{not} depend on any relationship between $N_\mathrm{T}$, $N_\mathrm{S}$, and $N_\mathrm{R}$.
		Its importance lies in that our channel shaping question can be recast as a well-studied linear algebra question: \emph{How the singular values of matrix product are bounded by the singular values of its individual factors?}
		The question is partially answered in Corollaries \ref{co:nd_sv_prod_subset} -- \ref{co:nd_sv_indl} over definitions $\bar{N} = \max(N_\mathrm{T},N_\mathrm{S},N_\mathrm{R})$ and $\sigma_n(\mathbf{H})=\sigma_n(\mathbf{H}_\mathrm{F})=\sigma_n(\mathbf{H}_\mathrm{B})=0, \ \forall n \in [\bar{N}] \setminus [N]$.
		This is equivalent to padding zero blocks at the end of $\mathbf{H}, \mathbf{H}_\mathrm{B}, \mathbf{H}_\mathrm{F}$ to make square matrices of dimension $\bar{N}$.
		The results are by no means complete and interested readers are referred to \cite[Chapter 16, 24]{Hogben2013} and \cite[Chapter 3]{Horn1994} for more information.

		\begin{corollary}[Product of subset of singular values]
			\label{co:nd_sv_prod_subset}
			If the direct channel is negligible,
			then the product of subset of singular values of $\mathbf{H}$ is bounded from above by those of $\mathbf{H}_\mathrm{B}$ and $\mathbf{H}_\mathrm{F}$, that is,
			\begin{equation}
				\label{iq:horn}
				\prod_{k \in {K}} \sigma_k(\mathbf{H}) \le \prod_{i \in {I}} \sigma_i(\mathbf{H}_\mathrm{B}) \prod_{j \in {J}} \sigma_j(\mathbf{H}_\mathrm{F}),
			\end{equation}
			for all admissible triples $(I, J, K) \in T_r^{\bar{N}}$ with $r < \bar{N}$, where
			\begin{equation*}
				\begin{split}
					T_r^{\bar{N}} \triangleq \Bigl\{(I, J, K) \in U_r^{\bar{N}} \bigm\vert & \forall p < r, \ \forall (F, G, H) \in T_p^r,                                              \\
					                                                                       & \sum_{f \in F} i_f + \sum_{g \in G} j_g \le \sum_{h \in H} k_h + \frac{p(p+1)}{2} \Bigr\},
				\end{split}
			\end{equation*}
			\begin{equation*}
				U_r^{\bar{N}} \triangleq \Bigl\{(I, J, K) \subseteq [\bar{N}]^3 \bigm\vert \sum_{i \in I} i + \sum_{j \in J} j = \sum_{k \in K} k + \frac{r(r+1)}{2}\Bigr\}.
			\end{equation*}
		\end{corollary}

		\begin{proof}
			Please refer to \cite[Theorem~8]{Fulton2000}.
		\end{proof}


		Corollary~\ref{co:nd_sv_prod_subset} applies to arbitrary number of \gls{ris} elements as inherited from Proposition~\ref{pp:nd}.
		The set \eqref{iq:horn}, also recognized as a variation of Horn's inequality \cite{Bhatia2001}, provides a comprehensive analytical answer to the shaping question -- it can be interpreted as the \emph{outer bounds} of the achievable singular value region of the \gls{bd}-\gls{ris}-aided \gls{mimo} channel.
		An example is given by \eqref{iq:sv_3_bounds} and their visualization in Fig.~\ref{fg:singular_region}.
		Remarkably, the number of inequalities in \eqref{iq:horn} increases exponentially with $N_\mathrm{S}$\footnote{For example, the number of inequalities described by \eqref{iq:horn} grows from 12 to 2062 when $N_\mathrm{S}$ increases from 3 to 7.}.
		At a first glance the results may seem excessive to be useful; but they are given in this form to be general and one can pick any \emph{subset} of them for specific applications.
		Below we showcase how to induce some ready-to-use wireless performance bounds with closed-form \gls{bd}-\gls{ris} solutions from Corollary~\ref{co:nd_sv_prod_subset}.
		The applications mentioned therein are non-exhaustive; we really hope our ingenious readers can discover more results specific to their research.

		\begin{corollary}[Product of some largest or smallest singular values]
			\label{co:nd_sv_prod_tail}
			If the direct channel is negligible,
			then the product of the first (resp. last) $k$ singular values of $\mathbf{H}$ is bounded from above (resp. below) by those of $\mathbf{H}_\mathrm{B}$ and $\mathbf{H}_\mathrm{F}$, that is,
			\begin{subequations}
				\begin{align}
					\prod_{n=1}^k \sigma_n(\mathbf{H})                              & \le \prod_{n=1}^k \sigma_n(\mathbf{H}_\mathrm{B}) \sigma_n(\mathbf{H}_\mathrm{F}), \label{iq:sv_nd_prod_largest}                               \\
					\prod_{n=\bar{N}}^{\mathclap{\bar{N}-k+1}} \sigma_n(\mathbf{H}) & \ge \prod_{n=\bar{N}}^{\mathclap{\bar{N}-k+1}} \sigma_n(\mathbf{H}_\mathrm{B}) \sigma_n(\mathbf{H}_\mathrm{F}). \label{iq:sv_nd_prod_smallest}
				\end{align}
			\end{subequations}
		\end{corollary}

		\begin{proof}
			Please refer to Appendix~\ref{ap:nd_sv_prod_tail}.
		\end{proof}

		Corollary~\ref{co:nd_sv_prod_tail} reveals the shaping limits on the product of some extreme channel singular values.
		The lower bounds \eqref{iq:sv_nd_prod_smallest} coincide at zero when $\bar{N} \ne N$ (i.e., $N_\mathrm{T} = N_\mathrm{S} = N_\mathrm{R}$ being false).
		These bounds can be applied, for instance, as a shortcut to establish the capacity of \gls{bd}-\gls{ris}-aided \gls{mimo} channels at extreme \gls{snr}, as shown in Corollary~\ref{co:nd_capacity_snr_extreme}.
		In the special case $k=1$, we arrive at the upper bound on the largest channel singular value $\sigma_1(\mathbf{H}) {\le} \sigma_1(\mathbf{H}_\mathrm{B}) \sigma_1(\mathbf{H}_\mathrm{F})$.
		This is particularly useful for \gls{mimo} wireless power transfer with \gls{rf} combining where the harvested power depends merely on, and is a quartic function of, the largest channel singular value \cite{Shen2021}.
		A closed-form \gls{bd}-\gls{ris} solution to attain this upper bound can be found below in \eqref{eq:ris_nd_sv_indl_max}.

		\begin{corollary}[Individual singular value]
			\label{co:nd_sv_indl}
			If the direct channel is negligible,
			then the $n$-th channel singular value can be manipulated within the range of
			\begin{equation}
				\label{iq:sv_nd_indl}
				\max_{\mathclap{i+j=n+N_\mathrm{S}}} \ \sigma_i(\mathbf{H}_\mathrm{B}) \sigma_j(\mathbf{H}_\mathrm{F}) \le \sigma_n(\mathbf{H}) \le \min_{\mathclap{i+j=n+1}} \ \sigma_i(\mathbf{H}_\mathrm{B}) \sigma_j(\mathbf{H}_\mathrm{F}),
			\end{equation}
			where $(i, j) \in [N_\mathrm{S}]^2$.
			The upper and lower bounds are attained respectively at
			\begin{subequations}
				\label{eq:ris_nd_sv_indl}
				\begin{align}
					\mathbf{\Theta}_{\textnormal{sv-}n\textnormal{-max}}^\textnormal{MIMO-ND} & = \mathbf{V}_\mathrm{B} \mathbf{P} \mathbf{U}_\mathrm{F}^\mathsf{H}, \label{eq:ris_nd_sv_indl_max} \\
					\mathbf{\Theta}_{\textnormal{sv-}n\textnormal{-min}}^\textnormal{MIMO-ND} & = \mathbf{V}_\mathrm{B} \mathbf{Q} \mathbf{U}_\mathrm{F}^\mathsf{H}, \label{eq:ris_nd_sv_indl_min}
				\end{align}
			\end{subequations}
			where $\mathbf{V}_\mathrm{B}, \mathbf{U}_\mathrm{F} \in \mathbb{U}^{N_\mathrm{S} \times N_\mathrm{S}}$ are any right and left singular matrices of $\mathbf{H}_\mathrm{B}$ and $\mathbf{H}_\mathrm{F}$, respectively,
			and $\mathbf{P},\mathbf{Q} \in \mathbb{P}^{n \times n}$ are any permutation matrices of dimension $N_\mathrm{S}$ satisfying:
			\begin{itemize}
				\item The $(i, j)$-th entry is $1$, where
				\begin{subnumcases}{(i, j) =}
					\ \underset{\mathclap{i+j=n+1}}{\arg\min} \ \sigma_i(\mathbf{H}_\mathrm{B}) \sigma_j(\mathbf{H}_\mathrm{F}) & for $\mathbf{P}$, \label{eq:idx_nd_sv_indl_max} \\
					\ \underset{\mathclap{i+j=n+N_\mathrm{S}}}{\arg\max} \ \sigma_i(\mathbf{H}_\mathrm{B}) \sigma_j(\mathbf{H}_\mathrm{F}) & for $\mathbf{Q}$, \label{eq:idx_nd_sv_indl_min}
				\end{subnumcases}
				and ties may be broken arbitrarily;
				\item After deleting the $i$-th row and $j$-th column, the resulting submatrix $\mathbf{Y} \in \mathbb{P}^{(N_\mathrm{S}-1) \times (N_\mathrm{S}-1)}$ is any permutation matrix satisfying
				\begin{subequations}
					\begin{alignat}{2}
						\sigma_{n{-}1}(\hat{\mathbf{\Sigma}}_{\mathrm{B}} \mathbf{Y} \hat{\mathbf{\Sigma}}_{\mathrm{F}}) & {\ge} \ \ \min_{\mathclap{i+j=n+1}} \ \sigma_i(\mathbf{H}_\mathrm{B}) \sigma_j(\mathbf{H}_\mathrm{F})            &  & \text{ for } \mathbf{P}, \label{eq:perm_nd_sv_indl_max} \\
						\sigma_{n{+}1}(\hat{\mathbf{\Sigma}}_{\mathrm{B}} \mathbf{Y} \hat{\mathbf{\Sigma}}_{\mathrm{F}}) & {\le} \ \ \max_{\mathclap{i+j=n+N_\mathrm{S}}} \ \sigma_i(\mathbf{H}_\mathrm{B}) \sigma_j(\mathbf{H}_\mathrm{F}) &  & \text{ for } \mathbf{Q}, \label{eq:perm_nd_sv_indl_min}
					\end{alignat}
				\end{subequations}
				where $\hat{\mathbf{\Sigma}}_{\mathrm{B}}$ and $\hat{\mathbf{\Sigma}}_{\mathrm{F}}$ are diagonal singular value matrices of $\mathbf{H}_\mathrm{B}$ and $\mathbf{H}_\mathrm{F}$ with both $i$-th row and $j$-th column deleted, respectively.
			\end{itemize}
		\end{corollary}

		\begin{proof}
			Please refer to Appendix~\ref{ap:nd_sv_indl}.
		\end{proof}

		\begin{remark}
			\label{rm:svd}
			We emphasize that the singular matrices in the \gls{svd} are not uniquely defined.
			When a singular value has multiplicity $k$, the corresponding singular vectors can be any orthonormal basis of the $k$-dimensional subspace. Even if all singular values are distinct, the singular vectors of each can be scaled by a phase factor of choice.
			Consequently, all \gls{svd}-based scattering matrices in this paper are inherently non-unique.
		\end{remark}

		Corollary~\ref{co:nd_sv_indl} and Proposition~\ref{pp:rd} both reveal the shaping limits of the $n$-th largest channel singular value.
		The two results are derived under different assumptions are not special cases of each other.
		Importantly, Corollary~\ref{co:nd_sv_indl} establishes upper and lower bounds for \emph{each} channel singular value (c.f. first and last $k$ in Proposition~\ref{pp:rd}) and provides general solutions for fully-connected \gls{bd}-\gls{ris} of arbitrary (c.f. sufficiently large) size to attain the equalities.
		These bounds enable closed-form passive beamforming, and hence fixed channel and closed-form active beamforming, for spatial multiplexing with a limited number $n$ of \gls{rf} chains.
		We emphasize that in \eqref{eq:ris_nd_sv_indl} the mode alignment is realized by $\mathbf{V}_\mathrm{B}$ and $\mathbf{U}_\mathrm{F}$ while the ordering is enabled by permutation matrices $\mathbf{P}$ and $\mathbf{Q}$, which are special cases of $\mathbf{X}$ defined in \eqref{eq:ris_decompose_group}.
		Specially, the extreme channel singular values can be manipulated within the range of
		\begin{subequations}
			\label{iq:sv_nd_extreme}
			\begin{gather}
				\max_{\mathclap{i+j=N_\mathrm{S}+1}} \ \sigma_i(\mathbf{H}_\mathrm{B}) \sigma_j(\mathbf{H}_\mathrm{F}) {\le} \sigma_1(\mathbf{H}) {\le} \sigma_1(\mathbf{H}_\mathrm{B}) \sigma_1(\mathbf{H}_\mathrm{F}), \label{iq:sv_nd_largest} \\
				\min_{\mathclap{i+j=\bar{N}+1}} \ \sigma_i(\mathbf{H}_\mathrm{B}) \sigma_j(\mathbf{H}_\mathrm{F}) {\ge} \sigma_{\bar{N}}(\mathbf{H}) {\ge} \sigma_{\bar{N}}(\mathbf{H}_\mathrm{B}) \sigma_{\bar{N}}(\mathbf{H}_\mathrm{F}). \label{iq:sv_nd_smallest}
			\end{gather}
		\end{subequations}
		We notice that the right halves in \eqref{iq:sv_nd_largest} and \eqref{iq:sv_nd_smallest} are also special cases of \eqref{iq:sv_nd_prod_largest} and \eqref{iq:sv_nd_prod_smallest} with $k=1$.

		\begin{example}[Bounds on $3 \times 3 \times 3$ shaping]
			\label{eg:shaping_bounds}
			Consider a $3 \times 3 \times 3$ setup with $\mathbf{H}_\mathrm{D} = \mathbf{0}$, $\mathbf{H}_\mathrm{B} = \diag(3, 2, 1)$, and $\mathbf{H}_\mathrm{F} = \diag(4, 0, 5)$.
			\begin{itemize}
				\item \gls{d}-\gls{ris}: It is evident that any \gls{d}-\gls{ris} can only achieve $\sv(\mathbf{H}) = [12, 5, 0]^\mathsf{T}$ due to limited branch matching and mode alignment capabilities;
				\item \gls{bd}-\gls{ris}: According to \eqref{iq:sv_nd_indl}, a fully-connected \gls{bd}-\gls{ris} can manipulate the singular values within the range of
				\begin{equation*}
					8 \le \sigma_1(\mathbf{H}) \le 15, \quad 4 \le \sigma_2(\mathbf{H}) \le 10, \quad 0 \le \sigma_3(\mathbf{H}) \le 0.
				\end{equation*}
				To attain the upper and lower bounds, $(i,j)$ in \eqref{eq:ris_nd_sv_indl_max} and \eqref{eq:ris_nd_sv_indl_min} takes $(1, 1)$ and $(2, 2)$ when $n=1$, and $(2, 1)$ and $(3, 2)$ when $n=2$, respectively.
			\end{itemize}
		\end{example}

		We conclude from Example~\ref{eg:shaping_bounds} that a fully-connected \gls{bd}-\gls{ris} can widen the dynamic range of channel singular values by properly aligning and ordering the modes of $\mathbf{H}_\mathrm{B}$ and $\mathbf{H}_\mathrm{F}$.
		However, the individual bounds \eqref{iq:sv_nd_indl} may not be simultaneously tight when the problem of interest is a function of multiple singular values.
		Some case studies are presented below.

		\begin{corollary}[Channel power gain]
			\label{co:nd_power}
			If the direct channel is negligible, then the channel power gain is bounded from above (resp. below) by the inner product of squared singular values of $\mathbf{H}_\mathrm{B}$ and $\mathbf{H}_\mathrm{F}$ when they are sorted similarly (resp. oppositely), that is,
			\begin{equation}
				\label{iq:power_nd}
				\sum_{n=1}^N \sigma_n^2(\mathbf{H}_\mathrm{B}) \sigma_{N_\mathrm{S}-n+1}^2(\mathbf{H}_\mathrm{F}) \le \lVert \mathbf{H} \rVert _\mathrm{F}^2 \le \sum_{n=1}^N \sigma_n^2(\mathbf{H}_\mathrm{B}) \sigma_n^2(\mathbf{H}_\mathrm{F}),
			\end{equation}
			whose upper and lower bounds are attained respectively at
			\begin{subequations}
				\label{eq:ris_nd_power}
				\begin{align}
					\mathbf{\Theta}_\textnormal{P-max}^\textnormal{MIMO-ND} & = \mathbf{V}_\mathrm{B} \mathbf{U}_\mathrm{F}^\mathsf{H}, \label{eq:ris_nd_power_max}            \\
					\mathbf{\Theta}_\textnormal{P-min}^\textnormal{MIMO-ND} & = \mathbf{V}_\mathrm{B} \mathbf{J} \mathbf{U}_\mathrm{F}^\mathsf{H} \label{eq:ris_nd_power_min},
				\end{align}
			\end{subequations}
			where $\mathbf{J}$ is the exchange (a.k.a. backward identity) matrix of dimension $N_\mathrm{S}$.
		\end{corollary}
		\begin{proof}
			Please refer to Appendix~\ref{ap:nd_power}.
		\end{proof}

		We notice that \eqref{eq:ris_nd_power_max} and \eqref{eq:ris_nd_power_min} are special cases of \eqref{eq:ris_nd_sv_indl_max} and \eqref{eq:ris_nd_sv_indl_min} with $\mathbf{P} = \mathbf{I}$ and $\mathbf{Q} = \mathbf{J}$, which also attain the right and left halves of \eqref{iq:sv_nd_extreme}, respectively.
		That is to say, there exists a closed-form \gls{bd}-\gls{ris} solution \eqref{eq:ris_nd_power_max} maximizing the channel power gain that is also optimal for wireless power transfer.
		We will shortly see that this solution also achieves the channel capacity.
		The upper bound \eqref{eq:ris_nd_power_max} is also reminiscent of the optimal \gls{af} relay beamforming design \cite[(16), (17)]{Rong2009a} where the diagonal power allocation matrices boil down to $\mathbf{I}$ due to the passive nature of \gls{ris}.
		As a side note, when both $\mathbf{H}_\mathrm{B}$ and $\mathbf{H}_\mathrm{F}$ follow Rayleigh fading, the expectation of maximum channel power gain can be numerically evaluated as
		\begin{equation}
			\label{eq:power_nd_rayleigh}
			\mathbb{E}\bigl\{ \lVert \mathbf{H} \rVert _ \mathrm{F}^2 \bigr\} =
			\sum_{n=1}^N \iint_0^\infty xy
			f_{\lambda_n^{\min(N_\mathrm{R},N_\mathrm{S})}}(x)
			f_{\lambda_n^{\min(N_\mathrm{S},N_\mathrm{T})}}(y) \, dx \, dy,
		\end{equation}
		where $\lambda_n^{K}$ is the $n$-th eigenvalue of the complex $K \times K$ Wishart matrix with probability density function $f_{\lambda_n^{K}}(\cdot)$ given by \cite[(51)]{Zanella2009}.
		\eqref{eq:power_nd_rayleigh} generalizes the \gls{siso} channel power gain aided by \gls{bd}-\gls{ris} \cite[(58)]{Shen2020a} to \gls{mimo} but a closed-form expression is non-trivial.
		The next corollary has been derived in \cite{Bartoli2023} independently of Proposition~\ref{pp:nd} and we include it here for the completeness of results.
		\begin{corollary}[Channel capacity at general \gls{snr}]
			\label{co:nd_capacity_snr_general}
			If the direct channel is negligible, then the \gls{bd}-\gls{ris}-aided \gls{mimo} channel capacity is
			\begin{equation}
				\label{eq:capacity_nd}
				C^\textnormal{MIMO-ND} = \sum_{n=1}^N \log \left(1 + \frac{s_n \sigma_n^2(\mathbf{H}_\mathrm{B}) \sigma_n^2(\mathbf{H}_\mathrm{F})}{\eta}\right),
			\end{equation}
			where $\eta$ is the average noise power, $s_n = \mu - \frac{\eta}{\sigma_n^2(\mathbf{H}_\mathrm{B}) \sigma_n^2(\mathbf{H}_\mathrm{F})}$ is the power allocated to the $n$-th mode obtainable by the water-filling algorithm \cite{Clerckx2013}.
			The capacity-achieving \gls{bd}-\gls{ris} scattering matrix is
			\begin{equation}
				\label{eq:ris_nd_rate_max}
				\mathbf{\Theta}_\textnormal{R-max}^\textnormal{MIMO-ND} = \mathbf{V}_\mathrm{B} \mathbf{U}_\mathrm{F}^\mathsf{H}.
			\end{equation}
		\end{corollary}

		\begin{proof}
			Please refer to \cite[Appendix~A]{Bartoli2023}.
		\end{proof}

		One can observe from \eqref{eq:ris_nd_power_max} and \eqref{eq:ris_nd_rate_max} that the optimal channel shaping solution for channel power gain maximization, wireless power transfer, and wireless communication coincide with each other when the direct channel is negligible and the \gls{bd}-\gls{ris} is fully-connected.
		If either condition is false, the active and passive beamforming would be coupled and the rate-optimal solution involves numerical optimization.
		In such case, the power gain-optimal \gls{ris} can still provide a low-complexity decoupled solution and the details will be discussed in Section~\ref{sc:rate}.

		\begin{corollary}[Channel capacity at extreme \gls{snr}]
			\label{co:nd_capacity_snr_extreme}
			If the direct channel is negligible, then the channel capacity at extremely low and high \gls{snr} $\rho$ are approximately bounded from above by
			\begin{subequations}
				\label{iq:capacity_nd_snr_extreme}
				\begin{align}
					C_{\rho_\downarrow} & \lessapprox \rho \sigma_1^2(\mathbf{H}_\mathrm{B}) \sigma_1^2(\mathbf{H}_\mathrm{F}), \label{iq:capacity_nd_snr_low}                                      \\
					C_{\rho_\uparrow}   & \lessapprox N \log \frac{\rho}{N} + 2 \log \prod_{n=1}^N \sigma_n(\mathbf{H}_\mathrm{B}) \sigma_n(\mathbf{H}_\mathrm{F}). \label{iq:capacity_nd_snr_high}
				\end{align}
			\end{subequations}
		\end{corollary}

		\begin{proof}
			Please refer to Appendix~\ref{ap:nd_capacity}.
		\end{proof}

		The ergodic capacity \eqref{eq:capacity_nd} and \eqref{iq:capacity_nd_snr_extreme} when both $\mathbf{H}_\mathrm{B}$ and $\mathbf{H}_\mathrm{F}$ follow Rayleigh fading can be evaluated similarly to \eqref{eq:power_nd_rayleigh} and the details are omitted here.

		Proposition~\ref{pp:dof} -- \ref{pp:nd} and the resulting Corollaries provide a partial answer to the channel shaping question.
		These analyses provide a theoretical foundation for understanding exactly \emph{how} and \emph{to what extent} a \gls{bd}-\gls{ris} can shape the \gls{mimo} channel regarding typical singular value functions under specific channel conditions.
		Extending the results to more general setups is challenging due to the limited branch matching capability and the inherent tradeoff in mode alignment.
	\end{subsection}

	\begin{subsection}{Numerical Shaping Solution}
		\label{sc:shaping_numerical}
		Below we propose a numerical method to optimize \gls{bd}-\gls{ris} for a broad class of singular value functions.
		\begin{definition}[Locally Lipschitz]
			A function $f: \mathbb{R}^N \to \mathbb{R}$ is locally Lipschitz if for any compact set $\mathcal{S} \subset \mathbb{R}^N$, there exists a constant $L > 0$ such that $\lvert f(\mathbf{x}) - f(\mathbf{y}) \rvert \le L \lVert \mathbf{x} - \mathbf{y} \rVert$ for all $\mathbf{x}, \mathbf{y} \in \mathcal{S}$.
		\end{definition}

		\begin{proposition}
			\label{pp:shaping}
			Consider channel shaping problems of the form
			\begin{maxi!}
				{\scriptstyle{\mathbf{\Theta}}}{f\bigl(\sv(\mathbf{H})\bigr)}{\label{op:shaping}}{\label{ob:shaping}}
				\addConstraint{\mathbf{\Theta}_g^\mathsf{H} \mathbf{\Theta}_g=\mathbf{I},}{\quad \forall g,}{\label{cn:shaping_unitary}}
			\end{maxi!}
			where $f: \mathbb{R}^{N} \to \mathbb{R}$ is arbitrary locally Lipschitz function of channel singular values.
			The Clarke subdifferential of \eqref{ob:shaping} with respect to \gls{bd}-\gls{ris} block $g$ is
			\begin{equation}
				\label{eq:shaping_subdiff}
				\partial_{\mathbf{\Theta}_g^*} f\bigl(\sv(\mathbf{H})\bigr) = \conv \bigl\{ \mathbf{H}_{\mathrm{B},g}^\mathsf{H} \mathbf{U} \mathbf{D} \mathbf{V}^\mathsf{H} \mathbf{H}_{\mathrm{F},g}^\mathsf{H} \bigr\},
			\end{equation}
			where $\mathbf{D} \in \mathbb{C}^{N_\mathrm{R} \times N_\mathrm{T}}$ is a rectangular diagonal matrix with $[\mathbf{D}]_{n,n} \in \partial_{\sigma_n(\mathbf{H})} f\bigl(\sv(\mathbf{H})\bigr)$, $\forall n \in [N]$, and $\mathbf{U}$, $\mathbf{V}$ are any left and right singular matrices of $\mathbf{H}$.
		\end{proposition}

		\begin{proof}
			Please refer to Appendix~\ref{ap:shaping}.
		\end{proof}

		Proposition~\ref{pp:shaping} enables subgradient-based optimization for arbitrary locally Lipschitz function of channel singular values (e.g., Pareto frontier, power gain, capacity, condition number) via Clarke subdifferential \eqref{eq:shaping_subdiff}.
		Next, we introduce a \emph{geodesic\footnote{A geodesic is a curve representing the shortest path between two points in a Riemannian manifold, whose tangent vectors remain parallel when transporting along the curve.} \gls{rcg}} method modified from \cite{Abrudan2008,Abrudan2009} for the optimization of \gls{bd}-\gls{ris}.
		Our contribution is an extension to the block-unitary case with sequential, parallel, or unified updates for accelerated convergence.
		The steps for updating $\mathbf{\Theta}_g$ at iteration $r$ are summarized below, where the gradients are replaced by Clarke subgradients for non-smooth $f$.
		\begin{enumerate}[label=(\roman*)]
			\item \emph{Compute the Euclidean gradient at $\mathbf{\Theta}_g^{(r)}$:} The gradient of $f$ with respect to $\mathbf{\Theta}_g$ in the Euclidean space is
			\begin{equation}
				\label{eq:gradient_eucl}
				\nabla_{\mathrm{E},g}^{(r)} = 2 \frac{\partial f(\mathbf{\Theta}_g^{(r)})}{\partial \mathbf{\Theta}_g^*};
			\end{equation}
			\item \emph{Translate to the Riemannian gradient at $\mathbf{\Theta}_g^{(r)}$:} At point $\mathbf{\Theta}_g^{(r)}$, the Riemannian gradient gives the steepest ascent direction on the manifold. It lies in the tangent space of the manifold $\mathcal{T}_{\smash{\mathbf{\Theta}_g^{(r)}}}\mathbb{U}^{L \times L} \triangleq \{\mathbf{M} \in \mathbb{C}^{L \times L} \mid \mathbf{M}^\mathsf{H} \mathbf{\Theta}_g^{(r)} + {\mathbf{\Theta}_g^{(r)\mathsf{H}}} \mathbf{M} = \mathbf{0}\}$ and is obtainable by projection:
			\begin{equation}
				\label{eq:gradient_riem}
				\nabla_{\mathrm{R},g}^{(r)} = \nabla_{\mathrm{E},g}^{(r)} - \mathbf{\Theta}_g^{(r)} {\nabla_{\mathrm{E},g}^{(r)\mathsf{H}}} \mathbf{\Theta}_g^{(r)};
			\end{equation}
			\item \emph{Translate to the Riemannian gradient at the identity:} The Riemannian gradient should be translated back to the identity for exploiting the Lie algebra\footnote{Lie algebra refers to the tangent space of the Lie group at the identity element. A Lie group is simultaneously a continuous group and a differentiable manifold. In this example, $\mathbb{U}^{L \times L}$ formulates a Lie group and the corresponding Lie algebra consists of skew-Hermitian matrices $\mathfrak{u}(L) \triangleq \mathcal{T}_{\mathbf{I}}\mathbb{U}^{L \times L} = \{\mathbf{M} \in \mathbb{C}^{L \times L} \mid \mathbf{M}^\mathsf{H} + \mathbf{M} = \mathbf{0}\}$.}:
			\begin{equation}
				\label{eq:gradient_riem_tran}
				\tilde{\nabla}_{\mathrm{R},g}^{(r)} = \nabla_{\mathrm{R},g}^{(r)} \mathbf{\Theta}_g^{(r)\mathsf{H}} = \nabla_{\mathrm{E},g}^{(r)} \mathbf{\Theta}_g^{(r)\mathsf{H}} - \mathbf{\Theta}_g^{(r)} {\nabla_{\mathrm{E},g}^{(r)\mathsf{H}}}.
			\end{equation}
			\item \emph{Determine the conjugate direction:} The conjugate direction is obtained over the Riemannian gradient and the previous direction as
			\begin{equation}
				\label{eq:conjugate_dirn_geod}
				{\mathbf{D}}_g^{(r)} = \tilde{\nabla}_{\mathrm{R},g}^{(r)} + {\gamma}_g^{(r)} {\mathbf{D}}_g^{(r-1)},
			\end{equation}
			where $\gamma_g^{(r)}$ deviates the conjugate direction from the tangent space for accelerated convergence. A popular choice is the Polak-Ribi\`{e}re formula \cite{Hager2006}
			\begin{equation}
				\label{eq:conjugate_parm_geod}
				{\gamma}_g^{(r)} = \frac{\tr\bigl((\tilde{\nabla}_{\mathrm{R},g}^{(r)} - \tilde{\nabla}_{\mathrm{R},g}^{(r-1)}) {\tilde{\nabla}_{\mathrm{R},g}^{(r)\mathsf{H}}}\bigr)}{\tr\bigl(\tilde{\nabla}_{\mathrm{R},g}^{(r-1)} {\tilde{\nabla}_{\mathrm{R},g}^{(r-1)\mathsf{H}}}\bigr)}. 
			\end{equation}
			\item \emph{Evaluate the geodesic at the identity:} The geodesic emanating from the identity with velocity $\mathbf{D} \in \mathfrak{u}(L)$ is described by
			\begin{equation}
				\label{eq:geodesic_iden}
				\mathbf{G}_\mathbf{I}(\mu) = \exp(\mu \mathbf{D}),
			\end{equation}
			where $\exp(\mathbf{A}) = \sum_{k=0}^\infty (\mathbf{A}^k/k!)$ is the matrix exponential and $\mu$ is the step size (i.e., magnitude of the tangent vector).
			\item \emph{Translate to the geodesic at $\mathbf{\Theta}_g^{(r)}$:} The geodesic emanating from $\mathbf{\Theta}_g^{(r)}$ terminates at $\mathbf{\Theta}_g^{(r+1)}$ by multiplicative updates
			\begin{equation}
				\label{eq:geodesic_tran}
				\mathbf{\Theta}_g^{(r+1)} {=} \mathbf{G}_{\smash{\mathbf{\Theta}_g^{(r)}}}(\mu) {=} \mathbf{G}_\mathbf{I}(\mu) \mathbf{\Theta}_g^{(r)} {=} \exp(\mu \mathbf{D}_g^{(r)}) \mathbf{\Theta}_g^{(r)},
			\end{equation}
			where $\mu$ is the step size refinable\footnote{To double the step size, one can simply square the rotation matrix instead of recomputing the matrix exponential, that is, $\exp^2(\mu \mathbf{D}_g^{(r)}) = \exp(2 \mu \mathbf{D}_g^{(r)})$.} by the Armijo rule \cite{Armijo1966}.
		\end{enumerate}

		\begin{algorithm}[!t]
			\footnotesize
			\caption{Geodesic \gls{rcg} for \gls{bd}-\gls{ris} design}
			\label{ag:rcg}
			\begin{algorithmic}[1]
				\Require $f(\mathbf{\Theta})$, $G$
				\Ensure $\mathbf{\Theta}^\star$
				\Initialize {$r \gets 0$, $\mathbf{\Theta}^{(0)}$}
				\Repeat
				\For {$g \gets 1$ to $G$}
				\State $\nabla_{\mathrm{E},g}^{(r)} \gets$ \eqref{eq:gradient_eucl}, $\tilde{\nabla}_{\mathrm{R},g}^{(r)} \gets$ \eqref{eq:gradient_riem_tran}, ${\gamma}_g^{(r)} \gets$ \eqref{eq:conjugate_parm_geod}, $\mathbf{D}_g^{(r)} \gets$ \eqref{eq:conjugate_dirn_geod}
				\If {$\Re\bigl\{\tr({\mathbf{D}_g^{(r)\mathsf{H}}} \tilde{\nabla}_{\mathrm{R},g}^{(r)})\bigr\} < 0$} \Comment{Not ascent}
				\State $\mathbf{D}_g^{(r)} \gets \tilde{\nabla}_{\mathrm{R},g}^{(r)}$
				\EndIf
				\State $\mu \gets 0.1$, $\mathbf{G}_{\smash{\mathbf{\Theta}_g^{(r)}}}(\mu) \gets$ \eqref{eq:geodesic_tran}
				\While {$ f\bigl(\mathbf{G}_{\smash{\mathbf{\Theta}_g^{(r)}}}(2\mu)\bigr) - f(\mathbf{\Theta}_g^{(r)}) \ge \mu \cdot \smash{\frac{\tr(\mathbf{D}_g^{(r)} {\mathbf{D}_g^{(r)\mathsf{H}}})}{2}}$}\label{ln:armijo_start}
				\State $\mu \gets 2 \mu$
				\EndWhile
				\While {$f\bigl(\mathbf{G}_{\smash{\mathbf{\Theta}_g^{(r)}}}(\mu)\bigr) - f(\mathbf{\Theta}_g^{(r)}) < \frac{\mu}{2} \cdot \smash{\frac{\tr(\mathbf{D}_g^{(r)} {\mathbf{D}_g^{(r)\mathsf{H}}})}{2}}$}
				\State $\mu \gets \mu / 2$
				\EndWhile \label{ln:armijo_end}
				\State $\mathbf{\Theta}_g^{(r+1)} \gets$ \eqref{eq:geodesic_tran}
				\EndFor
				\State $r \gets r+1$
				\Until $\lvert f(\mathbf{\Theta}^{(r)}) - f(\mathbf{\Theta}^{(r-1)}) \rvert / f(\mathbf{\Theta}^{(r-1)}) \le \epsilon$
			\end{algorithmic}
		\end{algorithm}

		Algorithm~\ref{ag:rcg} summarizes the proposed geodesic \gls{rcg} method with sequential group-wise updates.
		Each iteration leverages Lie algebra to perform a multiplicative update \eqref{eq:geodesic_tran} along the geodesics of the Stiefel manifold.
		This appropriate parameter space leads to faster convergence and easier step size tuning.
		We remark that the additive update of the non-geodesic \gls{rcg} can be interpreted a first-order Taylor approximation to the multiplicative update of the proposed geodesic \gls{rcg}, thus necessitating a retraction step to remain on the manifold.
		The group-wise updates can be performed in parallel to facilitate large-scale \gls{bd}-\gls{ris} design problems.
		One may also operate on $\mathbf{\Theta}$ and pinching (i.e., keeping the main block diagonal and nulling the rest) \eqref{eq:gradient_eucl} to unify the step size selection for further acceleration.

		We now analyze the computational complexity of solving singular value shaping problem \eqref{op:shaping} by Algorithm~\ref{ag:rcg}.
		To update each \gls{bd}-\gls{ris} group, \gls{svd} of $\mathbf{H}$ requires $\mathcal{O}(N N_\mathrm{T} N_\mathrm{R})$ flops, Euclidean subgradient \eqref{eq:shaping_subdiff} requires $\mathcal{O}\bigl(L N (N_\mathrm{T}+N_\mathrm{R}+L) \bigr)$ flops, Riemannian subgradient translation \eqref{eq:gradient_riem_tran} requires $\mathcal{O}(L^3)$ flops, deviation parameter \eqref{eq:conjugate_parm_geod} and conjugate direction \eqref{eq:conjugate_dirn_geod} together require $\mathcal{O}(L^2)$ flops, and matrix exponential \eqref{eq:geodesic_tran} requires $\mathcal{O}(L^3)$ flops \cite{Moler2003}.
		The overall complexity is thus $\mathcal{O}\bigl(I_\text{RCG} G \bigl(N N_\mathrm{T} N_\mathrm{R} + L N (N_\mathrm{T}+N_\mathrm{R}+L) + I_\text{BLS} L^3\bigr)\bigr)$, where $I_\text{RCG}$ and $I_\text{BLS}$ are the number of iterations for geodesic \gls{rcg} and backtracking line search (i.e., lines \ref{ln:armijo_start} -- \ref{ln:armijo_end} of Algorithm \ref{ag:rcg}), respectively.
		That is, $\mathcal{O}\bigl(N_\mathrm{S}\bigr)$ for \gls{d}-\gls{ris} and $\mathcal{O}\bigl(N_\mathrm{S}^3\bigr)$ for fully-connected \gls{bd}-\gls{ris}.

		To validate Algorithm~\ref{ag:rcg} and quantify the shaping capability of \gls{bd}-\gls{ris}, we aim to characterize the achievable singular value region of \gls{bd}-\gls{ris}-aided \gls{mimo} channel by considering the Pareto optimization problem
		\begin{maxi!}
			{\scriptstyle{\mathbf{\Theta}}}{\sum_{n=1}^N \rho_n \sigma_n(\mathbf{H})}{\label{op:pareto}}{\label{ob:pareto}}
			\addConstraint{\mathbf{\Theta}_g^\mathsf{H} \mathbf{\Theta}_g=\mathbf{I},}{\quad \forall g,}{\label{cn:pareto_unitary}}
		\end{maxi!}
		where $\rho_n \ge 0$ is the weight associated with the $n$-th channel singular value.
		Varying those weights help to characterize the Pareto frontier that encloses the achievable singular value region.
		While the objective \eqref{ob:pareto} may seem obscure, a larger quantity translates to a stronger singular value redistribution capability and thus better wireless performance (e.g., channel capacity for communication \cite{Clerckx2013}, detection probability for sensing \cite{Liu2022c}, and harvested power for power transfer \cite{Shen2021}).
		Problem~\eqref{op:pareto} also generalizes the \gls{dof} problem in Proposition~\ref{pp:dof} and the individual singular value shaping problem in Proposition~\ref{pp:rd} and Corollary~\ref{co:nd_sv_indl}.
		It can be solved optimally by Algorithm~\ref{ag:rcg} with $[\mathbf{D}]_{n,n} = \rho_n$ in \eqref{eq:shaping_subdiff}.
	\end{subsection}
\end{section}

\begin{section}{Rate Maximization}
	\label{sc:rate}
	In this section, we first solve the \gls{bd}-\gls{ris}-aided \gls{mimo} rate maximization problem optimally by joint beamforming design, and then exploit channel shaping for a low-complexity two-stage solution.
	The problem is formulated as
	\begin{maxi!}
		{\scriptstyle{\mathbf{W},\mathbf{\Theta}}}{R = \log \det \biggl(\mathbf{I} + \frac{\mathbf{W}^\mathsf{H}\mathbf{H}^\mathsf{H}\mathbf{H}\mathbf{W}}{\eta}\biggr)}{\label{op:rate}}{\label{ob:rate}}
		\addConstraint{\lVert \mathbf{W} \rVert _\mathrm{F}^2}{\le P}
		\addConstraint{\mathbf{\Theta}_g^\mathsf{H} \mathbf{\Theta}_g}{=\mathbf{I}, \quad \forall g,\label{cn:rate_unitary}}
	\end{maxi!}
	where $\mathbf{W}$ is the transmit precoder, $R$ is the achievable rate, $\eta$ is the average noise power, and $P$ is maximum average transmit power.
	Problem~\eqref{op:rate} is non-convex due to the block-unitary constraint \eqref{cn:rate_unitary} and the coupling between variables.

	\begin{subsection}{Alternating Optimization}
		\label{sc:rate_ao}
		This approach updates $\mathbf{\Theta}$ and $\mathbf{W}$ iteratively until convergence.
		For a given $\mathbf{W}$, the passive beamforming subproblem is
		\begin{maxi!}
			{\scriptstyle{\mathbf{\Theta}}}{\log \det \biggl(\mathbf{I} + \frac{\mathbf{H} \mathbf{Q} \mathbf{H}^\mathsf{H}}{\eta}\biggr)}{\label{op:rate_ris}}{\label{ob:rate_ris}}
			\addConstraint{\mathbf{\Theta}_g^\mathsf{H} \mathbf{\Theta}_g=\mathbf{I}, \quad \forall g,}{}{}
		\end{maxi!}
		where $\mathbf{Q} \triangleq \mathbf{W} \mathbf{W}^\mathsf{H}$ is the transmit covariance matrix.
		Problem~\eqref{op:rate_ris} can be solved optimally by Algorithm \ref{ag:rcg} with the partial derivative given in Lemma~\ref{lm:rate}.
		\begin{lemma}
			\label{lm:rate}
			The partial derivative of \eqref{ob:rate_ris} with respect to \gls{bd}-\gls{ris} block $g$ is
			\begin{equation}
				\label{eq:gradient_eucl_rate}
				\frac{\partial R}{\partial \mathbf{\Theta}_g^*} = \frac{1}{\eta} \mathbf{H}_{\mathrm{B},g}^\mathsf{H} \biggl(\mathbf{I} + \frac{\mathbf{H}\mathbf{Q}\mathbf{H}^\mathsf{H}}{\eta}\biggr)^{-1} \mathbf{H} \mathbf{Q} \mathbf{H}_{\mathrm{F},g}^\mathsf{H},
			\end{equation}
			which is a special case of \eqref{eq:shaping_subdiff}.
		\end{lemma}

		\begin{proof}
			Please refer to Appendix~\ref{ap:rate}.
		\end{proof}
		For a given $\mathbf{\Theta}$, the optimal transmit precoder is given by eigenmode transmission \cite{Clerckx2013}
		\begin{equation}
			\label{eq:precoder_rate}
			\mathbf{W}^\star = \mathbf{V} {\diag(\mathbf{s}^\star)}^{1/2},
		\end{equation}
		where $\mathbf{V}$ is the right singular matrix of $\mathbf{H}$ and $\mathbf{s}^\star$ can be retrieved by the water-filling algorithm \cite{Clerckx2013}.
		The \gls{ao} algorithm is guaranteed to converge to local-optimal points of problem \eqref{op:rate} since each subproblem is solved optimally and the objective is bounded above.
		The computational complexity of solving subproblem \eqref{op:rate_ris} by geodesic \gls{rcg} is $\mathcal{O}\bigl(I_\text{RCG} G (NL^2 + L N_\mathrm{T} N_\mathrm{R} + N_\mathrm{T}^2 N_\mathrm{R} + N_\mathrm{T} N_\mathrm{R}^2 + N_\mathrm{R}^3 + I_\text{BLS} L^3)\bigr)$.
		On the other hand, the complexity of active beamforming \eqref{eq:precoder_rate} is $\mathcal{O}\bigl(N N_\mathrm{T} N_\mathrm{R}\bigr)$.
		The overall complexity is thus $\mathcal{O}\bigl(I_\text{AO}\bigl(I_\text{RCG} G (NL^2 + L N_\mathrm{T} N_\mathrm{R} + N_\mathrm{T}^2 N_\mathrm{R} + N_\mathrm{T} N_\mathrm{R}^2 + N_\mathrm{R}^3 + I_\text{BLS} L^3) + N N_\mathrm{T} N_\mathrm{R}\bigr)\bigr)$, where $I_\text{AO}$ is the number of iterations for \gls{ao}.
		That is, $\mathcal{O}\bigl(N_\mathrm{S}\bigr)$ for \gls{d}-\gls{ris} and $\mathcal{O}\bigl(N_\mathrm{S}^3\bigr)$ for fully-connected \gls{bd}-\gls{ris}.
	\end{subsection}

	\begin{subsection}{Low-Complexity Solution}
		\label{sc:rate_lc}
		To reduce computational complexity, we decouple the joint beamforming design by first shaping the \gls{mimo} channel by \gls{bd}-\gls{ris} for maximum power gain and then performing eigenmode transmission.
		The shaping subproblem is formulated as
		\begin{maxi!}
			{\scriptstyle{\mathbf{\Theta}}}{\lVert \mathbf{H}_\mathrm{D} + \mathbf{H}_\mathrm{B} \mathbf{\Theta} \mathbf{H}_\mathrm{F} \rVert _\mathrm{F}^2}{\label{op:power}}{\label{ob:power}}
			\addConstraint{\mathbf{\Theta}_g^\mathsf{H} \mathbf{\Theta}_g=\mathbf{I}, \quad \forall g.}{\label{cn:power_unitary}}{}
		\end{maxi!}
		While similar problems have been studied in single-mode cases \cite{Shen2020a,Santamaria2023}, generalizing those methods to \gls{mimo} remains non-trivial due to the tradeoff between multi-mode alignments.
		One can see that the objective function \eqref{ob:power} is equivalent to $\sum_{n=1}^N \sigma_n^2(\mathbf{H})$ and thus readily solvable by Algorithm \ref{ag:rcg}.
		Below we propose a more elegant power iteration method inspired by \cite{Nie2017} that iterates in closed-form by orthogonal projection.
		The idea is to approximate the quadratic objective \eqref{ob:power} by its first-order Taylor expansion and solve each subproblem by group-wise \gls{svd}.

		\begin{proposition}
			\label{pp:power}
			Starting from any feasible $\mathbf{\Theta}^{(0)}$, the orthogonal projection of
			\begin{equation}
				\label{eq:auxiliary_power}
				\mathbf{M}_g^{(r)} = \mathbf{H}_{\mathrm{B},g}^\mathsf{H} \Bigl(\mathbf{H}_\mathrm{D} + \mathbf{H}_\mathrm{B} \diag\bigl(\mathbf{\Theta}_{[1:g-1]}^{(r+1)},\mathbf{\Theta}_{[g:G]}^{(r)}\bigr) \mathbf{H}_\mathrm{F}\Bigr) \mathbf{H}_{\mathrm{F},g}^\mathsf{H}
			\end{equation}
			onto the Stiefel manifold, given in the closed-form \cite{Manton2002}
			\begin{equation}
				\label{eq:ris_power}
				\mathbf{\Theta}_g^{(r+1)} = \underset{\mathbf{X}_g \in \mathbb{U}^{L \times L}}{\arg\min} \lVert \mathbf{M}_g - \mathbf{X}_g \rVert _\mathrm{F} = \mathbf{U}_g^{(r)} \mathbf{V}_g^{(r)\mathsf{H}},
			\end{equation}
			monotonically increases the objective function \eqref{ob:power},
			where $\mathbf{U}_g^{(r)}$ and $\mathbf{V}_g^{(r)}$ are any left and right singular matrices of $\mathbf{M}_g^{(r)}$.
			When \eqref{eq:auxiliary_power} converges, \eqref{eq:ris_power} leads to a convergence of the objective function \eqref{ob:power} towards a stationary point.
		\end{proposition}

		\begin{proof}
			Please refer to Appendix~\ref{ap:power}.
		\end{proof}

		\begin{remark}
			\label{rm:power}
			Although a rigorous convergence proof remains intricate due to the non-uniqueness of \gls{svd}, empirical evidence from extensive simulation suggests that
			\eqref{eq:auxiliary_power} converges reliably from random initializations such that \eqref{eq:ris_power} consistently provides an optimal solution to problem \eqref{op:power}.
		\end{remark}

		To update each \gls{bd}-\gls{ris} group, the matrix multiplication \eqref{eq:auxiliary_power} requires $\mathcal{O}\bigl(N_\mathrm{T} N_\mathrm{R} + NL^2+N_\mathrm{T} N_\mathrm{R} L\bigr)$ flops and its \gls{svd} requires $\mathcal{O}(L^3)$ flops.
		The overall complexity is thus $\mathcal{O}\bigl(I_\text{SAA} G \bigl(N_\mathrm{T} N_\mathrm{R} + NL^2+N_\mathrm{T} N_\mathrm{R} L + L^3\bigr)\bigr)$, where $I_\text{SAA}$ is the number iterations for successive affine approximation.
		That is, $\mathcal{O}\bigl(N_\mathrm{S}\bigr)$ for \gls{d}-\gls{ris} and $\mathcal{O}\bigl(N_\mathrm{S}^3\bigr)$ for fully-connected \gls{bd}-\gls{ris}.
		It is worth mentioning that the computational complexity for fully-connected \gls{bd}-\gls{ris} can be further reduced:
		\begin{itemize}
			\item \emph{Negligible direct channel:} The optimal solution to \eqref{op:power} has been solved in closed form by \eqref{eq:ris_nd_power_max};
			\item \emph{Non-negligible direct channel:} In terms of maximizing the inner product $\langle \mathbf{H}_\mathrm{D}, \mathbf{H}_\mathrm{B} \mathbf{\Theta} \mathbf{H}_\mathrm{F} \rangle$, \eqref{op:power} is reminiscent of the weighted orthogonal Procrustes problem \cite{Viklands2006}
			\begin{mini!}
				{\scriptstyle{\mathbf{\Theta}}}{\lVert \mathbf{H}_\mathrm{D} - \mathbf{H}_\mathrm{B} \mathbf{\Theta} \mathbf{H}_\mathrm{F} \rVert _\mathrm{F}^2}{\label{op:procrustes_wt}}{}
				\addConstraint{\mathbf{\Theta}^\mathsf{H} \mathbf{\Theta}=\mathbf{I},}{\label{cn:procrustes_wt}}{}
			\end{mini!}
			which still has no trivial solution.
			One \emph{lossy} transformation \cite{Bell2003} shifts $\mathbf{\Theta}$ to sides of the product by Moore-Penrose inverse, formulating standard orthogonal Procrustes problems
			\begin{mini!}
				{\scriptstyle{\mathbf{\Theta}}}{\lVert \mathbf{H}_\mathrm{B}^\dagger \mathbf{H}_\mathrm{D} - \mathbf{\Theta} \mathbf{H}_\mathrm{F} \rVert _\mathrm{F}^2 \text{ or } \lVert \mathbf{H}_\mathrm{D} \mathbf{H}_\mathrm{F}^\dagger - \mathbf{H}_\mathrm{B} \mathbf{\Theta} \rVert _\mathrm{F}^2}{\label{op:procrustes}}{}
				\addConstraint{\mathbf{\Theta}^\mathsf{H} \mathbf{\Theta}=\mathbf{I},}{}{}
			\end{mini!}
			with optimal solutions \cite[(6.4.1)]{Golub2013}
			\begin{equation}
				\label{eq:ris_procrustes}
				\mathbf{\Theta}_\textnormal{P-max-approx}^\textnormal{MIMO} = \mathbf{U} \mathbf{V}^\mathsf{H},
			\end{equation}
			where $\mathbf{U}$ and $\mathbf{V}$ are respectively any left and right singular matrices of $\mathbf{H}_\mathrm{B}^\dagger \mathbf{H}_\mathrm{D} \mathbf{H}_\mathrm{F}^\mathsf{H}$ or $\mathbf{H}_\mathrm{B}^\mathsf{H} \mathbf{H}_\mathrm{D} \mathbf{H}_\mathrm{F}^\dagger$.
		\end{itemize}

		Although \eqref{eq:ris_nd_power_max} and \eqref{eq:ris_procrustes} are of similar form, the latter is neither optimal nor a generalization of the former due to the lossy transformation.
		We will show by simulation that \eqref{eq:ris_procrustes} still achieves near-optimal performance on average.
		Once the channel is shaped in closed form by \eqref{eq:ris_power} or \eqref{eq:ris_nd_power_max} or \eqref{eq:ris_procrustes}, the active beamforming is retrieved in closed form by \eqref{eq:precoder_rate}.
		This two-stage solution avoids outer iterations and efficiently handles (or avoids) inner iterations.
	\end{subsection}

\end{section}

\begin{section}{Simulation Results}
	\label{sc:simulation}
	In this section, we provide numerical results to evaluate the proposed \gls{bd}-\gls{ris} designs.\footnote{Source code is available at \url{https://github.com/snowztail/channel-shaping}.}
	Consider a distance-dependent path loss model $\Lambda(d) = \Lambda_0 d^{-\gamma}$ where $\Lambda_0$ is the reference path loss at distance \qty{1}{m}, $d$ is the propagation distance, and $\gamma$ is the path loss exponent.
	We set $\Lambda_0=\qty{-30}{dB}$, $\gamma_\mathrm{D}=3$, $\gamma_\mathrm{F}=2.4$, $\gamma_\mathrm{B}=2$, $d_\mathrm{D}=\qty{14.7}{m}$, $d_\mathrm{F}=\qty{10}{m}$, $d_\mathrm{B}=\qty{6.3}{m}$, which corresponds to a typical indoor environment with $\Lambda_\mathrm{D}=\qty{-65}{dB}$, $\Lambda_\mathrm{F}=\qty{-54}{dB}$, $\Lambda_\mathrm{B}=\qty{-46}{dB}$.
	The small-scale fading model is $\mathbf{H} = \sqrt{\kappa/(1+\kappa)} \mathbf{H}_\text{LoS} + \sqrt{1/(1+\kappa)} \mathbf{H}_\text{NLoS}$, where $\kappa$ is the Rician K-factor, $\mathbf{H}_\text{LoS}$ is the deterministic \gls{los} component, and $\mathbf{H}_\text{NLoS} \sim \mathcal{N}_{\mathbb{C}}(\mathbf{0}, \mathbf{I})$ is the Rayleigh component.
	Unless otherwise specified, we assume the direct channel is present, $\kappa = 0$ (i.e., Rayleigh fading) for all channels, and $\eta = \qty{-75}{dB}$.

	\begin{subsection}{Algorithm Evaluation}
		Here are benchmarks obtained running \texttt{MATLAB~R2023a} on an octa-core \texttt{AMD~Ryzen7~5800U} processor @ \qty{4.5}{\GHz} with \qty{16}{\gibi\byte} memory.
		\begin{table}[!t]
			\caption{Performance of \gls{rcg} Algorithms on \eqref{op:pareto} with $N_\mathrm{T}=N_\mathrm{R}=4$, $L=16$}
			\label{tb:complexity_algorithm}
			\centering
			\resizebox{\columnwidth}{!}{
				\begin{tabular}{ccccccc}
					\toprule
					\multirow{2}{*}{\gls{rcg} path} & \multicolumn{2}{c}{$N_\mathrm{S}=16$} & \multicolumn{2}{c}{$N_\mathrm{S}=64$} & \multicolumn{2}{c}{$N_\mathrm{S}=256$}                                      \\
					\cmidrule(lr){2-3} \cmidrule(lr){4-5} \cmidrule(lr){6-7}
					                                & Iterations                            & Time [ms]                             & Iterations                             & Time [ms] & Iterations & Time [ms] \\
					\midrule
					Geodesic                        & 6.493                                 & 1.807                                 & 9.003                                  & 7.378     & 12.98      & 49.41     \\
					Non-geodesic (Manopt)           & 8.601                                 & 25.90                                 & 11.09                                  & 36.27     & 14.29      & 65.89     \\
					\bottomrule
				\end{tabular}
			}
		\end{table}
		Table~\ref{tb:complexity_algorithm} benchmarks two \gls{rcg} algorithms on the Pareto singular value problem \eqref{op:pareto} with $N_\mathrm{T}=N_\mathrm{R}=4$ and $L=16$.
		The geodesic \gls{rcg} is implemented with pinched gradients w.r.t. $\mathbf{\Theta}$ and unified step size selection; please refer to the discussion below \eqref{eq:geodesic_tran} for details.
		The non-geodesic \gls{rcg} is implemented by \texttt{Manopt} toolbox at commit \texttt{a879a0d} \cite{Boumal14a}.
		Both algorithms employ a stopping criterion of relative change in the objective function with a tolerance of $\epsilon = \num{1e-4}$, such that the final values are identical within reasonable precision.
		The statistics are averaged over 1000 independent channel realizations.
		We observe that the non-geodesic \gls{rcg} typically requires 1 to 2 more iterations than its geodesic counterpart.
		This is because the addition is in the tangent space of the manifold and is less effective than manifold-native updates.
		When it comes to elapsed time, the geodesic \gls{rcg} is $1333\%$ faster than the non-geodesic counterpart when $N_\mathrm{S}=16$.
		The main reason is that the geodesic \gls{rcg} avoids the retraction step from the Euclidean space to the manifold.
		According to the profiler report, around $60\%$ of the non-geodesic \gls{rcg} runtime is spent on retraction, which becomes the main bottleneck of the algorithm.
		The advantage narrows down to $391.6\%$ and $33.35\%$ when $N_\mathrm{S}=64$ and $256$, respectively.
		This is because accurately evaluating matrix exponential can be time-consuming for large $N_\mathrm{S}$.

		\begin{table}[!t]
			\caption{Performance of \gls{d}-\gls{ris} and \gls{bd}-\gls{ris} on \eqref{op:rate} with $N_\mathrm{T}=N_\mathrm{R}=4$}
			\label{tb:complexity_bond}
			\centering
			\resizebox{\columnwidth}{!}{
				\begin{tabular}{ccccccc}
					\toprule
					\multirow{2}{*}{\gls{ris} type} & \multicolumn{2}{c}{$N_\mathrm{S}=16$} & \multicolumn{2}{c}{$N_\mathrm{S}=64$} & \multicolumn{2}{c}{$N_\mathrm{S}=256$}                                      \\
					\cmidrule(lr){2-3} \cmidrule(lr){4-5} \cmidrule(lr){6-7}
					                                & Iterations                            & Time [ms]                             & Iterations                             & Time [ms] & Iterations & Time [ms] \\ \midrule
					Diagonal                        & 2.010                                 & 7.848                                 & 2.023                                  & 36.33     & 2.141      & 261.1     \\
					Fully-connected \gls{bd}        & 2.049                                 & 4.878                                 & 2.027                                  & 15.17     & 2.030      & 305.5     \\ \bottomrule
				\end{tabular}
			}
		\end{table}
		Table~\ref{tb:complexity_bond} compares the performance of \gls{d}-\gls{ris} and fully-connected \gls{bd}-\gls{ris} on rate maximization problem \eqref{op:rate} using the \gls{ao} design in Section~\ref{sc:rate_ao}, where $N_\mathrm{T}=N_\mathrm{R}=4$ and $P=\qty{20}{dB}$.
		The statistics are averaged over \num{1000} independent runs.
		Interestingly, as opposite to the asymptotic complexity analysis, {the optimization of fully-connected \gls{bd}-\gls{ris} actually takes shorter elapsed time than \gls{d}-\gls{ris} when $N_\mathrm{S}$ is not excessively large.}
		One possible reason is that fully-connected \gls{bd}-\gls{ris} only involves 1 backtracking line search per iteration while \gls{d}-\gls{ris} requires $N_\mathrm{S}$ times.
		Another reason is that the group-wise update of \gls{d}-\gls{ris} leads to slower convergence of inner iterations.
		These numerical results, together with the closed-form solutions provided in the analysis section, together suggest that designing a practically-sized \gls{bd}-\gls{ris} may be less computational expensive than expected.
	\end{subsection}

	\begin{subsection}{Channel Singular Value Redistribution}
		\begin{subsubsection}{Achievable Singular Value Region}
			\begin{figure}[!t]
				\centering
				\subfloat[Top view]{
					\label{fg:singular_region_top}
					\includegraphics[width=0.485\columnwidth]{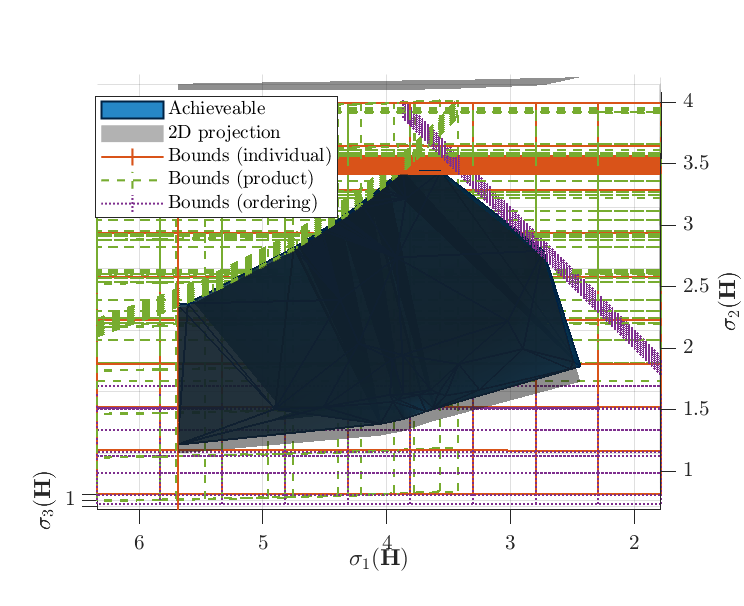}
				}
				\subfloat[Side view]{
					\label{fg:singular_region_side}
					\includegraphics[width=0.485\columnwidth]{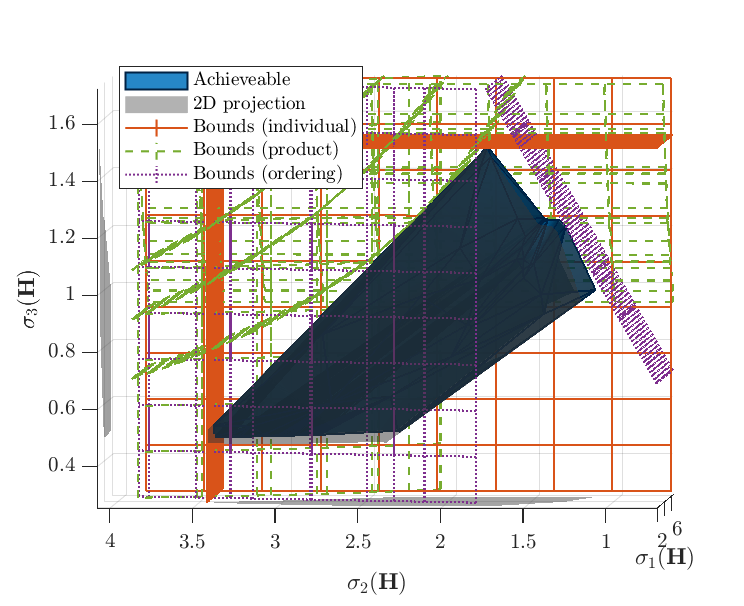}
				}
				\caption{Theoretical singular value outer bounds \eqref{iq:horn} (uniformly-spaced mesh grids) vs achievable singular value region by solving \eqref{op:pareto} (solid dark shape) for one channel realization, where $N_\mathrm{T}=N_\mathrm{S}=N_\mathrm{R}=3$, the direct channel is negligible, and the \gls{bd}-\gls{ris} is fully-connected. Small offsets are introduced on both views such that the active bounds are highlighted by densely-spaced curves/lines that marginally overlap the region from above. The achievable region lies entirely within the intersection of the bounding surfaces in the 3D space.}
				\label{fg:singular_region}
			\end{figure}

			Fig.~\ref{fg:singular_region} compares the achievable singular value region obtained by solving problem \eqref{op:pareto} and its outer bounds suggested by Corollary~\ref{co:nd_sv_prod_subset}.
			Here $\bar{N}=N_\mathrm{T}=N_\mathrm{S}=N_\mathrm{R}=3$ and the bounds are enumerated as
			\begin{subequations}
				\label{iq:sv_3_bounds}
				\small
				\begin{equation}
					\label{iq:sv_3_individual}
					\begin{gathered}
						\sigma_1(\mathbf{H}) \le \sigma_1(\mathbf{H}_\mathrm{B}) \sigma_1(\mathbf{H}_\mathrm{F}), \quad
						\sigma_2(\mathbf{H}) \le \sigma_1(\mathbf{H}_\mathrm{B}) \sigma_2(\mathbf{H}_\mathrm{F}), \\
						\sigma_2(\mathbf{H}) \le \sigma_2(\mathbf{H}_\mathrm{B}) \sigma_1(\mathbf{H}_\mathrm{F}), \quad
						\sigma_3(\mathbf{H}) \le \sigma_1(\mathbf{H}_\mathrm{B}) \sigma_3(\mathbf{H}_\mathrm{F}), \\
						\sigma_3(\mathbf{H}) \le \sigma_2(\mathbf{H}_\mathrm{B}) \sigma_2(\mathbf{H}_\mathrm{F}), \quad
						\sigma_3(\mathbf{H}) \le \sigma_3(\mathbf{H}_\mathrm{B}) \sigma_1(\mathbf{H}_\mathrm{F}),
					\end{gathered}
				\end{equation}
				\begin{equation}
					\label{iq:sv_3_product}
					\begin{gathered}
						\sigma_1(\mathbf{H}) \sigma_2(\mathbf{H}) \le \sigma_1(\mathbf{H}_\mathrm{B}) \sigma_2(\mathbf{H}_\mathrm{B}) \sigma_1(\mathbf{H}_\mathrm{F}) \sigma_2(\mathbf{H}_\mathrm{F}), \\
						\sigma_1(\mathbf{H}) \sigma_3(\mathbf{H}) \le \sigma_1(\mathbf{H}_\mathrm{B}) \sigma_3(\mathbf{H}_\mathrm{B}) \sigma_1(\mathbf{H}_\mathrm{F}) \sigma_2(\mathbf{H}_\mathrm{F}), \\
						\sigma_1(\mathbf{H}) \sigma_3(\mathbf{H}) \le \sigma_1(\mathbf{H}_\mathrm{B}) \sigma_2(\mathbf{H}_\mathrm{B}) \sigma_1(\mathbf{H}_\mathrm{F}) \sigma_3(\mathbf{H}_\mathrm{F}), \\
						\sigma_2(\mathbf{H}) \sigma_3(\mathbf{H}) \le \sigma_1(\mathbf{H}_\mathrm{B}) \sigma_2(\mathbf{H}_\mathrm{B}) \sigma_2(\mathbf{H}_\mathrm{F}) \sigma_3(\mathbf{H}_\mathrm{F}), \\
						\sigma_2(\mathbf{H}) \sigma_3(\mathbf{H}) \le \sigma_2(\mathbf{H}_\mathrm{B}) \sigma_3(\mathbf{H}_\mathrm{B}) \sigma_1(\mathbf{H}_\mathrm{F}) \sigma_2(\mathbf{H}_\mathrm{F}), \\
						\sigma_2(\mathbf{H}) \sigma_3(\mathbf{H}) \le \sigma_1(\mathbf{H}_\mathrm{B}) \sigma_3(\mathbf{H}_\mathrm{B}) \sigma_1(\mathbf{H}_\mathrm{F}) \sigma_3(\mathbf{H}_\mathrm{F}),
					\end{gathered}
				\end{equation}
				\begin{equation}
					\label{iq:sv_3_ordering}
					\sigma_1(\mathbf{H}) \ge \sigma_2(\mathbf{H}) \ge \sigma_3(\mathbf{H}),
				\end{equation}
			\end{subequations}
			where \eqref{iq:sv_3_individual}, \eqref{iq:sv_3_product} are explicit results of \eqref{iq:horn} while \eqref{iq:sv_3_ordering} denotes the ordering of singular values.
			Those are labeled respectively as `Bounds (individual)', `Bounds (product)', and `Bounds (ordering)' in Fig.~\ref{fg:singular_region}.
			The two views confirm that the theoretical outer bounds are not everywhere tight with many entries being redundant, but they provide a conservative estimate of the achievable singular value region.
			Importantly, the vertices of the region lie on the bounding surfaces and can be obtained in closed form without performing optimization.



			\begin{figure}[!t]
				\centering
				\subfloat[$2 \times 32 \times 2$ (no direct)]{
					\label{fg:singular_pareto_sx32_nd}
					\resizebox{!}{3.25cm}{
						\input{assets/simulation/pc_singular_pareto_sx32_nd.tex}

					}
				}
				\subfloat[$2 \times 32 \times 2$]{
					\label{fg:singular_pareto_sx32}
					\resizebox{!}{3.25cm}{
						\input{assets/simulation/pc_singular_pareto_sx32.tex}
					}
				}
				\\
				\subfloat[$2 \times 64 \times 2$]{
					\label{fg:singular_pareto_sx64}
					\resizebox{!}{3.25cm}{
						\input{assets/simulation/pc_singular_pareto_sx64.tex}
					}
				}
				\subfloat[$2 \times 128 \times 2$]{
					\label{fg:singular_pareto_sx128}
					\resizebox{!}{3.25cm}{
						\input{assets/simulation/pc_singular_pareto_sx128.tex}
					}
				}
				\caption{Achievable singular value regions of an $N_\mathrm{T}=N_\mathrm{R}=2$ channel shaped by \gls{bd}-\gls{ris}. The singular value pair of the direct channel are marked as baseline.
					On the Pareto frontiers, `P-max', `E-max', and `R-max' refer to the channel power gain-optimal point, wireless power transfer-optimal point, and rate-optimal arc, respectively.}
				\label{fg:singular_pareto}
			\end{figure}
			Fig.~\ref{fg:singular_pareto} illustrates the achievable regions of singular values of an $N_\mathrm{T}=N_\mathrm{R}=2$ point-to-point \gls{mimo} shaped by \gls{ris}, where the channel power gain-optimal point, wireless power transfer-optimal point, and rate-optimal arc are highlighted on the Pareto frontiers.
			The results are obtained by solving the channel shaping problem \eqref{op:pareto} merely without any application-specific optimization.
			As the \gls{snr} increases, the rate-optimal point proceeds on the arc from the east (favoring $\sigma_1(\mathbf{H})$) to the north (favoring $\sigma_2(\mathbf{H})$), which aligns with the expected behavior of water filling.
			When the direct channel is negligible, the achievable regions in Fig.~\subref*{fg:singular_pareto_sx32_nd} are shaped like pizza slices.
			This is because $\sigma_1(\mathbf{H}) \ge \sigma_2(\mathbf{H})$ and there exists a tradeoff between the alignment of two modes.
			The smallest singular value can be enhanced up to \num{2e-4} by \gls{d}-\gls{ris} and \num{3e-4} by fully-connected \gls{bd}-\gls{ris}, corresponding to a \qty{50}{\percent} gain.
			We also see that for fully-connected \gls{bd}-\gls{ris}, there exists a point that is simultaneously optimal for channel power gain, harvested power of wireless power transfer, and achievable rate of wireless communication, as indicated by \eqref{eq:ris_nd_power_max}, \eqref{eq:ris_nd_sv_indl_max}, and \eqref{eq:ris_nd_rate_max}.
			Interestingly, this observation still holds in Figs.~\subref*{fg:singular_pareto_sx32} -- \subref*{fg:singular_pareto_sx128} where the direct channel is \emph{not} negligible.
			It is a pity that we could not provide a formal proof on this due to the non-trivial solution structures.
			The shape of the singular value region depends heavily on the relative strength of the indirect channels, which increases with $N_\mathrm{S}$ from the baseline $\Lambda_\mathrm{F}\Lambda_\mathrm{B}/\Lambda_\mathrm{D}=\qty{-35}{dB}$.
			Fig.~\subref*{fg:singular_pareto_sx32} shows that a 32-element \gls{ris} is insufficient to compensate this imbalance and results in a limited singular value region that is symmetric around the direct point.
			As the group size $L$ increases, the shape of the region evolves from elliptical to square.
			This transformation not only improves the dynamic range of $\sigma_1(\mathbf{H})$ and $\sigma_2(\mathbf{H})$ by \qty{22}{\percent} and \qty{38}{\percent} respectively, but also provides a better tradeoff in manipulating both singular values.
			The observation verifies that the design flexibility of \gls{bd}-\gls{ris} allows better alignment of multiple modes simultaneously.
			As a consequence, {the optimally shaped channels for power gain, communication, and power transfer coincide,} implying that {a fully-connected \gls{bd}-\gls{ris} may be designed in closed-form for simultaneous multi-functional optimality.}
			The singular value region also enlarges as the number of scattering elements $N_\mathrm{S}$ increases.
			In particular, Fig.~\subref*{fg:singular_pareto_sx128} shows that the equivalent channel can be completely nulled (corresponding to the origin) by a 128-element \gls{bd}-\gls{ris} but not by a diagonal one.
			The effect may be leveraged for interference cancellation and covert communication.
			Those results demonstrate the superior channel shaping capability of \gls{bd}-\gls{ris} and emphasizes the importance of adding reconfigurable components between \gls{ris} elements.
		\end{subsubsection}

		\begin{subsubsection}{Analytical Bounds and Numerical Results}
			\begin{figure}[!t]
				\centering
				\subfloat[$4 \times 32 \times 4$ (rank-1)]{
					\label{fg:singular_bound_rank1_sx32}
					\resizebox{!}{3.25cm}{
						\input{assets/simulation/pc_singular_bound_rank1_sx32.tex}
					}
				}
				\subfloat[$4 \times 64 \times 4$ (rank-1)]{
					\label{fg:singular_bound_rank1_sx64}
					\resizebox{!}{3.25cm}{
						\input{assets/simulation/pc_singular_bound_rank1_sx64.tex}
					}
				}
				\\
				\subfloat[$4 \times 128 \times 4$ (rank-2)]{
					\label{fg:singular_bound_rank2_sx128}
					\resizebox{!}{3.25cm}{
						\input{assets/simulation/pc_singular_bound_rank2_sx128.tex}
					}
				}
				\subfloat[$4 \times 256 \times 4$ (rank-4)]{
					\label{fg:singular_bound_rank4_sx256}
					\resizebox{!}{3.25cm}{
						\input{assets/simulation/pc_singular_bound_rank4_sx256.tex}
					}
				}
				\caption{
					Achievable channel singular values: analytical bounds (lines) and numerical results (bars).
					Baselines of bars denote the singular values of the direct channel.
					Blue (resp. red) bars denote the lower (resp. upper) dynamic range of singular values obtained by solving \eqref{op:pareto} with $\rho_n/\rho_{n'} \to 0$ (resp. $\to \infty$), $\forall n, \ n' \ne n$.
					`D' means \gls{d}-\gls{ris} and `BD' refers to fully-connected \gls{bd}-\gls{ris}.
					`rank-$k$' refers to the rank of the forward channel.
				}
				\label{fg:singular_bound}
			\end{figure}
			We focus on achieving the asymptotic bounds in Proposition~\ref{pp:rd} by finite $N_\mathrm{S}$, since most results from Proposition~\ref{pp:nd} are supplied with closed-form \gls{ris} solutions.
			For a rank-$k$ forward channel, Fig.~\ref{fg:singular_bound} compares the individual singular value bounds in Proposition~\ref{pp:rd} and the numerical results obtained by solving problem \eqref{op:pareto} with proper weights.
			When the \gls{ris} is in the \gls{los} of the transmitter, Figs.~\subref*{fg:singular_bound_rank1_sx32} and \subref*{fg:singular_bound_rank1_sx64} show that the achievable channel singular values indeed satisfy Corollary~\ref{co:los}, namely $\sigma_1(\mathbf{H}) \ge \sigma_1(\mathbf{T})$, $\sigma_2(\mathbf{T}) \le \sigma_2(\mathbf{H}) \le \sigma_1(\mathbf{T})$, etc.
			It is obvious that \gls{bd}-\gls{ris} can approach those bounds better than \gls{d}-\gls{ris} with a small $N_\mathrm{S}$.
			Another example is given in Fig.~\subref*{fg:singular_bound_rank2_sx128} with rank-2 forward channel.
			The first two channel singular values are unbounded above and bounded below by the first two singular values of $\mathbf{T}$, while the last two singular values can be suppressed to zero and bounded above by the first two singular values of $\mathbf{T}$.
			Those observations align with Proposition~\ref{pp:rd}.
			Finally, Fig.~\subref*{fg:singular_bound_rank4_sx256} confirms there are no extra singular value bounds when both backward and forward channels are full-rank.
			This can be predicted from \eqref{eq:auxiliary_rd} where $\mathbf{V}_\mathrm{F}$ becomes unitary and $\mathbf{T}=\mathbf{0}$.
			The numerical results are consistent with the analytical bounds, and we conclude that the channel shaping advantage of \gls{bd}-\gls{ris} over \gls{d}-\gls{ris} scales with the rank of backward and forward channels.

			\begin{figure}[!t]
				\centering
				\subfloat[$1 \times 256 \times 1$ (no direct)]{
					\label{fg:power_bond_txrx1_nd}
					\resizebox{!}{3.25cm}{
						\input{assets/simulation/pc_power_bond_txrx1_nd.tex}
					}
				}
				\subfloat[$4 \times 256 \times 4$ (no direct)]{
					\label{fg:power_bond_txrx4_nd}
					\resizebox{!}{3.25cm}{
						\input{assets/simulation/pc_power_bond_txrx4_nd.tex}
					}
				}
				\caption{
					Average maximum channel power gain versus \gls{bd}-\gls{ris} group size and \gls{mimo} dimensions.
					`Cascaded' refers to the upper bound in \eqref{iq:power_nd}.
				}
				\label{fg:power_bond}
			\end{figure}

			Fig.~\ref{fg:power_bond} compares the analytical bounds on the channel power gain in Corollary~\ref{co:nd_power} and the numerical results obtained by solving problem \eqref{op:power} when the direct channel is negligible.
			Here, a fully-connected \gls{bd}-\gls{ris} can attain the upper bound either in closed form \eqref{eq:ris_nd_power_max} or via optimization approach \eqref{eq:ris_power}.
			For the \gls{siso} case in Fig.~\subref*{fg:power_bond_txrx1_nd}, the maximum channel power gain is approximately \num{4e-6} by \gls{d}-\gls{ris} and \num{6.5e-6} by fully-connected \gls{bd}-\gls{ris}, corresponding to a \qty{62.5}{\percent} gain.
			It comes purely from branch matching as discussed in Example~\ref{eg:siso} and confirms the asymptotic power scaling law derived in \cite[(30)]{Shen2020a}
			Interestingly, Fig.~\subref*{fg:power_bond_txrx4_nd} shows that this relative gain, inferrable from the expectation analysis \eqref{eq:power_nd_rayleigh}, surges to \qty{270}{\percent} in $N_\mathrm{T}=N_\mathrm{R}=4$ \gls{mimo}.
			We thus conclude that the power gain of \gls{bd}-\gls{ris} scales with the group size and \gls{mimo} dimensions.
		\end{subsubsection}
	\end{subsection}

	\begin{subsection}{Achievable Rate Maximization}
		\begin{figure}[!t]
			\centering
			\subfloat[$16 \times N_\mathrm{S} \times 16$ (no direct)]{
				\label{fg:power_sx_txrx16_nd}
				\resizebox{!}{3.15cm}{
					\input{assets/simulation/pc_power_sx_txrx16_nd.tex}
				}
			}
			\subfloat[$16 \times N_\mathrm{S} \times 16$]{
				\label{fg:power_sx_txrx16}
				\resizebox{!}{3.15cm}{
					\input{assets/simulation/pc_power_sx_txrx16.tex}
				}
			}
			\caption{
				Average maximum channel power gain versus \gls{ris} configuration.
				`Explicit' refers to the optimal solution \eqref{eq:ris_nd_power_max} when the direct channel is negligible.
				`OP-left' and `OP-right' refer to the suboptimal solutions, when the direct channel is significant, by lossy transformation \eqref{op:procrustes} where $\mathbf{\Theta}$ is to the left and right of the product, respectively.
			}
			\label{fg:power_sx}
		\end{figure}

		We first focus on the channel power gain problem \eqref{op:power}.
		Fig.~\ref{fg:power_sx} shows the maximum channel power gain under different \gls{ris} configurations.
		An interesting observation is that the relative power gain of \gls{bd}-\gls{ris} over \gls{d}-\gls{ris} is even larger when the direct channel is significant.
		As shown in Figs.~\subref*{fg:power_sx_txrx16_nd} and \subref*{fg:power_sx_txrx16}, a 64-element fully \gls{bd}-\gls{ris} can almost provide the same channel power gain as a 256-element \gls{d}-\gls{ris} when the direct channel is significant, but less so when it is negligible.
		This is because the mode alignment advantage of \gls{bd}-\gls{ris} becomes more pronounced when the modes of direct channel is taken into account.
		We also notice that the suboptimal solutions \eqref{eq:ris_procrustes} for fully-connected \gls{bd}-\gls{ris} by lossy transformation \eqref{op:procrustes} are very close to optimal especially for a large $N_\mathrm{S}$.

		\begin{figure}[!t]
			\centering
			\subfloat[$4 \times 128 \times 4$]{
				\label{fg:rate_beamforming}
				\resizebox{!}{3.25cm}{
					\input{assets/simulation/pc_rate_beamforming.tex}
				}
			}
			\subfloat[$N_\mathrm{T} \times 128 \times N_\mathrm{R}$]{
				\label{fg:rate_txrx}
				\resizebox{!}{3.25cm}{
					\input{assets/simulation/pc_rate_txrx.tex}
				}
			}
			\\
			\subfloat[$4 \times N_\mathrm{S} \times 4$]{
				\label{fg:rate_sx}
				\resizebox{!}{3.25cm}{
					\input{assets/simulation/pc_rate_sx.tex}
				}
			}
			\subfloat[$4 \times 128 \times 4$]{
				\label{fg:rate_kfactor}
				\resizebox{!}{3.25cm}{
					\input{assets/simulation/pc_rate_kfactor.tex}
				}
			}
			\caption{
				Average achievable rate versus \gls{mimo} and \gls{ris} configurations.
				The transmit power corresponds to a direct \gls{snr} of \num{-10} to \qty{30}{dB}.
				`Alternate' refers to the alternating optimization and `Decouple' refers to the low-complexity design.
				`D' means \gls{d}-\gls{ris} and `BD' refers to fully-connected \gls{bd}-\gls{ris}.
			}
			\label{fg:rate}
		\end{figure}

		Fig.~\ref{fg:rate} presents the achievable rate under different \gls{mimo} and \gls{ris} configurations.
		At a transmit power $P = \qty{10}{dB}$, Fig.~\subref*{fg:rate_beamforming} shows that introducing a 128-element \gls{d}-\gls{ris} to $N_\mathrm{T}=N_\mathrm{R}=4$ \gls{mimo} can improve the achievable rate from \qty{22.2}{bps/Hz} to \qty{29.2}{bps/Hz} ($+\qty{31.5}{\percent}$).
		A \gls{bd}-\gls{ris} of group size 4 and 128 can further elevate those to \qty{32.1}{bps/Hz} ($+\qty{44.6}{\percent}$) and \qty{34}{bps/Hz}  ($+\qty{53.2}{\percent}$), respectively.
		An interesting observation is that the rate gap between the optimal \gls{ao} approach in Section~\ref{sc:rate_ao} and the low-complexity shaping solution in Section~\ref{sc:rate_lc} narrows as group size $L$ increases and completely vanishes for a fully-connected \gls{bd}-\gls{ris}.
		This implies that joint beamforming designs may be decoupled with minimal performance degradation by first shaping the wireless channel and then optimizing the transceiver, which substantially simplifies the design.
		Figs.~\subref*{fg:rate_txrx} and \subref*{fg:rate_sx} also show that both \emph{absolute and relative} rate gains of \gls{bd}-\gls{ris} over \gls{d}-\gls{ris}
		increases with the number of transmit and receive antennas and scattering elements, especially at high \gls{snr}.
		For $N_\mathrm{S}=128$ and $P = \qty{20}{dB}$,
		the achievable rate ratio of \gls{bd}-\gls{ris} over \gls{d}-\gls{ris} is \num{1.04}, \num{1.11}, and \num{1.13} for $N_\mathrm{T}=N_\mathrm{R}=1$, \num{4}, and \num{16}, respectively.
		For $N_\mathrm{T}=N_\mathrm{R}=4$ and $P = \qty{20}{dB}$, this ratio amounts to \num{1.03}, \num{1.08}, and \num{1.13} for $N_\mathrm{S}=16$, \num{64}, and \num{256}, respectively.
		Those observations align with the power gain results in Fig.~\ref{fg:power_sx} and highlight the rate benefits of \gls{bd}-\gls{ris} over \gls{d}-\gls{ris} in large-scale \gls{mimo} systems.
		In the low power regime (\num{-20} to \qty{-10}{dB}), we also notice that the slope of the achievable rate of \gls{bd}-\gls{ris} is steeper than that of \gls{d}-\gls{ris}.
		That is, \gls{bd}-\gls{ris} can help to activate more streams and achieve the asymptotic \gls{dof} at a low transmit \gls{snr}.
		This is particularly visible in Fig.~\subref*{fg:rate_sx} where the topmost curve is almost a linear function of the transmit power.
		It can be predicted from Fig.~\ref{fg:singular_pareto} that \gls{bd}-\gls{ris} can significantly enlarge all channel singular values for higher receive \gls{snr}.
		Finally, Fig.~\subref*{fg:rate_kfactor} shows that the gap between \gls{d}- and \gls{bd}-\gls{ris} narrows as the Rician K-factor increases and becomes indistinguishable in \gls{los} environment.
		The observation is expected from previous studies \cite{Shen2020a,Nerini2023} and aligns with Corollary~\ref{co:los}, which suggests that the \gls{bd}-\gls{ris} should be deployed in rich-scattering environments to exploit its channel shaping potential.

	\end{subsection}

	\begin{subsection}{Practical Constraints}
		\begin{subsubsection}{RIS Symmetry}
			\label{sc:ris_symmetry}
			\begin{figure}[!t]
				\centering
				\subfloat[Power: $2 \times N_\mathrm{S} \times 2$]{
					\label{fg:pc_power_symmetry}
					\resizebox{!}{3.15cm}{
						\input{assets/simulation/pc_power_symmetry.tex}
					}
				}
				\subfloat[Rate: $4 \times N_\mathrm{S} \times 4$]{
					\label{fg:pc_rate_symmetry}
					\resizebox{!}{3.15cm}{
						\input{assets/simulation/pc_rate_symmetry.tex}
					}
				}
				\caption{
					Impact of \gls{ris} symmetry on the \gls{mimo} power gain and achievable rate.
				}
			\end{figure}
			Symmetric \gls{ris} satisfying $\mathbf{\Theta} = \mathbf{\Theta}^\mathsf{T}$ are often considered in the literature due to hardware constraints.
			This study aim to investigate the impact of \gls{ris} symmetry on the system performance.
			\begin{remark}
				All proposed asymmetric \gls{bd}-\gls{ris} solutions are readily modifiable for symmetry. In particular,
				\begin{enumerate}[label=(\roman*)]
					\item \emph{\gls{svd}-based (e.g., \eqref{eq:ris_nd_sv_indl}, \eqref{eq:ris_nd_power}, \eqref{eq:ris_nd_rate_max}, \eqref{eq:ris_power}, \eqref{eq:ris_procrustes}):} Those closed-form asymmetric solutions are constructed from the product of singular matrices. If symmetry is required, one can replace the $\mathbf{U}, \mathbf{V}^\mathsf{H}$ in the \gls{svd} of $\mathbf{A} = \mathbf{U} \mathbf{\Sigma} \mathbf{V}^\mathsf{H}$ by $\mathbf{Q}, \mathbf{Q}^\mathsf{T}$ in the Autonne-Takagi factorization \cite{Ikramov2012} of $\frac{\mathbf{A} + \mathbf{A}^\mathsf{T}}{2} = \mathbf{Q} \mathbf{\Sigma} \mathbf{Q}^\mathsf{T}$ to construct $\mathbf{\Theta}$; \label{it:takagi}
					\item \emph{\gls{rcg}-based (e.g., \eqref{eq:shaping_subdiff}, \eqref{eq:gradient_eucl_rate}):} The symmetry constraint is added to the corresponding optimization problems, and one can project the solution to the nearest symmetric point $\mathbf{\Theta} \gets \frac{\mathbf{\Theta} + \mathbf{\Theta}^\mathsf{T}}{2}$ after each iteration. \label{it:projection}
				\end{enumerate}
			\end{remark}

			Figs.~\subref*{fg:pc_power_symmetry} and \subref*{fg:pc_rate_symmetry} compare the power gain and achievable rate of \gls{mimo} point-to-point channel under asymmetric and various symmetric \gls{ris} configurations.
			Here, `asymmetric' refers to the benchmark solution by \eqref{eq:ris_power} or \eqref{eq:gradient_eucl_rate}, `enforced' refers to enforcing symmetry on above, `legacy' refers to a straightforward extension of the single-mode \gls{snr}-optimal solution \cite[(6)]{Santamaria2023}, `Takagi' refers to the modification \ref{it:takagi}, and `projection' refers to the modification \ref{it:projection}.
			We observe that the performance gaps between the asymmetric and symmetric \gls{ris} configurations are insignificant and tends to widen with the number of scattering elements.
			The two proposed modifications also outperform other candidates in both problems.
		\end{subsubsection}

		\begin{subsubsection}{Channel Estimation Error}
			\label{sc:estimation_error}
			\begin{figure}[!t]
				\centering
				\subfloat[Pareto: $2 \times 64 \times 2$]{
					\label{fg:pc_singular_csi}
					\resizebox{!}{3.15cm}{
						\input{assets/simulation/pc_singular_csi.tex}
					}
				}
				\subfloat[Rate: $4 \times 128 \times 4$]{
					\label{fg:pc_rate_csi}
					\resizebox{!}{3.15cm}{
						\input{assets/simulation/pc_rate_csi.tex}
					}
				}
				\caption{
					Impact of \gls{ris} channel estimation error on the \gls{mimo} singular value region and achievable rate.
					A higher transparency of the Pareto frontier indicates a larger channel estimation error.
					`D' means \gls{d}-\gls{ris} and `BD' refers to fully-connected \gls{bd}-\gls{ris}.
				}
			\end{figure}

			Figs.~\subref*{fg:pc_singular_csi} and \subref*{fg:pc_rate_csi} investigates how \gls{ris} channel estimation errors affect the system performance in terms of singular value region and achievable rate.
			We assume the direct channel can be perfectly acquired and the estimated backward and forward channels are modeled by
			\begin{equation*}
				\hat{\mathbf{H}}_{\mathrm{B/F}} = \mathbf{H}_{\mathrm{B/F}} + \tilde{\mathbf{H}}_{\mathrm{B/F}},
			\end{equation*}
			where the error follows $\mathrm{vec}(\tilde{\mathbf{H}}_{\mathrm{B/F}}) \sim \mathcal{N}_\mathbb{C}(\mathbf{0}, \epsilon \Lambda_\mathrm{B} \Lambda_\mathrm{F}\mathbf{I})$.
			The results are evaluated over the ground truth channels.
			It is observed that the proposed channel shaping and joint beamforming solutions are reasonably robust to channel estimation errors.
			An interesting observation is that a \gls{bd}-\gls{ris} designed over extremely poorly estimated channels ($\epsilon = 0.5$) may still outpeform a \gls{d}-\gls{ris} designed over almost perfectly estimated channels ($\epsilon = 0.01$).
			We hope those results can motivate further research on the robust shaping design and provide insights for practical \gls{bd}-\gls{ris} deployment.

		\end{subsubsection}
	\end{subsection}
\end{section}

\begin{section}{Conclusion}
	This paper investigated the capability of \gls{bd}-\gls{ris} to shape a \gls{mimo} channel in terms of singular values and their functions.
	We started from a toy example and derived some analytical bounds (with closed-form solutions) on the channel singular values, power gain, and capacity.
	An efficient framework was then proposed to optimize the \gls{bd}-\gls{ris} for a broader class of singular value functions.
	We also presented two beamforming designs for the rate maximization problem, one for optimal performance and the other exploits shaping implications for much lower complexity while remaining close-to-optimal.
	Extensive simulation show that the significant power and rate gains of \gls{bd}-\gls{ris} over \gls{d}-\gls{ris} stems from its superior \gls{mimo} branch matching and mode alignment potentials, which scales with the number of elements, group size, and \gls{mimo} dimensions.

	The analysis and optimization methods in this paper have been tailored for group-connected \gls{bd}-\gls{ris}.
	Extension to other \gls{ris} architectures remains a promising area for future research.
	One straightforward extension to the multi-sector model \cite{Li2023c} is to retrieve the optimal scattering matrix for each sector individually by methods in this paper and then play with the power splitting factors.
	Meanwhile, transitioning from single- to multi-layer \gls{ris} models \cite{An23b} mirrors that from single- to multi-hop \gls{af} relays; interested readers may be inspired by \cite{Rong2009a}.

	Finally, we remark that the principle of channel shaping is not limited to point-to-point \gls{mimo}.
	Algorithm \ref{ag:rcg} and the two solutions in Section~\ref{sc:rate} are readily extendable to weighted sum-rate maximization and leakage interference minimization in \gls{mimo} interference channel; please refer to our \href{https://github.com/snowztail/channel-shaping}{GitHub} for details.

\end{section}

\begin{appendix}
	\begin{subsection}{Proof of Proposition~\ref{pp:dof}}
		\label{ap:dof}
		It suffices to consider the rank of the indirect channel.
		Denote the \gls{svd} of the backward and forward channels as
		\begin{equation*}
			\mathbf{H}_\mathrm{B/F} = \begin{bmatrix}
				\mathbf{U}_{\mathrm{B/F},1} & \mathbf{U}_{\mathrm{B/F},2}
			\end{bmatrix}
			\begin{bmatrix}
				\mathbf{\Sigma}_{\mathrm{B/F},1} & \mathbf{0} \\
				\mathbf{0}                       & \mathbf{0}
			\end{bmatrix}
			\begin{bmatrix}
				\mathbf{V}_{\mathrm{B/F},1}^\mathsf{H} \\
				\mathbf{V}_{\mathrm{B/F},2}^\mathsf{H}
			\end{bmatrix},
		\end{equation*}
		where $\mathbf{U}_{\mathrm{B/F},1}$ and $\mathbf{V}_{\mathrm{B/F},1}$ are any left and right singular matrices of $\mathbf{H}_\mathrm{B/F}$ corresponding to non-zero singular values $\mathbf{\Sigma}_{\mathrm{B/F},1}$, and $\mathbf{U}_{\mathrm{B/F},2}$ and $\mathbf{V}_{\mathrm{B/F},2}$ are those corresponding to zero singular values.
		The rank of the indirect channel is \cite[(16.5.10.b)]{Hogben2013}
		\begin{equation*}
			\begin{split}
				\rank(\mathbf{H}_\mathrm{B} \mathbf{\Theta} \mathbf{H}_\mathrm{F})
				 & = \rank(\mathbf{H}_\mathrm{B}) - \dim \bigl(\ker(\mathbf{H}_\mathrm{F}^\mathsf{H} \mathbf{\Theta}^\mathsf{H}) \cap \ran(\mathbf{H}_\mathrm{B}^\mathsf{H})\bigr) \\
				 & = \rank(\mathbf{H}_\mathrm{B}) - \dim \bigl(\ran(\mathbf{\Theta} \mathbf{U}_{\mathrm{F},2}) \cap \ran(\mathbf{V}_{\mathrm{B},1})\bigr)                          \\
				 & \triangleq r_\mathrm{B} - r_\mathrm{L}(\mathbf{\Theta}),
			\end{split}
		\end{equation*}
		where we define $r_\mathrm{L}(\mathbf{\Theta}) \triangleq \dim \bigl(\ran(\mathbf{\Theta} \mathbf{U}_{\mathrm{F},2}) \cap \ran(\mathbf{V}_{\mathrm{B},1})\bigr)$ and $r_\mathrm{B/F} \triangleq \rank(\mathbf{H}_\mathrm{B/F})$.
		Since $\mathbf{U}_{\mathrm{F},2} \in \mathbb{U}^{N_\mathrm{S} \times (N_\mathrm{S} - r_\mathrm{F})}$ and $\mathbf{V}_{\mathrm{B},1} \in \mathbb{U}^{N_\mathrm{S} \times r_\mathrm{B}}$, we have $\max(r_\mathrm{B} - r_\mathrm{F}, 0) \le r_\mathrm{L}(\mathbf{\Theta}) \le \min(N_\mathrm{S} - r_\mathrm{F}, r_\mathrm{B})$
		and thus
		\begin{equation}
			\label{iq:rank_indirect}
			\max(r_\mathrm{B} + r_\mathrm{F} - N_\mathrm{S}, 0) \le \rank(\mathbf{H}_\mathrm{B} \mathbf{\Theta} \mathbf{H}_\mathrm{F}) \le \min(r_\mathrm{B}, r_\mathrm{F}).
		\end{equation}
		To attain the upper bound in \eqref{iq:rank_indirect}, the \gls{ris} needs to minimize $r_\mathrm{L}(\mathbf{\Theta})$ by aligning the ranges of $\mathbf{\Theta} \mathbf{U}_{\mathrm{F},2}$ and $\mathbf{V}_{\mathrm{B},2}$ as much as possible.
		This is achieved by
		\begin{equation}
			\label{eq:ris_dof_max}
			\mathbf{\Theta}_{\textnormal{DoF-max}}^{\textnormal{MIMO}} = \mathbf{Q}_{\mathrm{B},2} \mathbf{Q}_{\mathrm{F},2}^\mathsf{H},
		\end{equation}
		where $\mathbf{Q}_{\mathrm{B},2}$ and $\mathbf{Q}_{\mathrm{F},2}$ are the unitary matrices of the QR decomposition of $\mathbf{V}_{\mathrm{B},2}$ and $\mathbf{U}_{\mathrm{F},2}$, respectively.
		Similarly, the lower bound in \eqref{iq:rank_indirect} is attained at
		\begin{equation}
			\label{eq:ris_dof_min}
			\mathbf{\Theta}_{\textnormal{DoF-min}}^{\textnormal{MIMO}} = \mathbf{Q}_{\mathrm{B},1} \mathbf{Q}_{\mathrm{F},2}^\mathsf{H},
		\end{equation}
		where $\mathbf{Q}_{\mathrm{B},1}$ is the unitary matrix of the QR decomposition of $\mathbf{V}_{\mathrm{B},1}$.
		While the \gls{dof}-optimal structures \eqref{eq:ris_dof_max} and \eqref{eq:ris_dof_min} are always feasible for fully-connected \gls{bd}-\gls{ris}, they are generally infeasible for \gls{d}-\gls{ris} unless there exist some QR decomposition that diagonalize $\mathbf{Q}_{\mathrm{B},2} \mathbf{Q}_{\mathrm{F},2}^\mathsf{H}$ and $\mathbf{Q}_{\mathrm{B},1} \mathbf{Q}_{\mathrm{F},2}^\mathsf{H}$ simultaneously.
		That is, \gls{bd}-\gls{ris} may achieve a larger or smaller number of \gls{dof} of indirect channel, and thus equivalent channel, than \gls{d}-\gls{ris}.
	\end{subsection}

	\begin{subsection}{Proof of Proposition~\ref{pp:rd}}
		\label{ap:rank_deficient}
		We consider rank-$k$ forward channel and the proof follows similarly for rank-$k$ backward channel.
		Let $\mathbf{H}_\mathrm{F} = \mathbf{U}_\mathrm{F} \mathbf{\Sigma}_\mathrm{F} \mathbf{V}_\mathrm{F}^\mathsf{H}$ be the \gls{svd} of the forward channel.
		The channel Gram matrix $\mathbf{G} \triangleq \mathbf{H} \mathbf{H}^\mathsf{H} $ can be written as
		\begin{equation*}
			\begin{split}
				\mathbf{G}
				 & = \mathbf{H}_\mathrm{D} \mathbf{H}_\mathrm{D}^\mathsf{H} + \mathbf{H}_\mathrm{B} \mathbf{\Theta} \mathbf{U}_\mathrm{F} \mathbf{\Sigma}_\mathrm{F} \mathbf{\Sigma}_\mathrm{F}^\mathsf{H} \mathbf{U}_\mathrm{F}^\mathsf{H} \mathbf{\Theta}^\mathsf{H} \mathbf{H}_\mathrm{B}^\mathsf{H}                                                         \\
				 & \quad + \mathbf{H}_\mathrm{B} \mathbf{\Theta} \mathbf{U}_\mathrm{F} \mathbf{\Sigma}_\mathrm{F} \mathbf{V}_\mathrm{F}^\mathsf{H} \mathbf{H}_\mathrm{D}^\mathsf{H} + \mathbf{H}_\mathrm{D} \mathbf{V}_\mathrm{F} \mathbf{\Sigma}_\mathrm{F} \mathbf{U}_\mathrm{F}^\mathsf{H} \mathbf{\Theta}^\mathsf{H} \mathbf{H}_\mathrm{B}^\mathsf{H}       \\
				 & = \mathbf{H}_\mathrm{D} (\mathbf{I} - \mathbf{V}_\mathrm{F} \mathbf{V}_\mathrm{F}^\mathsf{H}) \mathbf{H}_\mathrm{D}^\mathsf{H}                                                                                                                                                                                                               \\
				 & \quad + (\mathbf{H}_\mathrm{B} \mathbf{\Theta} \mathbf{U}_\mathrm{F} \mathbf{\Sigma}_\mathrm{F} + \mathbf{H}_\mathrm{D} \mathbf{V}_\mathrm{F}) (\mathbf{\Sigma}_\mathrm{F} \mathbf{U}_\mathrm{F}^\mathsf{H} \mathbf{\Theta}^\mathsf{H} \mathbf{H}_\mathrm{B}^\mathsf{H} + \mathbf{V}_\mathrm{F}^\mathsf{H} \mathbf{H}_\mathrm{D}^\mathsf{H}) \\
				 & = \mathbf{Y} + \mathbf{Z} \mathbf{Z}^\mathsf{H},
			\end{split}
		\end{equation*}
		where we define $\mathbf{Y} \triangleq \mathbf{H}_\mathrm{D} (\mathbf{I} - \mathbf{V}_\mathrm{F} \mathbf{V}_\mathrm{F}^\mathsf{H}) \mathbf{H}_\mathrm{D}^\mathsf{H} \in \mathbb{H}^{N_\mathrm{R} \times N_\mathrm{R}}$ and $\mathbf{Z} \triangleq \mathbf{H}_\mathrm{B} \mathbf{\Theta} \mathbf{U}_\mathrm{F} \mathbf{\Sigma}_\mathrm{F} + \mathbf{H}_\mathrm{D} \mathbf{V}_\mathrm{F} \in \mathbb{C}^{N_\mathrm{R} \times k}$.
		That is to say, $\mathbf{G}$ can be expressed as a Hermitian matrix plus $k$ rank-1 perturbations.
		According to the Cauchy interlacing formula \cite[Theorem 8.4.3]{Golub2013}, the $n$-th eigenvalue of $\mathbf{G}$ is bounded by
		\begin{align}
			\lambda_n(\mathbf{G}) & \le \lambda_{n-k}(\mathbf{Y}), &  & \text{if } n > k, \label{iq:ev_rd_max}         \\
			\lambda_n(\mathbf{G}) & \ge \lambda_n(\mathbf{Y}),     &  & \text{if } n < N - k + 1 \label{iq:ev_rd_min}.
		\end{align}
		Since $\mathbf{Y} = \mathbf{T} \mathbf{T}^\mathsf{H}$ is positive semi-definite, taking the square roots of \eqref{iq:ev_rd_max} and \eqref{iq:ev_rd_min} gives \eqref{iq:sv_rd_max} and \eqref{iq:sv_rd_min}.
	\end{subsection}

	\begin{subsection}{Proof of Proposition~\ref{pp:nd}}
		\label{ap:nd}
		Let $\mathbf{H}_\mathrm{B} = \mathbf{U}_\mathrm{B} \mathbf{\Sigma}_\mathrm{B} \mathbf{V}_\mathrm{B}^\mathsf{H}$ and $\mathbf{H}_\mathrm{F} = \mathbf{U}_\mathrm{F} \mathbf{\Sigma}_\mathrm{F} \mathbf{V}_\mathrm{F}^\mathsf{H}$ be the \gls{svd} of the backward and forward channels, respectively.
		The scattering matrix of fully-connected \gls{bd}-\gls{ris} can be decomposed as
		\begin{equation}
			\label{eq:ris_decompose}
			\mathbf{\Theta} = \mathbf{V}_\mathrm{B} \mathbf{X} \mathbf{U}_\mathrm{F}^\mathsf{H},
		\end{equation}
		where $\mathbf{X} \in \mathbb{U}^{N_\mathrm{S} \times N_\mathrm{S}}$ is a unitary matrix to be designed.
		The equivalent channel is thus a function of $\mathbf{X}$
		\begin{equation}
			\label{eq:channel_nd}
			\mathbf{H} = \mathbf{H}_\mathrm{B} \mathbf{\Theta} \mathbf{H}_\mathrm{F} = \mathbf{U}_\mathrm{B} \mathbf{\Sigma}_\mathrm{B} \mathbf{X} \mathbf{\Sigma}_\mathrm{F} \mathbf{V}_\mathrm{F}^\mathsf{H}.
		\end{equation}
		Since $\sv(\mathbf{U} \mathbf{A} \mathbf{V}^\mathsf{H}) = \sv(\mathbf{A})$ for unitary $\mathbf{U}$ and $\mathbf{V}$, we have
		\begin{equation}
			\label{eq:sv_factor}
			\begin{split}
				\sv(\mathbf{H}) & = \sv(\mathbf{U}_\mathrm{B} \mathbf{\Sigma}_\mathrm{B} \mathbf{X} \mathbf{\Sigma}_\mathrm{F} \mathbf{V}_\mathrm{F}^\mathsf{H})                                                                     \\
				                & = \sv(\mathbf{\Sigma}_\mathrm{B} \mathbf{X} \mathbf{\Sigma}_\mathrm{F})                                                                                                                            \\
				                & = \sv(\bar{\mathbf{U}}_\mathrm{B} \mathbf{\Sigma}_\mathrm{B} \mathbf{\bar{V}}_\mathrm{B}^\mathsf{H} \bar{\mathbf{U}}_\mathrm{F} \mathbf{\Sigma}_\mathrm{F} \mathbf{\bar{V}}_\mathrm{F}^\mathsf{H}) \\
				                & = \sv(\mathbf{BF}),
			\end{split}
		\end{equation}
		where $\bar{\mathbf{U}}_{\mathrm{B}} \in \mathbb{U}^{N_\mathrm{R} \times N_\mathrm{R}}$, $\bar{\mathbf{V}}_\mathrm{B}, \bar{\mathbf{U}}_\mathrm{F} \in \mathbb{U}^{N_\mathrm{S} \times N_\mathrm{S}}$, and $\bar{\mathbf{V}}_\mathrm{F} \in \mathbb{U}^{N_\mathrm{T} \times N_\mathrm{T}}$ can be designed arbitrarily.
	\end{subsection}

	\begin{subsection}{Proof of Corollary~\ref{co:nd_sv_prod_tail}}
		\label{ap:nd_sv_prod_tail}
		\eqref{iq:sv_nd_prod_largest} follows from \eqref{iq:horn} when $r = k$.
		On the other hand, if we can prove
		\begin{equation}
			\label{eq:sv_prod_all_ext}
			\prod_{n=1}^{\bar{N}} \sigma_n(\mathbf{H}) = \prod_{n=1}^{\bar{N}} \sigma_n(\mathbf{H}_\mathrm{B}) \sigma_n(\mathbf{H}_\mathrm{F}),
		\end{equation}
		then \eqref{iq:sv_nd_prod_smallest} follows from \eqref{iq:sv_nd_prod_largest} and the non-negativity of singular values.
		To see \eqref{eq:sv_prod_all_ext}, we start from a stricter result
		\begin{equation}
			\label{eq:sv_product_all}
			\prod_{n=1}^{N_\mathrm{S}} \sigma_n(\mathbf{H}) = \prod_{n=1}^{N_\mathrm{S}} \sigma_n(\mathbf{H}_\mathrm{B}) \sigma_n(\mathbf{H}_\mathrm{F}),
		\end{equation}
		which is provable by cases.
		When $N_\mathrm{S} > N$, both sides of \eqref{eq:sv_product_all} become zero since $\sigma_n(\mathbf{H}) = \sigma_n(\mathbf{H}_\mathrm{B}) = \sigma_n(\mathbf{H}_\mathrm{F}) = 0$ for $n > N$.
		When $N_\mathrm{S} \le N$, we have
		\begin{equation*}
			\begin{split}
				\prod\nolimits_{n=1}^{N_\mathrm{S}} \sigma_n(\mathbf{H})
				 & = \prod\nolimits_{n=1}^{N_\mathrm{S}} \sigma_n(\mathbf{\Sigma}_\mathrm{B} \mathbf{X} \mathbf{\Sigma}_\mathrm{F})             \\
				 & = \prod\nolimits_{n=1}^{N_\mathrm{S}} \sigma_n(\hat{\mathbf{\Sigma}}_\mathrm{B} \mathbf{X} \hat{\mathbf{\Sigma}}_\mathrm{F}) \\
				 & = \det\bigl(\hat{\mathbf{\Sigma}}_\mathrm{B} \mathbf{X} \hat{\mathbf{\Sigma}}_\mathrm{F}\bigr)                               \\
				 & = \det\bigl(\hat{\mathbf{\Sigma}}_\mathrm{B}\bigr) \det(\mathbf{X}) \det\bigl(\hat{\mathbf{\Sigma}}_\mathrm{F}\bigr)         \\
				 & = \prod\nolimits_{n=1}^{N_\mathrm{S}} \sigma_n(\mathbf{\Sigma}_\mathrm{B}) \sigma_n(\mathbf{\Sigma}_\mathrm{F}),
			\end{split}
		\end{equation*}
		where the first equality follows from \eqref{eq:sv_factor} and $\hat{\mathbf{\Sigma}}_\mathrm{B}, \hat{\mathbf{\Sigma}}_\mathrm{F}$ truncate $\mathbf{\Sigma}_\mathrm{B}, \mathbf{\Sigma}_\mathrm{F}$ to square matrices of dimension $N_\mathrm{S}$, respectively.
		It is evident that \eqref{eq:sv_product_all} implies \eqref{eq:sv_prod_all_ext} and thus \eqref{iq:sv_nd_prod_smallest}.
	\end{subsection}

	\begin{subsection}{Proof of Corollary~\ref{co:nd_sv_indl}}
		\label{ap:nd_sv_indl}
		In \eqref{iq:sv_nd_indl}, the set of upper bounds
		\begin{equation}
			\label{iq:sv_nd_indl_set_max}
			\bigl\{\sigma_n(\mathbf{H}) \le \sigma_i(\mathbf{H}_\mathrm{B}) \sigma_j(\mathbf{H}_\mathrm{F}) \mid [i,j,k] \in [N_\mathrm{S}]^3, i+j=n+1\bigr\}
		\end{equation}
		is a special case of \eqref{iq:horn} with $(I, J, K) \in [N_\mathrm{S}]^3$.
		The minimum\footnote{One may think to take the maximum of those upper bounds as the problem of interest is the attainable dynamic range of $n$-th singular value. This is infeasible since the singular values will be reordered.} of \eqref{iq:sv_nd_indl_set_max} is selected as the tightest upper bound in \eqref{iq:sv_nd_indl}.
		On the other hand, the set of lower bounds
		\begin{equation}
			\label{iq:sv_nd_indl_set_min}
			\bigl\{\sigma_n(\mathbf{H}) \ge \sigma_i(\mathbf{H}_\mathrm{B}) \sigma_j(\mathbf{H}_\mathrm{F}) \mid [i,j,k] \in [N_\mathrm{S}]^3, i+j=n+N_\mathrm{S}\bigr\}
		\end{equation}
		can be induced by \eqref{iq:sv_nd_indl_set_max}, \eqref{eq:sv_product_all}, and the non-negativity of singular values.
		The maximum of \eqref{iq:sv_nd_indl_set_min} is selected as the tightest lower bound in \eqref{iq:sv_nd_indl}.
		Interested readers are also referred to \cite[(2.0.3)]{Zhang2005}.

		To attain the upper bound, the \gls{bd}-\gls{ris} needs to maximize the minimum of the first $n$ channel singular values.
		It follows from \eqref{eq:ris_nd_sv_indl_max} that
		\begin{equation*}
			\begin{split}
				\sv(\mathbf{H})
				 & = \sv(\mathbf{H}_\mathrm{B} \mathbf{V}_\mathrm{B} \mathbf{P} \mathbf{U}_\mathrm{F}^\mathsf{H} \mathbf{H}_\mathrm{F})                                                                                                                         \\
				 & = \sv(\mathbf{U}_\mathrm{B} \mathbf{\Sigma}_\mathrm{B} \mathbf{V}_\mathrm{B}^\mathsf{H} \mathbf{V}_\mathrm{B} \mathbf{P} \mathbf{U}_\mathrm{F}^\mathsf{H} \mathbf{U}_\mathrm{F} \mathbf{\Sigma}_\mathrm{F} \mathbf{U}_\mathrm{F}^\mathsf{H}) \\
				 & = \sv(\mathbf{\Sigma}_\mathrm{B} \mathbf{P} \mathbf{\Sigma}_\mathrm{F}).
			\end{split}
		\end{equation*}
		On the one hand, $\mathbf{P}_{ij}=1$ with $(i, j)$ satisfying \eqref{eq:idx_nd_sv_indl_max} ensures $\min_{i+j=n+1} \sigma_i(\mathbf{H}_\mathrm{B}) \sigma_j(\mathbf{H}_\mathrm{F})$ is a singular value of $\mathbf{H}$.
		It is actually among the first $n$ since the number of pairs $(i',j')$ not majorized by $(i,j)$ is $n-1$.
		On the other hand, \eqref{eq:perm_nd_sv_indl_max} ensures the first $(n-1)$-th singular values are no smaller than $\min_{i+j=n+1} \sigma_i(\mathbf{H}_\mathrm{B}) \sigma_j(\mathbf{H}_\mathrm{F})$.
		Combining both facts, we claim the upper bound $\sigma_n(\mathbf{H}) = \min_{i+j=n+1} \sigma_i(\mathbf{H}_\mathrm{B}) \sigma_j(\mathbf{H}_\mathrm{F})$ is attainable by \eqref{eq:ris_nd_sv_indl_max}.
		The attainability of the lower bound can be proved similarly and the details are omitted.
	\end{subsection}

	\begin{subsection}{Proof of Corollary~\ref{co:nd_power}}
		\label{ap:nd_power}
		From \eqref{eq:ris_decompose} and \eqref{eq:channel_nd} we have
		\begin{equation}
			\begin{split}
				\lVert \mathbf{H} \rVert _\mathrm{F}^2
				 & = \tr \bigl(\mathbf{V}_\mathrm{F} \mathbf{\Sigma}_\mathrm{F}^\mathsf{H} \mathbf{X}^\mathsf{H} \mathbf{\Sigma}_\mathrm{B}^\mathsf{H} \mathbf{U}_\mathrm{B}^\mathsf{H} \mathbf{U}_\mathrm{B} \mathbf{\Sigma}_\mathrm{B} \mathbf{X} \mathbf{\Sigma}_\mathrm{F} \mathbf{V}_\mathrm{F}^\mathsf{H}\bigr) \\
				 & = \tr \bigl(\mathbf{\Sigma}_\mathrm{B}^\mathsf{H} \mathbf{\Sigma}_\mathrm{B} \cdot \mathbf{X} \mathbf{\Sigma}_\mathrm{F} \mathbf{\Sigma}_\mathrm{F}^\mathsf{H} \mathbf{X}^\mathsf{H}\bigr)                                                                                                         \\
				 & \triangleq \tr \bigl(\tilde{\mathbf{B}} \tilde{\mathbf{F}}\bigr),
			\end{split}
		\end{equation}
		where $\mathbf{X} \triangleq \mathbf{V}_\mathrm{B}^\mathsf{H} \mathbf{\Theta} \mathbf{U}_\mathrm{F} \in \mathbb{U}^{N_\mathrm{S} \times N_\mathrm{S}}$, $\tilde{\mathbf{B}} \triangleq \mathbf{\Sigma}_\mathrm{B}^\mathsf{H} \mathbf{\Sigma}_\mathrm{B} \in \mathbb{H}_+^{N_\mathrm{S} \times N_\mathrm{S}}$, and $\tilde{\mathbf{F}} \triangleq \mathbf{X} \mathbf{\Sigma}_\mathrm{F} \mathbf{\Sigma}_\mathrm{F}^\mathsf{H} \mathbf{X}^\mathsf{H} \in \mathbb{H}_+^{N_\mathrm{S} \times N_\mathrm{S}}$.
		By Ruhe's trace inequality for positive semi-definite matrices \cite[(H.1.g) and (H.1.h)]{Marshall2010},
		\begin{equation*}
			\sum_{n=1}^N \lambda_n(\tilde{\mathbf{B}}) \lambda_{N_\mathrm{S}-n+1}(\tilde{\mathbf{F}}) \le \tr \bigl(\tilde{\mathbf{B}} \tilde{\mathbf{F}}\bigr) \le \sum_{n=1}^N \lambda_n(\tilde{\mathbf{B}}) \lambda_n(\tilde{\mathbf{F}}),
		\end{equation*}
		which simplifies to \eqref{iq:power_nd}.
		The upper bound is attained when $\mathbf{X}$ is chosen to match the singular values of $\tilde{\mathbf{F}}$ to those of $\tilde{\mathbf{B}}$ in similar order.
		Apparently this occurs at $\mathbf{X} = \mathbf{I}$ and $\mathbf{\Theta} = \mathbf{V}_\mathrm{B} \mathbf{U}_\mathrm{F}^\mathsf{H}$.
		On the other hand, the lower bound is attained when the singular values of $\tilde{\mathbf{F}}$ and $\tilde{\mathbf{B}}$ are matched in reverse order, namely $\mathbf{X} = \mathbf{J}$ and $\mathbf{\Theta} = \mathbf{V}_\mathrm{B} \mathbf{J} \mathbf{U}_\mathrm{F}^\mathsf{H}$.
	\end{subsection}

	\begin{subsection}{Proof of Corollary~\ref{co:nd_capacity_snr_extreme}}
		\label{ap:nd_capacity}
		When perfect \gls{csi} is available at the transmitter, in the low-\gls{snr} regime, the capacity is achieved by dominant eigenmode transmission \cite[(5.26)]{Clerckx2013}
		\begin{align*}
			C_{\rho_\downarrow}
			 & = \log\bigl(1 + \rho \lambda_1(\mathbf{H}^\mathsf{H} \mathbf{H})\bigr)        \\
			 & = \log\bigl(1 + \rho \sigma_1^2(\mathbf{H})\bigr)                             \\
			 & \approx \rho \sigma_1^2(\mathbf{H})                                           \\
			 & \le \rho \sigma_1^2(\mathbf{H}_\mathrm{B}) \sigma_1^2(\mathbf{H}_\mathrm{F}),
		\end{align*}
		where the approximation is $\log(1 + x) \approx x$ for small $x$ and the inequality follows from \eqref{iq:sv_nd_prod_largest} with $k=1$.
		In the high-\gls{snr} regime, the capacity is achieved by multiple eigenmode transmission with uniform power location \cite[(5.27)]{Clerckx2013}
		\begin{align*}
			C_{\rho_\uparrow}
			 & = \sum\nolimits_{n=1}^N \log\Bigl(1 + \frac{\rho}{N} \lambda_n(\mathbf{H}^\mathsf{H} \mathbf{H})\Bigr)                     \\
			 & \approx \sum\nolimits_{n=1}^N \log\Bigl(\frac{\rho}{N} \sigma_n^2(\mathbf{H})\Bigr)                                        \\
			 & = N \log \frac{\rho}{N} + \sum\nolimits_{n=1}^N \log \sigma_n^2(\mathbf{H})                                                \\
			 & = N \log \frac{\rho}{N} + \log \prod\nolimits_{n=1}^N \sigma_n^2(\mathbf{H})                                               \\
			 & \le N \log \frac{\rho}{N} + 2 \log \prod\nolimits_{n=1}^N \sigma_n(\mathbf{H}_\mathrm{B}) \sigma_n(\mathbf{H}_\mathrm{F}),
		\end{align*}
		where the approximation is $\log(1 + x) \approx \log(x)$ for large $x$ and the inequality follows from \eqref{iq:sv_nd_prod_largest} with $k=N$.

		We now show \eqref{eq:ris_nd_rate_max} can achieve the upper bounds in \eqref{iq:capacity_nd_snr_low} and \eqref{iq:capacity_nd_snr_high} simultaneously.
		On the one hand, \eqref{eq:ris_nd_rate_max} is a special case of \eqref{eq:ris_nd_sv_indl_max} with $\mathbf{P} = \mathbf{I}$, which satisfies \eqref{eq:idx_nd_sv_indl_max} and \eqref{eq:perm_nd_sv_indl_max} for $n=1$ and thus attain $\sigma_1(\mathbf{H}) = \sigma_1(\mathbf{H}_\mathrm{B}) \sigma_1(\mathbf{H}_\mathrm{F})$.
		On the other hand, since $\log(\cdot)$ is a monotonic function, we can prove similar to Appendix \ref{ap:nd_power} that $\sum_{n=1}^N \log \sigma_n^2(\mathbf{H}) \le \sum_{n=1}^N \log \sigma_n^2(\mathbf{H}_\mathrm{B}) \sigma_n^2(\mathbf{H}_\mathrm{F})$ and the bound is tight at \eqref{eq:ris_nd_rate_max}.
		The proof is complete.
	\end{subsection}

	\begin{subsection}{Proof of Proposition~\ref{pp:shaping}}
		\label{ap:shaping}
		A straightforward extension to \cite[Theorem 2]{Watson1992} shows that the Clarke subdifferential of a locally Lipschitz function of singular values of a matrix with respect to the matrix itself is given by
		\begin{equation}
			\partial_{\mathbf{H}^*} f\bigl(\sv(\mathbf{H})\bigr) = \conv \bigl\{ \mathbf{U} \mathbf{D} \mathbf{V}^\mathsf{H} \bigr\},
		\end{equation}
		where $\mathbf{D} \in \mathbb{C}^{N_\mathrm{R} \times N_\mathrm{T}}$ is a rectangular diagonal matrix with $[\mathbf{D}]_{n,n} \in \partial_{\sigma_n(\mathbf{H})} f\bigl(\sv(\mathbf{H})\bigr)$, $\forall n \in [N]$, and $\mathbf{U}$, $\mathbf{V}$ are any left and right singular matrices of $\mathbf{H}$.
		Therefore,
		\begin{align*}
			\partial f\bigl(\sv(\mathbf{H})\bigr)
			 & \ni \tr \bigl(\mathbf{V}^* \mathbf{D}^\mathsf{T} \mathbf{U}^\mathsf{T} \partial \mathbf{H}^*\bigr)                                                                  \\
			 & = \tr \bigl(\mathbf{V}^* \mathbf{D}^\mathsf{T} \mathbf{U}^\mathsf{T} \mathbf{H}_{\mathrm{B},g}^* \partial {\mathbf{\Theta}_g^*} \mathbf{H}_{\mathrm{F},g}^*\bigr)   \\
			 & = \tr \bigl(\mathbf{H}_{\mathrm{F},g}^* \mathbf{V}^* \mathbf{D}^\mathsf{T} \mathbf{U}^\mathsf{T} \mathbf{H}_{\mathrm{B},g}^* \partial {\mathbf{\Theta}_g^*} \bigr),
		\end{align*}
		such that $\mathbf{H}_{\mathrm{B},g}^\mathsf{H} \mathbf{U} \mathbf{D} \mathbf{V}^\mathsf{H} \mathbf{H}_{\mathrm{F},g}^\mathsf{H}$ constitutes a Clarke subgradient of $f\bigl(\sv(\mathbf{H})\bigr)$ with respect to $\mathbf{\Theta}_g$.
		The convex hull of those subgradients is the subdifferential \eqref{eq:shaping_subdiff}.
	\end{subsection}

	\begin{subsection}{Proof of Lemma~\ref{lm:rate}}
		\label{ap:rate}
		The differential of $R$ with respect to $\mathbf{\Theta}_g^*$ is \cite{Hjorungnes2007}
		\begin{align*}
			\partial R
			 & = \frac{1}{\eta} \tr \biggl\{ \partial \mathbf{H}^* \cdot \mathbf{Q}^\mathsf{T} \mathbf{H}^\mathsf{T} \Bigl(\mathbf{I} + \frac{\mathbf{H}^* \mathbf{Q}^\mathsf{T} \mathbf{H}^\mathsf{T}}{\eta}\Bigr)^{-1} \biggr\}                                                                      \\
			 & = \frac{1}{\eta} \tr \biggl\{ \mathbf{H}_{\mathrm{B},g}^* \cdot \partial \mathbf{\Theta}_g^* \cdot \mathbf{H}_{\mathrm{F},g}^* \mathbf{Q}^\mathsf{T} \mathbf{H}^\mathsf{T} \Bigl(\mathbf{I} + \frac{\mathbf{H}^* \mathbf{Q}^\mathsf{T} \mathbf{H}^\mathsf{T}}{\eta}\Bigr)^{-1} \biggr\} \\
			 & = \frac{1}{\eta} \tr \biggl\{ \mathbf{H}_{\mathrm{F},g}^* \mathbf{Q}^\mathsf{T} \mathbf{H}^\mathsf{T} \Bigl(\mathbf{I} + \frac{\mathbf{H}^* \mathbf{Q}^\mathsf{T} \mathbf{H}^\mathsf{T}}{\eta}\Bigr)^{-1} \mathbf{H}_{\mathrm{B},g}^* \cdot \partial \mathbf{\Theta}_g^* \biggr\},
		\end{align*}
		and the corresponding complex derivative is \eqref{eq:gradient_eucl_rate}.
	\end{subsection}

	\begin{subsection}{Proof of Proposition~\ref{pp:power}}
		\label{ap:power}
		The differential of \eqref{ob:power} with respect to $\mathbf{\Theta}_g^*$ is
		\begin{align*}
			\partial \lVert \mathbf{H} \rVert _\mathrm{F}^2
			 & = \tr\bigl(\mathbf{H}_{\mathrm{B},g}^* \cdot \partial \mathbf{\Theta}_g^* \cdot \mathbf{H}_{\mathrm{F},g}^* (\mathbf{H}_\mathrm{D}^\mathsf{T} + \mathbf{H}_\mathrm{F}^\mathsf{T} \mathbf{\Theta}^\mathsf{T} \mathbf{H}_\mathrm{B}^\mathsf{T})\bigr) \\
			 & = \tr\bigl(\mathbf{H}_{\mathrm{F},g}^* (\mathbf{H}_\mathrm{D}^\mathsf{T} + \mathbf{H}_\mathrm{F}^\mathsf{T} \mathbf{\Theta}^\mathsf{T} \mathbf{H}_\mathrm{B}^\mathsf{T}) \mathbf{H}_{\mathrm{B},g}^* \cdot \partial \mathbf{\Theta}_g^*\bigr)
		\end{align*}
		and the corresponding complex derivative is
		\begin{equation}
			\frac{\partial \lVert \mathbf{H} \rVert _\mathrm{F}^2}{\partial \mathbf{\Theta}_g^*} = \mathbf{H}_{\mathrm{B},g}^\mathsf{H} (\mathbf{H}_\mathrm{D} + \mathbf{H}_\mathrm{B} \mathbf{\Theta} \mathbf{H}_\mathrm{F}) \mathbf{H}_{\mathrm{F},g}^\mathsf{H} \triangleq \mathbf{M}_g,
		\end{equation}
		whose \gls{svd} is denoted as $\mathbf{M}_g = \mathbf{U}_g \mathbf{\Sigma}_g \mathbf{V}_g^\mathsf{H}$.
		The quadratic objective \eqref{ob:power} can be successively approximated by its first-order Taylor expansion, resulting in the subproblem
		\begin{maxi!}
			{\scriptstyle{\mathbf{\Theta}}}{\sum_g 2 \Re\bigl\{ \tr(\mathbf{\Theta}_g^\mathsf{H} \mathbf{M}_g) \bigr\}}{\label{op:power_ris_taylor}}{\label{ob:power_ris_taylor}}
			\addConstraint{\mathbf{\Theta}_g^\mathsf{H} \mathbf{\Theta}_g=\mathbf{I}, \quad \forall g,}{}{}
		\end{maxi!}
		whose optimal solution is
		\begin{equation}
			\label{eq:ris_power_taylor}
			\tilde{\mathbf{\Theta}}_g = \mathbf{U}_g \mathbf{V}_g^\mathsf{H}, \quad \forall g.
		\end{equation}
		This is because $\Re \bigl\{\tr(\mathbf{\Theta}_g^\mathsf{H} \mathbf{M}_g)\bigr\} = \Re \bigl\{ \tr(\mathbf{\Sigma}_g \mathbf{V}_g^\mathsf{H} \mathbf{\Theta}_g^\mathsf{H} \mathbf{U}_g) \bigr\} \le \tr(\mathbf{\Sigma}_g)$ and the bound is tight when $\mathbf{V}_g^\mathsf{H} \mathbf{\Theta}_g^\mathsf{H} \mathbf{U}_g = \mathbf{I}$.

		Next, we prove that solving the affine approximation \eqref{op:power_ris_taylor} by \eqref{eq:ris_power_taylor} does not decrease \eqref{ob:power}.
		Since $\tilde{\mathbf{\Theta}} = \diag(\tilde{\mathbf{\Theta}}_1, \ldots, \tilde{\mathbf{\Theta}}_G)$ is optimal for \eqref{op:power_ris_taylor}, we have
		\begin{equation}
			\label{iq:power_ris_taylor}
			\begin{split}
				2 \Re \bigl\{     & \sum\limits_g \tr(\tilde{\mathbf{\Theta}}_g^\mathsf{H} \mathbf{H}_{\mathrm{B},g}^\mathsf{H} \mathbf{H}_\mathrm{D} \mathbf{H}_{\mathrm{F},g}^\mathsf{H})                                                                                  \\
				                  & + \sum\limits_{g_1,g_2} \tr(\tilde{\mathbf{\Theta}}_{g_1}^\mathsf{H} \mathbf{H}_{\mathrm{B},g_1}^\mathsf{H} \mathbf{H}_{\mathrm{B},g_2} \mathbf{\Theta}_{g_2} \mathbf{H}_{\mathrm{F},g_2} \mathbf{H}_{\mathrm{F},g_1}^\mathsf{H})\bigr\} \\
				\ge 2 \Re \bigl\{ & \sum\limits_g \tr({\mathbf{\Theta}}_g^\mathsf{H} \mathbf{H}_{\mathrm{B},g}^\mathsf{H} \mathbf{H}_\mathrm{D} \mathbf{H}_{\mathrm{F},g}^\mathsf{H})                                                                                        \\
				                  & + \sum\limits_{g_1,g_2} \tr({\mathbf{\Theta}}_{g_1}^\mathsf{H} \mathbf{H}_{\mathrm{B},g_1}^\mathsf{H} \mathbf{H}_{\mathrm{B},g_2} \mathbf{\Theta}_{g_2} \mathbf{H}_{\mathrm{F},g_2} \mathbf{H}_{\mathrm{F},g_1}^\mathsf{H})\bigr\}.
			\end{split}
		\end{equation}
		Besides, $\lVert \sum_g \mathbf{H}_{\mathrm{B},g} \tilde{\mathbf{\Theta}}_g \mathbf{H}_{\mathrm{F},g} - \sum_g \mathbf{H}_{\mathrm{B},g} \mathbf{\Theta}_g \mathbf{H}_{\mathrm{F},g} \rVert _\mathrm{F}^2 \ge 0$ implies
		\begin{equation}
			\label{iq:auxiliary_nonnegative}
			\begin{split}
				    & \sum\limits_{g_1,g_2} \tr(\mathbf{H}_{\mathrm{F},g_1}^\mathsf{H} \tilde{\mathbf{\Theta}}_{g_1}^\mathsf{H} \mathbf{H}_{\mathrm{B},g_1}^\mathsf{H} \mathbf{H}_{\mathrm{B},g_2} \tilde{\mathbf{\Theta}}_{g_2} \mathbf{H}_{\mathrm{F},g_2})              \\
				    & \quad + \sum\limits_{g_1,g_2} \tr(\mathbf{H}_{\mathrm{F},g_1}^\mathsf{H} {\mathbf{\Theta}}_{g_1}^\mathsf{H} \mathbf{H}_{\mathrm{B},g_1}^\mathsf{H} \mathbf{H}_{\mathrm{B},g_2} {\mathbf{\Theta}}_{g_2} \mathbf{H}_{\mathrm{F},g_2})                  \\
				\ge & 2 \Re \bigl\{\sum\limits_{g_1,g_2} \tr(\mathbf{H}_{\mathrm{F},g_1}^\mathsf{H} \tilde{\mathbf{\Theta}}_{g_1}^\mathsf{H} \mathbf{H}_{\mathrm{B},g_1}^\mathsf{H} \mathbf{H}_{\mathrm{B},g_2} \mathbf{\Theta}_{g_2} \mathbf{H}_{\mathrm{F},g_2})\bigr\}.
			\end{split}
		\end{equation}
		Adding \eqref{iq:power_ris_taylor} and \eqref{iq:auxiliary_nonnegative}, we have
		\begin{equation}
			\begin{split}
				    & 2 \Re \bigl\{\tr(\tilde{\mathbf{\Theta}}^\mathsf{H} \mathbf{H}_\mathrm{B}^\mathsf{H} \mathbf{H}_\mathrm{D} \mathbf{H}_\mathrm{F}^\mathsf{H}) \bigr\} + \tr(\mathbf{H}_\mathrm{F}^\mathsf{H} \tilde{\mathbf{\Theta}}^\mathsf{H} \mathbf{H}_\mathrm{B}^\mathsf{H} \mathbf{H}_\mathrm{B} \tilde{\mathbf{\Theta}} \mathbf{H}_\mathrm{F}) \\
				\ge & 2 \Re \bigl\{\tr({\mathbf{\Theta}}^\mathsf{H} \mathbf{H}_\mathrm{B}^\mathsf{H} \mathbf{H}_\mathrm{D} \mathbf{H}_\mathrm{F}^\mathsf{H}) \bigr\} + \tr(\mathbf{H}_\mathrm{F}^\mathsf{H} {\mathbf{\Theta}}^\mathsf{H} \mathbf{H}_\mathrm{B}^\mathsf{H} \mathbf{H}_\mathrm{B} {\mathbf{\Theta}} \mathbf{H}_\mathrm{F}),
			\end{split}
		\end{equation}
		which suggests that \eqref{ob:power} is non-decreasing as the solution iterates over \eqref{eq:ris_power_taylor}.
		Since \eqref{ob:power} is also bounded from above, the sequence of objective value converges.

		Finally, we prove that any solution when \eqref{eq:auxiliary_power} converges, denoted by $\mathbf{\Theta}'$, is a stationary point of \eqref{op:power}.
		The \gls{kkt} conditions of \eqref{op:power} and \eqref{op:power_ris_taylor} are equivalent in terms of primal/dual feasibility and complementary slackness, while the stationary conditions are respectively, $\forall g$,
		\begin{gather}
			\mathbf{H}_{\mathrm{B},g}^\mathsf{H} (\mathbf{H}_\mathrm{D} + \mathbf{H}_\mathrm{B} \mathbf{\Theta}^\star \mathbf{H}_\mathrm{F}) \mathbf{H}_{\mathrm{F},g}^\mathsf{H} - \mathbf{\Theta}_g^\star \mathbf{\Lambda}_g^\mathsf{H} = 0,\label{eq:power_ris_optimal}\\
			\mathbf{M}_g - \mathbf{\Theta}_g^\star \mathbf{\Lambda}_g^\mathsf{H} = 0.\label{eq:power_ris_taylor_optimal}
		\end{gather}
		When \eqref{eq:auxiliary_power} converges, $\mathbf{H}_{\mathrm{B},g}^\mathsf{H} (\mathbf{H}_\mathrm{D} + \mathbf{H}_\mathrm{B} \mathbf{\Theta}' \mathbf{H}_\mathrm{F}) \mathbf{H}_{\mathrm{F},g}^\mathsf{H} = \mathbf{H}_{\mathrm{B},g}^\mathsf{H} (\mathbf{H}_\mathrm{D} + \mathbf{H}_\mathrm{B} \mathbf{\Theta}^\star \mathbf{H}_\mathrm{F}) \mathbf{H}_{\mathrm{F},g}^\mathsf{H}$ and \eqref{eq:power_ris_taylor_optimal} reduces to \eqref{eq:power_ris_optimal}.
		The proof is thus completed.
	\end{subsection}

\end{appendix}

\begin{section}{Acknowledgement}
	The authors would like to thank the anonymous reviewers for their insightful criticisms and suggestions that helped us correct a few technical errors.
\end{section}

\bibliographystyle{IEEEtran}
\bibliography{library.bib}
\end{document}

%% file: preamble.tex
\usepackage{adjustbox}
\usepackage{algorithm}
\usepackage{algpseudocode}
\usepackage{amsfonts}
\usepackage{amsmath}
\usepackage{amssymb}
\usepackage{amsthm}
\usepackage{booktabs}
\usepackage{cases}
\usepackage{cite}
\usepackage[acronym]{glossaries-extra}
\usepackage{hyperref}
\usepackage{mathtools}
\usepackage[nopatch=footnote]{microtype}
\usepackage[short]{optidef}
\usepackage{pgfplots}
\usepackage[subtle]{savetrees}
\usepackage{siunitx}
\usepackage[caption=false,font=footnotesize,subrefformat=parens,labelformat=parens]{subfig}
\usepackage{tabularx}
\usepackage{tikz}
\usepackage{epstopdf}
\usepackage{multirow}
\usepackage{enumitem}
\usepackage{xr-hyper}

\allowdisplaybreaks
\DeclareMathOperator{\diag}{diag}
\DeclareMathOperator{\tr}{tr}
\DeclareMathOperator{\sv}{sv}
\DeclareMathOperator{\conv}{conv}
\DeclareMathOperator{\rank}{rank}
\DeclareMathOperator{\ran}{ran}

\theoremstyle{plain}
\newtheorem{proposition}{Proposition}
\newtheorem{corollary}{Corollary}[proposition]
\newtheorem{lemma}{Lemma}

\theoremstyle{definition}
\newtheorem{definition}{Definition}
\newtheorem{example}{Example}
\newtheorem{remark}{Remark}

\setlist[itemize]{leftmargin=5mm}
\setlist[enumerate]{leftmargin=5mm}

\usetikzlibrary{arrows,calc,matrix,patterns,plotmarks,positioning,shapes}
\usetikzlibrary{decorations.pathmorphing,decorations.pathreplacing,decorations.shapes,shapes.geometric}
\usepgfplotslibrary{groupplots,patchplots}
\pgfplotsset{compat=1.18}
\pgfplotsset{every axis/.append style={
			title style={font=\Large},
			label style={font=\Large},
			tick label style={font=\Large},
			legend style={font=\large},
		}}

\algrenewcommand{\algorithmicrequire}{\textbf{Input:}}
\algrenewcommand{\algorithmicensure}{\textbf{Output:}}
\algrenewcommand{\algorithmicwhile}{\textbf{While}}
\algrenewcommand{\algorithmicend}{\textbf{End}}
\algrenewcommand{\algorithmicrepeat}{\textbf{Repeat}}
\algrenewcommand{\algorithmicuntil}{\textbf{Until}}
\algrenewcommand{\algorithmicfor}{\textbf{For}}
\algrenewcommand{\algorithmicif}{\textbf{If}}
\algrenewcommand{\algorithmicelse}{\textbf{Else}}
\algrenewcommand{\algorithmicdo}{}
\algrenewcommand{\algorithmicthen}{}
\algnewcommand{\Initialize}[1]{%
	\State \textbf{Initialize }{#1}
}

\glsdisablehyper
\setabbreviationstyle[acronym]{long-short}
\newacronym{ao}{AO}{Alternating Optimization}
\newacronym{bd}{BD}{Beyond-Diagonal}
\newacronym{bcd}{BCD}{Block Coordinate Descent}
\newacronym{d}{D}{Diagonal}
\newacronym{dof}{DoF}{Degrees of Freedom}
\newacronym{siso}{SISO}{Single-Input Single-Output}
\newacronym{miso}{MISO}{Multiple-Input Single-Output}
\newacronym{mimo}{MIMO}{Multiple-Input Multiple-Output}
\newacronym{rcg}{RCG}{Riemannian Conjugate Gradient}
\newacronym{ris}{RIS}{Reconfigurable Intelligent Surface}
\newacronym{gp}{GP}{Globally Passive}
\newacronym{pc}{PC}{Point-to-Point Channel}
\newacronym{p2p}{P2P}{Point-to-Point}
\newacronym{ic}{IC}{Interference Channel}
\newacronym{qos}{QoS}{Quality of Service}
\newacronym{wsr}{WSR}{Weighted Sum-Rate}
\newacronym{snr}{SNR}{Signal-to-Noise Ratio}
\newacronym{svd}{SVD}{Singular Value Decomposition}
\newacronym{mmse}{MMSE}{Minimum Mean-Square Error}
\newacronym{wmmse}{WMMSE}{Weighted \gls{mmse}}
\newacronym{mse}{MSE}{Mean-Square Error}
\newacronym{los}{LoS}{Line-of-Sight}
\newacronym{csi}{CSI}{Channel State Information}
\newacronym{cscg}{CSCG}{Circularly Symmetric Complex Gaussian}
\newacronym{sca}{SCA}{Successive Convex Approximation}
\newacronym{kkt}{KKT}{Karush-Kuhn-Tucker}
\newacronym{rf}{RF}{Radio Frequency}
\newacronym{sv}{SV}{Singular Value}
\newacronym{pdd}{PDD}{Penalty Dual Decompistion}
\newacronym{sinr}{SINR}{Signal-to-Interference Noise Ratio}
\newacronym{af}{AF}{Amplify-and-Forward}

\usepackage{setspace}


%% file: assets/simulation/pc_singular_pareto_sx32_nd.tex
%
%
\definecolor{mycolor1}{rgb}{0.00000,0.44706,0.74118}%
\definecolor{mycolor2}{rgb}{0.92900,0.69400,0.12500}%
\definecolor{mycolor3}{rgb}{0.49400,0.18400,0.55600}%
\definecolor{mycolor4}{rgb}{0.85098,0.32549,0.09804}%
\definecolor{mycolor5}{rgb}{0.63500,0.07800,0.18400}%
\definecolor{mycolor6}{rgb}{0.00000,0.44700,0.74100}%
\definecolor{mycolor7}{rgb}{0.92941,0.69412,0.12549}%
\definecolor{mycolor8}{rgb}{0.46600,0.67400,0.18800}%
\definecolor{mycolor9}{rgb}{0.49412,0.18431,0.55686}%
\definecolor{mycolor10}{rgb}{0.85000,0.32500,0.09800}%
\definecolor{mycolor11}{rgb}{0.50196,0.50196,0.50196}%
\begin{tikzpicture}

\begin{axis}[%
width=9.509cm,
height=7.5cm,
at={(0cm,0cm)},
scale only axis,
unbounded coords=jump,
xmin=0,
xmax=0.00038965048998045,
xlabel={$\sigma_1(\mathbf{H})$},
ymin=0,
ymax=0.000304493100518335,
ylabel={$\sigma_2(\mathbf{H})$},
axis background/.style={fill=white},
xmajorgrids,
ymajorgrids,
legend style={at={(0.03,0.97)}, anchor=north west, legend cell align=left, align=left, draw=white!15!black},
every axis plot/.append style={line width=1.5pt}
]
\addplot[only marks, mark=triangle, mark options={}, mark size=2.3570pt, draw=black] table[row sep=crcr]{%
x	y\\
0	0\\
};
\addlegendentry{Direct}

\addplot [color=mycolor1, dotted, line width=2.0pt]
  table[row sep=crcr]{%
0.000267895667198245	0.000171134879393267\\
0.000265726380728974	0.000173150928555956\\
0.000263620365603918	0.000175013140296501\\
0.000261445144530834	0.000176843202895948\\
0.000259164875117885	0.000178667936743009\\
0.000256772775708205	0.000180487719376887\\
0.000254278514486235	0.000182290392662069\\
0.000251704002269913	0.000184056601213238\\
0.000249073506157511	0.000185767894974696\\
0.000246400740748153	0.000187414817948743\\
0.000243687718680905	0.000188995994660418\\
0.000240928492221004	0.000190514582140261\\
0.000238098477127779	0.000191982763893684\\
0.000235165035954025	0.000193414254436945\\
0.000232089995279799	0.000194822374050852\\
0.000228824413086908	0.000196221753795029\\
0.000225316977451163	0.000197623817299328\\
0.000221535346161309	0.000199028850294541\\
0.000203259528520131	0.000203259526698551\\
0.000199827287451323	0.000199827287443645\\
0.000198123535486915	0.000198123535486911\\
0.000180148810088901	0.0001801488100889\\
8.42468741466946e-21	6.05571416226906e-21\\
6.19979890484365e-21	3.35030862543257e-21\\
6.60756407329265e-21	1.28456637898193e-21\\
7.66177135693828e-21	3.92046971006709e-22\\
0.000180096069830613	1.08306169491646e-20\\
0.000191831039008028	1.66189243361953e-20\\
0.000264113777998769	7.79621614858511e-20\\
0.000279564446825807	4.05296860635826e-19\\
0.000280589608744265	1.14924619753218e-18\\
0.000281565424666408	3.19287626893919e-17\\
0.000294887428789775	6.22978162678605e-16\\
0.000298297210243084	1.4033135819947e-15\\
0.000309415141512222	5.02384211049284e-15\\
0.000312286791270632	6.622168556894e-15\\
0.000313941579001134	1.5519784325904e-11\\
0.000315080469539148	2.60442929195929e-06\\
0.000316650462662155	6.20803322481037e-06\\
0.000317953658131026	1.0637150181653e-05\\
0.000319531793152611	1.60801484565072e-05\\
0.000321145121813445	2.39596943662345e-05\\
0.000321463002882309	3.21072961439859e-05\\
0.000321532819016651	3.76429148560442e-05\\
0.000321472892430131	4.25828187614136e-05\\
0.000321300760117647	4.72640647983865e-05\\
0.000321020731972484	5.18237773689732e-05\\
0.000320633271897779	5.63231191746697e-05\\
0.000320137546722193	6.07926465706622e-05\\
0.00031953359137215	6.52391237196603e-05\\
0.000318820183056992	6.96729345667314e-05\\
0.000317997781544904	7.40900801825946e-05\\
0.000317066229794775	7.84903211168037e-05\\
0.000316024782888273	8.28755034268974e-05\\
0.000314867375461116	8.72664954039366e-05\\
0.000313584778992714	9.16888362417147e-05\\
0.000312161323906551	9.61814565544305e-05\\
0.000310588510709343	0.000100752805689686\\
0.000308886100168095	0.000105333037823099\\
0.000307099022755485	0.000109802296733328\\
0.000304874570657941	0.000114975407123897\\
0.000295523412404848	0.000136118926480847\\
0.000293796779560532	0.000139257271731018\\
0.000292107248523008	0.000142157036202118\\
0.000290431849142458	0.000144876663627326\\
0.000288754605499061	0.000147455239474912\\
0.000287064679129902	0.000149918781615647\\
0.000285354493393488	0.000152285164138566\\
0.000283619431704674	0.00015456587543467\\
0.000281858022889306	0.000156766916808042\\
0.000280067185001826	0.00015889544523825\\
0.000278240363912649	0.000160961566094741\\
0.000276367480760894	0.000162977751672099\\
0.000267895667198245	0.000171134879393267\\
};
\addlegendentry{$L = 1$}

\addplot[only marks, mark=o, mark options={}, mark size=3.5355pt, draw=mycolor1, forget plot] table[row sep=crcr]{%
x	y\\
0.000314867375461116	8.72664954039366e-05\\
};
\addplot[only marks, mark=x, mark options={}, mark size=3.5355pt, draw=mycolor1, forget plot] table[row sep=crcr]{%
x	y\\
0.000321532819016651	3.76429148560442e-05\\
};
\addplot [color=mycolor1, line width=2.0pt, forget plot]
  table[row sep=crcr]{%
0.000321532819016651	3.76429148560442e-05\\
0.000321472892430131	4.25828187614136e-05\\
0.000321300760117647	4.72640647983865e-05\\
0.000321020731972484	5.18237773689732e-05\\
0.000320633271897779	5.63231191746697e-05\\
0.000320137546722193	6.07926465706622e-05\\
0.00031953359137215	6.52391237196603e-05\\
0.000318820183056992	6.96729345667314e-05\\
0.000317997781544904	7.40900801825946e-05\\
0.000317066229794775	7.84903211168037e-05\\
0.000316024782888273	8.28755034268974e-05\\
0.000314867375461116	8.72664954039366e-05\\
0.000313584778992714	9.16888362417147e-05\\
0.000312161323906551	9.61814565544305e-05\\
0.000310588510709343	0.000100752805689686\\
0.000308886100168095	0.000105333037823099\\
0.000307099022755485	0.000109802296733328\\
0.000304874570657941	0.000114975407123897\\
0.000295523412404848	0.000136118926480847\\
0.000293796779560532	0.000139257271731018\\
0.000292107248523008	0.000142157036202118\\
0.000290431849142458	0.000144876663627326\\
0.000288754605499061	0.000147455239474912\\
0.000287064679129902	0.000149918781615647\\
0.000285354493393488	0.000152285164138566\\
0.000283619431704674	0.00015456587543467\\
0.000281858022889306	0.000156766916808042\\
0.000280067185001826	0.00015889544523825\\
0.000278240363912649	0.000160961566094741\\
0.000276367480760894	0.000162977751672099\\
0.000267895667198245	0.000171134879393267\\
0.000267895667198245	0.000171134879393267\\
0.000265726380728974	0.000173150928555956\\
0.000263620365603918	0.000175013140296501\\
0.000261445144530834	0.000176843202895948\\
0.000259164875117885	0.000178667936743009\\
0.000256772775708205	0.000180487719376887\\
0.000254278514486235	0.000182290392662069\\
0.000251704002269913	0.000184056601213238\\
0.000249073506157511	0.000185767894974696\\
0.000246400740748153	0.000187414817948743\\
0.000243687718680905	0.000188995994660418\\
0.000240928492221004	0.000190514582140261\\
0.000238098477127779	0.000191982763893684\\
0.000235165035954025	0.000193414254436945\\
0.000232089995279799	0.000194822374050852\\
0.000228824413086908	0.000196221753795029\\
0.000225316977451163	0.000197623817299328\\
0.000221535346161309	0.000199028850294541\\
0.000203259528520131	0.000203259526698551\\
};
\addplot [color=mycolor4, dotted, line width=2.0pt]
  table[row sep=crcr]{%
0.000355739456151093	0.00023630212231011\\
0.000353983785016404	0.000238013710074763\\
0.000351859225447142	0.000239984970213792\\
0.000349054422291465	0.000242460729759867\\
0.000344270181551168	0.000246471602221089\\
0.000332985328589198	0.000255488751251254\\
0.000326266184355555	0.000260610115544627\\
0.000320562886323912	0.000264733162717323\\
0.000314325700971254	0.000269011913129879\\
0.000308560245761306	0.000272769530149513\\
0.000304688873039902	0.000275159663058865\\
0.000301895314065004	0.000276789729801272\\
0.000299616367457836	0.000278045082584301\\
0.000297629447604807	0.00027907656379607\\
0.000295820894328768	0.000279959625818801\\
0.000294119244751247	0.000280739245424517\\
0.000292480429716739	0.00028144185454148\\
0.000290869271980479	0.000282086202217012\\
0.000289258251330994	0.000282684955608739\\
0.000287439311775594	0.000283221155075434\\
0.000286111329523935	0.000283226241828683\\
0.000281749477516674	0.000281749467335335\\
0.000279440614092202	0.000279440611372918\\
0.000220457825755789	0.000220457825746552\\
0.000208068910616119	0.000208068910610979\\
0.000187946250277845	0.000187946250276701\\
0.000180148810088901	0.000180148810088901\\
6.5591718870255e-05	6.55917188702549e-05\\
4.20164701380515e-21	8.79388127033691e-22\\
6.32472477104445e-21	3.43356879819089e-22\\
0.000175524733181838	4.89153384024627e-21\\
0.000191831039008028	8.05014238109372e-21\\
0.000223156565739869	2.18554354445088e-20\\
0.000298795369668053	4.43259933963327e-19\\
0.000319836197632342	1.14449267159391e-17\\
0.000337395279306117	2.33052717459466e-17\\
0.000345669096124742	1.5235327945565e-16\\
0.000371901015970715	7.29825031398575e-16\\
0.000372437013168705	1.73328225712588e-13\\
0.000372486831885383	2.81023595538404e-05\\
0.000372471466593596	0.000201236641395997\\
0.000372426852724318	0.000202454117307786\\
0.000372354749362613	0.000203628518387191\\
0.000372256427569083	0.000204770528112299\\
0.000372133071603637	0.000205882935186015\\
0.000371985309281415	0.000206970833147438\\
0.000371813719503542	0.000208037182759402\\
0.00037161832382613	0.000209086468304937\\
0.00037139955457068	0.000210119666912738\\
0.00037115668565401	0.000211142180316026\\
0.000370889086656138	0.000212157232410383\\
0.000370596708144513	0.000213165221964369\\
0.000370278106444865	0.000214170692984417\\
0.000369931918553916	0.000215176768360127\\
0.000369557062102171	0.000216184995411197\\
0.000369151325283754	0.000217199302911258\\
0.000368713688516441	0.000218219947660847\\
0.000368241418310008	0.000219250752341472\\
0.000367731914734066	0.000220294367526888\\
0.000367183349266502	0.000221351363599926\\
0.000366591723007183	0.000222425941716797\\
0.000365955230149435	0.00022351770001446\\
0.000365268517707625	0.00022463180272036\\
0.000364529269287204	0.000225767731198768\\
0.000363731076897723	0.000226930750043674\\
0.00036286865154264	0.000228123468496278\\
0.000361936078587719	0.000229348684335404\\
0.000360923096272773	0.000230613841212628\\
0.000359816798832403	0.000231927997905375\\
0.000358602593712985	0.000233300381138745\\
0.000357257858194128	0.000234747077469494\\
0.000355739456151093	0.00023630212231011\\
};
\addlegendentry{$L = 4$}

\addplot[only marks, mark=o, mark options={}, mark size=3.5355pt, draw=mycolor4, forget plot] table[row sep=crcr]{%
x	y\\
0.000365955230149435	0.00022351770001446\\
};
\addplot[only marks, mark=x, mark options={}, mark size=3.5355pt, draw=mycolor4, forget plot] table[row sep=crcr]{%
x	y\\
0.000372486831885383	2.81023595538404e-05\\
};
\addplot [color=mycolor4, line width=2.0pt, forget plot]
  table[row sep=crcr]{%
0.000372486831885383	2.81023595538404e-05\\
0.000372471466593596	0.000201236641395997\\
0.000372426852724318	0.000202454117307786\\
0.000372354749362613	0.000203628518387191\\
0.000372256427569083	0.000204770528112299\\
0.000372133071603637	0.000205882935186015\\
0.000371985309281415	0.000206970833147438\\
0.000371813719503542	0.000208037182759402\\
0.00037161832382613	0.000209086468304937\\
0.00037139955457068	0.000210119666912738\\
0.00037115668565401	0.000211142180316026\\
0.000370889086656138	0.000212157232410383\\
0.000370596708144513	0.000213165221964369\\
0.000370278106444865	0.000214170692984417\\
0.000369931918553916	0.000215176768360127\\
0.000369557062102171	0.000216184995411197\\
0.000369151325283754	0.000217199302911258\\
0.000368713688516441	0.000218219947660847\\
0.000368241418310008	0.000219250752341472\\
0.000367731914734066	0.000220294367526888\\
0.000367183349266502	0.000221351363599926\\
0.000366591723007183	0.000222425941716797\\
0.000365955230149435	0.00022351770001446\\
0.000365268517707625	0.00022463180272036\\
0.000364529269287204	0.000225767731198768\\
0.000363731076897723	0.000226930750043674\\
0.00036286865154264	0.000228123468496278\\
0.000361936078587719	0.000229348684335404\\
0.000360923096272773	0.000230613841212628\\
0.000359816798832403	0.000231927997905375\\
0.000358602593712985	0.000233300381138745\\
0.000357257858194128	0.000234747077469494\\
0.000355739456151093	0.00023630212231011\\
0.000355739456151093	0.00023630212231011\\
0.000353983785016404	0.000238013710074763\\
0.000351859225447142	0.000239984970213792\\
0.000349054422291465	0.000242460729759867\\
0.000344270181551168	0.000246471602221089\\
0.000332985328589198	0.000255488751251254\\
0.000326266184355555	0.000260610115544627\\
0.000320562886323912	0.000264733162717323\\
0.000314325700971254	0.000269011913129879\\
0.000308560245761306	0.000272769530149513\\
0.000304688873039902	0.000275159663058865\\
0.000301895314065004	0.000276789729801272\\
0.000299616367457836	0.000278045082584301\\
0.000297629447604807	0.00027907656379607\\
0.000295820894328768	0.000279959625818801\\
0.000294119244751247	0.000280739245424517\\
0.000292480429716739	0.00028144185454148\\
0.000290869271980479	0.000282086202217012\\
0.000289258251330994	0.000282684955608739\\
0.000287439311775594	0.000283221155075434\\
0.000286111329523935	0.000283226241828683\\
};
\addplot [color=mycolor7, dotted, line width=2.0pt]
  table[row sep=crcr]{%
0.000387644725001034	0.000252163649786122\\
0.000387494244577878	0.000252309874048033\\
0.000387321385451581	0.00025246982456047\\
0.000387123619656507	0.00025264497436984\\
0.000354044414728323	0.000278788391537759\\
0.000324031566213212	0.00030181769538482\\
0.000323370970606122	0.000302322033286758\\
0.000322939503948698	0.000302634503158319\\
0.000322629710522514	0.000302847188929005\\
0.000322389738998535	0.0003030037313976\\
0.000322192012399091	0.000303125657570389\\
0.000322021223310372	0.000303225205191342\\
0.000321868058470525	0.000303309585408043\\
0.000321726332885747	0.00030338319606735\\
0.000321591838228086	0.000303448853401541\\
0.000321461357536437	0.000303508427690401\\
0.000321333801566517	0.000303563282339639\\
0.000321206511969601	0.000303614207600706\\
0.000321078376511536	0.00030366181921831\\
0.000320948398043425	0.00030370654144927\\
0.000320815486892904	0.000303748637788079\\
0.000320678813033656	0.000303788277927965\\
0.000309564989052386	0.000301409733854497\\
0.000288166772576663	0.000287746907546531\\
0.000279440619857278	0.000279440619671027\\
0.000220457825773397	0.000220457825766601\\
0.000193478737274631	0.000193478737273387\\
0.000187946250278386	0.00018794625027787\\
0.000180148810088901	0.000180148810088901\\
7.38779007531488e-05	7.38779007531487e-05\\
6.5591718870255e-05	6.55917188702549e-05\\
6.11855785437379e-21	3.45293064321707e-21\\
4.20164701380515e-21	8.79388127033691e-22\\
1.54145558743596e-20	4.5368881526408e-22\\
0.000121848054411894	4.00695794022811e-21\\
0.000191831039008028	8.05014238109372e-21\\
0.000298795369668053	1.54041673309025e-20\\
0.000319836197632354	2.36457288348854e-18\\
0.000333476889166338	1.75021877124016e-17\\
0.000335729299617477	4.56099183583301e-17\\
0.000337395279306272	8.16975720607573e-17\\
0.000345669096125748	4.1981459702903e-16\\
0.000346031641514439	4.39830021468541e-15\\
0.000387914582890843	1.13358549119329e-12\\
0.000388584823891499	8.0163235518879e-12\\
0.000389055459349451	2.17977036180119e-08\\
0.000389129203196301	2.36054120455867e-05\\
0.000389090206430443	0.000249310691234572\\
0.000389037671865594	0.00024972476835498\\
0.000388990920465773	0.000249934686797231\\
0.000388964382651109	0.000250035016505946\\
0.000388935788614941	0.000250134077714785\\
0.000388905063299093	0.000250231182320069\\
0.000388872172163426	0.000250326693388387\\
0.000388837071247733	0.00025042106587312\\
0.000388799688709563	0.000250514636529308\\
0.000388759882275995	0.000250607325760949\\
0.00038871754306719	0.000250699736848645\\
0.000388672525582172	0.000250792084880618\\
0.000388624572418996	0.000250884495860529\\
0.000388573499339921	0.000250977447103048\\
0.000388518934001622	0.000251071035900666\\
0.000388460497180009	0.000251165596815782\\
0.000388397897277804	0.000251262019216711\\
0.000388330323709399	0.000251359831654839\\
0.000388257446510124	0.0002514612060543\\
0.000388178442275181	0.000251564986042989\\
0.00038809203917229	0.000251672744789679\\
0.000387997416382034	0.000251785310299683\\
0.000387892731864093	0.000251903721206058\\
0.000387776056777691	0.000252029234268639\\
0.000387644725001034	0.000252163649786122\\
};
\addlegendentry{$L = 16$}

\addplot[only marks, mark=o, mark options={}, mark size=3.5355pt, draw=mycolor7, forget plot] table[row sep=crcr]{%
x	y\\
0.000388460497180009	0.000251165596815782\\
};
\addplot[only marks, mark=x, mark options={}, mark size=3.5355pt, draw=mycolor7, forget plot] table[row sep=crcr]{%
x	y\\
0.000389129203196301	2.36054120455867e-05\\
};
\addplot [color=mycolor7, line width=2.0pt, forget plot]
  table[row sep=crcr]{%
0.000389129203196301	2.36054120455867e-05\\
0.000389090206430443	0.000249310691234572\\
0.000389037671865594	0.00024972476835498\\
0.000388990920465773	0.000249934686797231\\
0.000388964382651109	0.000250035016505946\\
0.000388935788614941	0.000250134077714785\\
0.000388905063299093	0.000250231182320069\\
0.000388872172163426	0.000250326693388387\\
0.000388837071247733	0.00025042106587312\\
0.000388799688709563	0.000250514636529308\\
0.000388759882275995	0.000250607325760949\\
0.00038871754306719	0.000250699736848645\\
0.000388672525582172	0.000250792084880618\\
0.000388624572418996	0.000250884495860529\\
0.000388573499339921	0.000250977447103048\\
0.000388518934001622	0.000251071035900666\\
0.000388460497180009	0.000251165596815782\\
0.000388397897277804	0.000251262019216711\\
0.000388330323709399	0.000251359831654839\\
0.000388257446510124	0.0002514612060543\\
0.000388178442275181	0.000251564986042989\\
0.00038809203917229	0.000251672744789679\\
0.000387997416382034	0.000251785310299683\\
0.000387892731864093	0.000251903721206058\\
0.000387776056777691	0.000252029234268639\\
0.000387644725001034	0.000252163649786122\\
0.000387644725001034	0.000252163649786122\\
0.000387494244577878	0.000252309874048033\\
0.000387321385451581	0.00025246982456047\\
0.000387123619656507	0.00025264497436984\\
0.000354044414728323	0.000278788391537759\\
0.000324031566213212	0.00030181769538482\\
0.000323370970606122	0.000302322033286758\\
0.000322939503948698	0.000302634503158319\\
0.000322629710522514	0.000302847188929005\\
0.000322389738998535	0.0003030037313976\\
0.000322192012399091	0.000303125657570389\\
0.000322021223310372	0.000303225205191342\\
0.000321868058470525	0.000303309585408043\\
0.000321726332885747	0.00030338319606735\\
0.000321591838228086	0.000303448853401541\\
0.000321461357536437	0.000303508427690401\\
0.000321333801566517	0.000303563282339639\\
0.000321206511969601	0.000303614207600706\\
0.000321078376511536	0.00030366181921831\\
0.000320948398043425	0.00030370654144927\\
0.000320815486892904	0.000303748637788079\\
0.000320678813033656	0.000303788277927965\\
};
\addplot [color=mycolor9, dotted, line width=2.0pt]
  table[row sep=crcr]{%
0.000355261169740309	0.000280320559854823\\
0.00032414018186505	0.000304493100518335\\
0.000324140181659608	0.000304493100483771\\
0.000317179030115277	0.000302551878646465\\
0.000300436107604006	0.000297641691160318\\
0.000220457825791568	0.000220457825777682\\
0.000208068910631096	0.000208068910627274\\
0.000193478737275628	0.000193478737274652\\
0.00018794625027891	0.000187946250278393\\
0.000180148810088901	0.000180148810088901\\
7.38779007531488e-05	7.38779007531487e-05\\
6.71666052516329e-05	6.71666052516328e-05\\
3.15465671181684e-21	1.72642107428799e-21\\
4.53228785082067e-21	7.20639588176322e-22\\
5.64685722791264e-21	5.28482091624286e-22\\
8.20325856378038e-21	1.7761692113009e-22\\
0.000111248475093952	1.37743980728087e-21\\
0.000180096069830613	2.16135335281665e-21\\
0.000298795369668053	5.37505173232332e-21\\
0.000319836197632357	5.46792961128607e-19\\
0.000333476889166367	2.0283151490536e-17\\
0.000337395279306384	5.36077757300478e-17\\
0.000345669096126209	2.78635951089719e-16\\
0.000387914582997111	2.73836722723604e-14\\
0.000389542876653209	2.94230751176389e-08\\
0.000389650489944512	2.36057139085565e-05\\
0.00038965048998045	0.000253270842548633\\
0.000389650489894017	0.000253270843326888\\
0.000355261169740309	0.000280320559854823\\
};
\addlegendentry{$L = 32$}

\addplot[only marks, mark=o, mark options={}, mark size=3.5355pt, draw=mycolor9, forget plot] table[row sep=crcr]{%
x	y\\
0.000355261169740309	0.000280320559854823\\
0.00032414018186505	0.000304493100518335\\
0.000324140181659608	0.000304493100483771\\
0.00038965048998045	0.000253270842548633\\
0.000389650489894017	0.000253270843326888\\
0.000355261169740309	0.000280320559854823\\
};
\addplot[only marks, mark=x, mark options={}, mark size=3.5355pt, draw=mycolor9, forget plot] table[row sep=crcr]{%
x	y\\
0.00038965048998045	0.000253270842548633\\
};
\addplot [color=mycolor9, line width=2.0pt, forget plot]
  table[row sep=crcr]{%
0.00038965048998045	0.000253270842548633\\
0.000389650489894017	0.000253270843326888\\
0.000355261169740309	0.000280320559854823\\
0.000355261169740309	0.000280320559854823\\
0.00032414018186505	0.000304493100518335\\
};

\addlegendimage{only marks, mark=o, mark options={solid}, mark size=2pt, line width=2pt, draw=mycolor11}
\addlegendentry{P-max}

\addlegendimage{only marks, mark=x, mark options={solid}, mark size=3pt, line width=2pt, draw=mycolor11}
\addlegendentry{E-max}

\addplot [draw=none, color=mycolor11, line width=2.0pt]
  table[row sep=crcr]{%
0	0\\
};
\addlegendentry{R-max}

\end{axis}
\end{tikzpicture}%

%% file: assets/simulation/pc_singular_pareto_sx32.tex
%
%
\definecolor{mycolor1}{rgb}{0.00000,0.44706,0.74118}%
\definecolor{mycolor2}{rgb}{0.92900,0.69400,0.12500}%
\definecolor{mycolor3}{rgb}{0.49400,0.18400,0.55600}%
\definecolor{mycolor4}{rgb}{0.85098,0.32549,0.09804}%
\definecolor{mycolor5}{rgb}{0.63500,0.07800,0.18400}%
\definecolor{mycolor6}{rgb}{0.00000,0.44700,0.74100}%
\definecolor{mycolor7}{rgb}{0.92941,0.69412,0.12549}%
\definecolor{mycolor8}{rgb}{0.46600,0.67400,0.18800}%
\definecolor{mycolor9}{rgb}{0.49412,0.18431,0.55686}%
\definecolor{mycolor10}{rgb}{0.85000,0.32500,0.09800}%
\definecolor{mycolor11}{rgb}{0.50196,0.50196,0.50196}%
\begin{tikzpicture}

\begin{axis}[%
width=9.509cm,
height=7.5cm,
at={(0cm,0cm)},
scale only axis,
unbounded coords=jump,
xmin=0.000850130707459294,
xmax=0.001495567105556,
xlabel={$\sigma_1(\mathbf{H})$},
ymin=5.81124124465135e-05,
ymax=0.000756440229046408,
ylabel={$\sigma_2(\mathbf{H})$},
axis background/.style={fill=white},
xmajorgrids,
ymajorgrids,
legend style={at={(0.03,0.97)}, anchor=north west, legend cell align=left, align=left, draw=white!15!black},
every axis plot/.append style={line width=1.5pt}
]
\addplot[only marks, mark=triangle, mark options={}, mark size=2.3570pt, draw=black] table[row sep=crcr]{%
x	y\\
0.00117134827739843	0.000407871162732991\\
};
\addlegendentry{Direct}

\addplot [color=mycolor1, dotted, line width=2.0pt]
  table[row sep=crcr]{%
0.00137955431172681	0.000588505616267897\\
0.00137418269622086	0.000593615295742653\\
0.00136749093451528	0.000599376628494274\\
0.00135890213089589	0.000606072077153129\\
0.00134973459427533	0.000612532388777538\\
0.0013410733835681	0.000618021117436472\\
0.00133234723576559	0.000622963402732477\\
0.00132298835416539	0.000627671842535913\\
0.00131283927996255	0.000632167827335414\\
0.00130175068079053	0.000636441315549385\\
0.00128938330235212	0.000640521027244044\\
0.00127272827567336	0.000645071318717164\\
0.001239223656723	0.000652826329399127\\
0.00121181294625974	0.000657515986538869\\
0.0011948709631463	0.000659612394858911\\
0.00117924446441914	0.000660770327164558\\
0.00116465364423377	0.000661132293899478\\
0.00113619330680091	0.000660369140191549\\
0.00111681335066479	0.00065893828544607\\
0.00110400776403201	0.00065738003226823\\
0.00109347165774378	0.000655554479547099\\
0.00108325459149789	0.00065326015322085\\
0.00107295870959626	0.000650409170334344\\
0.00105545417968485	0.000644538951540662\\
0.00104038763128197	0.000638729468304811\\
0.00102417401193919	0.000631551077219278\\
0.00101241548515279	0.000625650365523887\\
0.00100425860502732	0.000621064981892173\\
0.0009940831826727	0.000614742850166296\\
0.000987415602981892	0.000610234296197683\\
0.000980849676409786	0.000605333755974594\\
0.000975257972036462	0.000600727614092169\\
0.000970012018476528	0.000595952796375679\\
0.000965688097012915	0.000591642872687785\\
0.000961613627835618	0.000587170882990767\\
0.000958073981707157	0.000582866103258482\\
0.000954801020084192	0.000578430681164638\\
0.000951712819265589	0.000573772911667256\\
0.000947720748747545	0.000567031466535514\\
0.00094368967956249	0.000559421880095604\\
0.000938666225540191	0.000548842479170087\\
0.000933989675734051	0.000537341214562972\\
0.000929618025042593	0.000524954702892616\\
0.000923396186557418	0.000504608362177913\\
0.000919374605694746	0.000488686430151083\\
0.00091547451826384	0.000469323228217002\\
0.000912531517333292	0.000449138286656906\\
0.000910612071336649	0.000429871807904457\\
0.000909389566483529	0.000408076663647609\\
0.00090933992249972	0.000390775894408823\\
0.000910025743387089	0.000375356272495326\\
0.00091091055269242	0.000366524713164421\\
0.000912645642715888	0.000355223472332591\\
0.000915376017019132	0.000342148360849908\\
0.000919105832897833	0.000327691792510907\\
0.00092356432588294	0.000313121837840753\\
0.000928759465232966	0.000298651133960116\\
0.000942287333852755	0.000265315511165294\\
0.000945407531781455	0.000258678174141125\\
0.000948227552284999	0.000253339610730937\\
0.000951261088922912	0.00024820687902244\\
0.000954330222322886	0.000243539341369861\\
0.000957758218767062	0.000238823590128389\\
0.000960765045145899	0.000235047301644474\\
0.00096247943759114	0.000233103328996495\\
0.000969263161794999	0.000225981910890112\\
0.000977251953109233	0.000218323364932485\\
0.000985528271575816	0.000211536294192776\\
0.000992872848283534	0.000206166024981668\\
0.0010013184827797	0.00020061025405219\\
0.00101013591969619	0.00019539863910245\\
0.00102058684140004	0.000189876058840151\\
0.00103338161892076	0.000183865790027054\\
0.00104875048028412	0.000177484198331724\\
0.00106282584387646	0.00017241993489089\\
0.0010901148768206	0.000164497710367294\\
0.00110349204607487	0.000161179713098831\\
0.00112069218342597	0.00015785505182807\\
0.00113884160166316	0.000155164648932022\\
0.0011562057946585	0.000153460124509895\\
0.00118292823905634	0.000152352333453505\\
0.00120578552949893	0.00015242439366199\\
0.00122445160762688	0.000153374551047328\\
0.00123563402207558	0.000154663630250101\\
0.00124664204314176	0.000156532905506736\\
0.00125646840843924	0.00015873919398177\\
0.0012662380127387	0.000161444101481681\\
0.00127841894557433	0.000165502994603158\\
0.00129598590726039	0.000172188795659219\\
0.00130943172976148	0.000178179258685881\\
0.00132294341473045	0.000184951527100318\\
0.00133349182303189	0.000190922665910368\\
0.0013495725197189	0.000201255032705785\\
0.00136120068833834	0.000209375390277804\\
0.00136890916574433	0.000215378485139867\\
0.00137417862830058	0.000219910697088492\\
0.00137806203541532	0.000223602625444923\\
0.00137806478360425	0.000223605373435129\\
0.00138134703659724	0.000227051007879811\\
0.00138434858265688	0.000230530827225564\\
0.00138726696437118	0.000234272897611889\\
0.00139743637599501	0.000249104704351882\\
0.00140097987249171	0.000254695880160429\\
0.00140427131598064	0.000260507662721602\\
0.0014072558936856	0.000266435943662982\\
0.00140993387764274	0.000272479387021766\\
0.0014125096938205	0.000279170248174948\\
0.00142243876078053	0.000309807745633092\\
0.00142481748232869	0.000318421975031348\\
0.00142751115465991	0.000330552652584366\\
0.00143206701205701	0.000356986980452433\\
0.00143443314590829	0.000376177388259869\\
0.00143628278130165	0.000401042934134447\\
0.00143665815258903	0.000415196865978549\\
0.00143626535833052	0.000429896769863785\\
0.00143475282867046	0.000450519703825262\\
0.00143293009981114	0.00046549066117573\\
0.00143115255664776	0.000475813256555725\\
0.00142892158373648	0.000485645686787124\\
0.00142327161992456	0.000505869693634385\\
0.00141658082101008	0.000525970631421433\\
0.00141195598395635	0.000538042748709412\\
0.00140755516351337	0.000547974537410286\\
0.00140338490443229	0.000556280887626802\\
0.00139978200072969	0.000562651295896027\\
0.00139627232965815	0.000568195828553697\\
0.00139257689630319	0.000573442682146656\\
0.00138859749227567	0.000578540827318686\\
0.00138428272742652	0.000583541389326504\\
0.00137955431172681	0.000588505616267897\\
};
\addlegendentry{$L = 1$}

\addplot[only marks, mark=o, mark options={}, mark size=3.5355pt, draw=mycolor1, forget plot] table[row sep=crcr]{%
x	y\\
0.00141658082101008	0.000525970631421433\\
};
\addplot[only marks, mark=x, mark options={}, mark size=3.5355pt, draw=mycolor1, forget plot] table[row sep=crcr]{%
x	y\\
0.00143665815258903	0.000415196865978549\\
};
\addplot [color=mycolor1, line width=2.0pt, forget plot]
  table[row sep=crcr]{%
0.00143665815258903	0.000415196865978549\\
0.00143626535833052	0.000429896769863785\\
0.00143475282867046	0.000450519703825262\\
0.00143293009981114	0.00046549066117573\\
0.00143115255664776	0.000475813256555725\\
0.00142892158373648	0.000485645686787124\\
0.00142327161992456	0.000505869693634385\\
0.00141658082101008	0.000525970631421433\\
0.00141195598395635	0.000538042748709412\\
0.00140755516351337	0.000547974537410286\\
0.00140338490443229	0.000556280887626802\\
0.00139978200072969	0.000562651295896027\\
0.00139627232965815	0.000568195828553697\\
0.00139257689630319	0.000573442682146656\\
0.00138859749227567	0.000578540827318686\\
0.00138428272742652	0.000583541389326504\\
0.00137955431172681	0.000588505616267897\\
0.00137955431172681	0.000588505616267897\\
0.00137418269622086	0.000593615295742653\\
0.00136749093451528	0.000599376628494274\\
0.00135890213089589	0.000606072077153129\\
0.00134973459427533	0.000612532388777538\\
0.0013410733835681	0.000618021117436472\\
0.00133234723576559	0.000622963402732477\\
0.00132298835416539	0.000627671842535913\\
0.00131283927996255	0.000632167827335414\\
0.00130175068079053	0.000636441315549385\\
0.00128938330235212	0.000640521027244044\\
0.00127272827567336	0.000645071318717164\\
0.001239223656723	0.000652826329399127\\
0.00121181294625974	0.000657515986538869\\
0.0011948709631463	0.000659612394858911\\
0.00117924446441914	0.000660770327164558\\
0.00116465364423377	0.000661132293899478\\
};
\addplot [color=mycolor4, dotted, line width=2.0pt]
  table[row sep=crcr]{%
0.00145175374728852	0.000712138617643247\\
0.00144806552022841	0.000715650077897501\\
0.00144419410443229	0.000718989447528968\\
0.00144013371659298	0.000722157982721491\\
0.00143587610825277	0.000725156273566722\\
0.00143140482272577	0.000727987297831363\\
0.00142670891913738	0.000730646898054616\\
0.00142176925584083	0.000733132631242624\\
0.00141655379316094	0.000735443677233748\\
0.00141103349490505	0.000737572040836793\\
0.00140517129493981	0.000739507419047906\\
0.00139891310187375	0.000741237867860969\\
0.00139220154369391	0.000742743267701533\\
0.00138496439763641	0.000743996938126859\\
0.0013771137514211	0.000744962703532665\\
0.00136853990585377	0.000745592064888499\\
0.00135911155012479	0.000745819620275147\\
0.000934187728827284	0.000745653154728374\\
0.000928012911193149	0.000745200409979498\\
0.000922456865460068	0.000744517424556519\\
0.000917421009822961	0.000743645523110819\\
0.000912832277376057	0.000742616675124022\\
0.000908626350511491	0.00074145409877695\\
0.000904751550943793	0.00074017518558608\\
0.000901167410918195	0.000738793606292529\\
0.000897836762226826	0.000737318079698494\\
0.000894730199791747	0.000735755094341854\\
0.000891821868056294	0.000734108147883148\\
0.000889091406840098	0.000732379592931443\\
0.000886519566943525	0.000730568689085288\\
0.000884716494186002	0.000729172459308234\\
0.000883213275809347	0.000727931011978652\\
0.00088137384608762	0.000726291004326046\\
0.000877105184234858	0.000657669335130148\\
0.000868779169627044	0.000451322305708609\\
0.000859278802721471	0.00020756981136125\\
0.000859216629255504	0.000195032344348207\\
0.00085972887866163	0.000183876743941325\\
0.000860516658984232	0.000175630380822671\\
0.000861632818316307	0.000167793668822133\\
0.00086315707658967	0.000159921188337145\\
0.000864789234581575	0.000153253650498972\\
0.000866811932488616	0.000146491998024455\\
0.00086883762385246	0.000140770733690998\\
0.000871035370689725	0.000135395325402882\\
0.000873393574769313	0.000130333101192121\\
0.000875897055486515	0.000125566912965659\\
0.000878527701202053	0.000121088928592918\\
0.000881385633076845	0.000116720731527713\\
0.000884020893818647	0.0001130770858632\\
0.000886321161444745	0.000110197361051585\\
0.000898351056422359	9.78610012968871e-05\\
0.000901901296192659	9.46800468924414e-05\\
0.000907314830212936	9.0577092539252e-05\\
0.000911918881765063	8.74958718784481e-05\\
0.00091655599180857	8.4726844148044e-05\\
0.000921573148532702	8.20656229455337e-05\\
0.000926196587980724	7.98892038939097e-05\\
0.000931365055731019	7.77424909679372e-05\\
0.000936609977830356	7.58344908886277e-05\\
0.000943587539873934	7.37060966302058e-05\\
0.000950226009355891	7.20523121823955e-05\\
0.000955326860908372	7.10122549005254e-05\\
0.000963072139371514	6.98099915465556e-05\\
0.000969832012280516	6.90544247771336e-05\\
0.000977527091520709	6.85621986428993e-05\\
0.0012130854863295	6.84254518863658e-05\\
0.00141791071928	6.96269161501433e-05\\
0.00142312391154118	7.05296931063686e-05\\
0.00142790940996811	7.16027851071856e-05\\
0.00143232897965355	7.2824540624029e-05\\
0.00143642224036855	7.41756620389875e-05\\
0.0014402317904773	7.56441857939603e-05\\
0.00144378785261983	7.72196357042662e-05\\
0.0014471187471446	7.88955489119111e-05\\
0.00145024686175705	8.06670009665958e-05\\
0.00145319228305501	8.25316293149959e-05\\
0.00145597392620647	8.44902046657583e-05\\
0.0014586063651231	8.65442027981609e-05\\
0.00146110210401823	8.86967904894178e-05\\
0.00146347060971972	9.09516238327398e-05\\
0.00146347270456283	9.09537163181542e-05\\
0.00146572500774135	9.33194750133863e-05\\
0.00146786828031155	9.58043438357939e-05\\
0.00146990747666019	9.84176959452521e-05\\
0.00147184750716623	0.00010117285184622\\
0.00147369058907476	0.000104084138471012\\
0.00147543748873705	0.000107168864311114\\
0.0014770869188056	0.000110447064687982\\
0.00147863604733218	0.000113943643453907\\
0.00148007869334683	0.000117685994871893\\
0.00148140657863155	0.000121708890609374\\
0.0014826071640195	0.000126051874666176\\
0.00148366437255438	0.000130766534909269\\
0.00148455607452328	0.000135916261268162\\
0.00148525193618964	0.000141576944272346\\
0.00148571220804542	0.000147855652949098\\
0.00148588155071661	0.000614577767814696\\
0.00148561238578428	0.00062581568874115\\
0.00148489612765628	0.000635596916778479\\
0.00148383317806433	0.000644248148062585\\
0.00148249253444526	0.000651991901170431\\
0.00148091800089575	0.000659014199832266\\
0.00147913751365726	0.000665454930191529\\
0.00147716830960976	0.000671420239806806\\
0.0014750228946532	0.000676984976134035\\
0.00147270587491031	0.000682214166201574\\
0.00147022084201018	0.000687152535646436\\
0.00146756790409176	0.000691836645914274\\
0.00146474747581704	0.000696291300259663\\
0.0014617564730793	0.000700538627867053\\
0.00145859606643639	0.000704588534607951\\
0.00145526262624861	0.000708453058056901\\
0.00145175374728852	0.000712138617643247\\
};
\addlegendentry{$L = 4$}

\addplot[only marks, mark=o, mark options={}, mark size=3.5355pt, draw=mycolor4, forget plot] table[row sep=crcr]{%
x	y\\
0.00147270587491031	0.000682214166201574\\
};
\addplot[only marks, mark=x, mark options={}, mark size=3.5355pt, draw=mycolor4, forget plot] table[row sep=crcr]{%
x	y\\
0.00148588155071661	0.000614577767814696\\
};
\addplot [color=mycolor4, line width=2.0pt, forget plot]
  table[row sep=crcr]{%
0.00148588155071661	0.000614577767814696\\
0.00148561238578428	0.00062581568874115\\
0.00148489612765628	0.000635596916778479\\
0.00148383317806433	0.000644248148062585\\
0.00148249253444526	0.000651991901170431\\
0.00148091800089575	0.000659014199832266\\
0.00147913751365726	0.000665454930191529\\
0.00147716830960976	0.000671420239806806\\
0.0014750228946532	0.000676984976134035\\
0.00147270587491031	0.000682214166201574\\
0.00147022084201018	0.000687152535646436\\
0.00146756790409176	0.000691836645914274\\
0.00146474747581704	0.000696291300259663\\
0.0014617564730793	0.000700538627867053\\
0.00145859606643639	0.000704588534607951\\
0.00145526262624861	0.000708453058056901\\
0.00145175374728852	0.000712138617643247\\
0.00145175374728852	0.000712138617643247\\
0.00144806552022841	0.000715650077897501\\
0.00144419410443229	0.000718989447528968\\
0.00144013371659298	0.000722157982721491\\
0.00143587610825277	0.000725156273566722\\
0.00143140482272577	0.000727987297831363\\
0.00142670891913738	0.000730646898054616\\
0.00142176925584083	0.000733132631242624\\
0.00141655379316094	0.000735443677233748\\
0.00141103349490505	0.000737572040836793\\
0.00140517129493981	0.000739507419047906\\
0.00139891310187375	0.000741237867860969\\
0.00139220154369391	0.000742743267701533\\
0.00138496439763641	0.000743996938126859\\
0.0013771137514211	0.000744962703532665\\
0.00136853990585377	0.000745592064888499\\
0.00135911155012479	0.000745819620275147\\
};
\addplot [color=mycolor7, dotted, line width=2.0pt]
  table[row sep=crcr]{%
0.00149046290152123	0.000749029254155564\\
0.0014900242269345	0.00074944700571706\\
0.00148956595140804	0.000749842385878001\\
0.00148908527851801	0.000750217514822503\\
0.00148858061179832	0.000750572888805431\\
0.00148804928788216	0.00075090929458397\\
0.00148748523631993	0.000751228774981614\\
0.00148688472888253	0.00075153094833365\\
0.00148624325212002	0.000751815163525175\\
0.00148555343341602	0.000752081095560801\\
0.0014848065906605	0.000752327583569775\\
0.00148399652283026	0.00075255146201138\\
0.00148310950182114	0.000752750321984406\\
0.00148213188457328	0.0007529195766878\\
0.00148104571787482	0.000753053088092032\\
0.00147982927991805	0.000753142253382241\\
0.00147845410444853	0.00075317530055406\\
0.000856399484299476	0.000753079955662908\\
0.00085561190314361	0.00075299087227254\\
0.000855168383252795	0.000752894042526675\\
0.000854808259349123	0.000752794560377031\\
0.000854476998753862	0.00075268531896617\\
0.000854171826735235	0.000752567788883797\\
0.000853892942724351	0.000752444475390079\\
0.000853623914142379	0.000752309358535738\\
0.00085337810464549	0.000752170293463684\\
0.000853143937523826	0.00075202220008706\\
0.000852924266392751	0.000751867645258384\\
0.000852713985039399	0.000751703649951764\\
0.000852573654187733	0.000751583253736799\\
0.000852509600662535	0.000751525774234404\\
0.00085216606910422	0.000751009418407836\\
0.000852087035106126	0.00075082980172535\\
0.000851924200217365	0.000750422497451242\\
0.000850371834940363	7.83330461071049e-05\\
0.000850367123973061	7.59963655680903e-05\\
0.000850452392963229	7.44253717148826e-05\\
0.000850572412869088	7.32556626151417e-05\\
0.000850729803715266	7.22122573874861e-05\\
0.000850963683544685	7.10836145599801e-05\\
0.00085117690039885	7.02749238911815e-05\\
0.00085140946312952	6.95414157594678e-05\\
0.000851663484465899	6.88630797042573e-05\\
0.000851921250054968	6.82738127080908e-05\\
0.000852235164654366	6.76406847805733e-05\\
0.00085244098022708	6.72916006764241e-05\\
0.000852575638270119	6.71441261169647e-05\\
0.000855293128156684	6.42194441589772e-05\\
0.000855628679613259	6.39493355269094e-05\\
0.000855976800622717	6.36976812377671e-05\\
0.000856342866396364	6.34551290253909e-05\\
0.00085673172979605	6.32317378177175e-05\\
0.000857371652381497	6.28771846181823e-05\\
0.00085823517243215	6.24998316958764e-05\\
0.000858658557351719	6.23446229182737e-05\\
0.000859302431077715	6.21722239703818e-05\\
0.00086014829123621	6.19824843312733e-05\\
0.000861443678839111	6.1765230282703e-05\\
0.000862185836573078	6.16446201569014e-05\\
0.000865019299296782	6.14668514592799e-05\\
0.00148652830672264	6.24988491252918e-05\\
0.00148756147713045	6.27715119705714e-05\\
0.0014880677690769	6.29637627044874e-05\\
0.00148847364518094	6.31658985332232e-05\\
0.0014889012889211	6.340163880027e-05\\
0.00148958946767849	6.38226172290342e-05\\
0.0014899657954339	6.40859123640802e-05\\
0.00149044958736498	6.44584450088417e-05\\
0.0014907766583995	6.47382107454055e-05\\
0.00149124853041068	6.51802683567342e-05\\
0.00149126387281642	6.51956454924086e-05\\
0.00149162376609705	6.5574211271998e-05\\
0.00149197557836475	6.59824758572283e-05\\
0.00149230879327787	6.64098900092647e-05\\
0.00149270640416903	6.69834976486238e-05\\
0.00149293882889736	6.73543860860554e-05\\
0.00149323394201324	6.78758388121919e-05\\
0.00149351604468895	6.84368795556996e-05\\
0.0014937864155005	6.90478292669155e-05\\
0.00149403855836293	6.97026338361394e-05\\
0.0014942762195696	7.04230012239587e-05\\
0.00149449628185129	7.1219554848062e-05\\
0.00149469367891662	7.21013984572583e-05\\
0.00149486425226757	7.3087028075443e-05\\
0.00149500061683729	7.41978783389704e-05\\
0.00149509793153764	0.000410863543510539\\
0.00149512842708102	0.000734389327803829\\
0.00149508222545898	0.000736322135349975\\
0.00149496077692915	0.000737981317030748\\
0.00149478459564268	0.0007394154266034\\
0.00149456731440495	0.000740671168966442\\
0.0014943182183291	0.000741782777794207\\
0.001494043997576	0.000742775138080022\\
0.00149375005962558	0.000743665850177549\\
0.00149343908306023	0.000744472718100399\\
0.00149311330100405	0.000745208160429862\\
0.00149277399382837	0.000745882586790691\\
0.00149242164900008	0.000746504763466269\\
0.00149205708068219	0.000747080593327281\\
0.00149167950055224	0.000747616787479555\\
0.00149128840330427	0.000748117964075322\\
0.00149088263226683	0.000748588356485789\\
0.00149046290152123	0.000749029254155564\\
};
\addlegendentry{$L = 16$}

\addplot[only marks, mark=o, mark options={}, mark size=3.5355pt, draw=mycolor7, forget plot] table[row sep=crcr]{%
x	y\\
0.00149311330100405	0.000745208160429862\\
};
\addplot[only marks, mark=x, mark options={}, mark size=3.5355pt, draw=mycolor7, forget plot] table[row sep=crcr]{%
x	y\\
0.00149512842708102	0.000734389327803829\\
};
\addplot [color=mycolor7, line width=2.0pt, forget plot]
  table[row sep=crcr]{%
0.00149512842708102	0.000734389327803829\\
0.00149508222545898	0.000736322135349975\\
0.00149496077692915	0.000737981317030748\\
0.00149478459564268	0.0007394154266034\\
0.00149456731440495	0.000740671168966442\\
0.0014943182183291	0.000741782777794207\\
0.001494043997576	0.000742775138080022\\
0.00149375005962558	0.000743665850177549\\
0.00149343908306023	0.000744472718100399\\
0.00149311330100405	0.000745208160429862\\
0.00149277399382837	0.000745882586790691\\
0.00149242164900008	0.000746504763466269\\
0.00149205708068219	0.000747080593327281\\
0.00149167950055224	0.000747616787479555\\
0.00149128840330427	0.000748117964075322\\
0.00149088263226683	0.000748588356485789\\
0.00149046290152123	0.000749029254155564\\
0.00149046290152123	0.000749029254155564\\
0.0014900242269345	0.00074944700571706\\
0.00148956595140804	0.000749842385878001\\
0.00148908527851801	0.000750217514822503\\
0.00148858061179832	0.000750572888805431\\
0.00148804928788216	0.00075090929458397\\
0.00148748523631993	0.000751228774981614\\
0.00148688472888253	0.00075153094833365\\
0.00148624325212002	0.000751815163525175\\
0.00148555343341602	0.000752081095560801\\
0.0014848065906605	0.000752327583569775\\
0.00148399652283026	0.00075255146201138\\
0.00148310950182114	0.000752750321984406\\
0.00148213188457328	0.0007529195766878\\
0.00148104571787482	0.000753053088092032\\
0.00147982927991805	0.000753142253382241\\
0.00147845410444853	0.00075317530055406\\
};
\addplot [color=mycolor9, dotted, line width=2.0pt]
  table[row sep=crcr]{%
0.00149386639048453	0.000755049655210065\\
0.00149371321782535	0.000755195521204306\\
0.0014935545341768	0.000755332417443745\\
0.00149338938893506	0.000755461305300265\\
0.00149321694115587	0.000755582756001372\\
0.00149303620516568	0.00075569719649723\\
0.0014928457894547	0.000755805039377688\\
0.00149264498691478	0.000755906076228377\\
0.00149243167931481	0.000756000585561203\\
0.00149220410845542	0.000756088315803518\\
0.00149196000195171	0.000756168891611219\\
0.00149169722763798	0.000756241543642034\\
0.00149141111040051	0.000756305712823123\\
0.0014910979926927	0.00075635991046275\\
0.00149075487766604	0.000756402060862943\\
0.00149037404242127	0.000756429972395182\\
0.00148995046881103	0.000756440229046408\\
0.000850713473664349	0.000756416718845723\\
0.000850663212280196	0.000756406361226111\\
0.000850618768517117	0.000756394757827047\\
0.000850579412502709	0.000756382153355347\\
0.000850543896081195	0.00075636862630528\\
0.000850511372754104	0.000756354265254022\\
0.000850481500978184	0.00075633924846132\\
0.000850453460881014	0.000756323372912526\\
0.000850427409551194	0.000756306881026027\\
0.000850402900072667	0.000756289627563039\\
0.000850379689043704	0.000756271519490907\\
0.000850362461715169	0.000756256819835924\\
0.000850342425569706	0.000756238193485083\\
0.000850130707459294	6.11099280701719e-05\\
0.000850148647025596	6.06874539733784e-05\\
0.000850172021980861	6.04826107331783e-05\\
0.000850196042692511	6.03508340672045e-05\\
0.000850224679739394	6.02346652794101e-05\\
0.000850252156330193	6.01552296980166e-05\\
0.000850927713023346	5.90066160227148e-05\\
0.000851015950565683	5.89088955082686e-05\\
0.000851139093088481	5.87801051186764e-05\\
0.000851355363190468	5.86051095368014e-05\\
0.000851402419252119	5.85682203487598e-05\\
0.000851559198154111	5.84712629153387e-05\\
0.000851614280135624	5.84404156361093e-05\\
0.000851839061286514	5.83251206964467e-05\\
0.000852079176178841	5.82432954935825e-05\\
0.000852544271774776	5.81452091614017e-05\\
0.000853026106262565	5.81124124465135e-05\\
0.00148202778580132	5.87561193247614e-05\\
0.00148995086077685	5.92432048779224e-05\\
0.00149164411536169	5.93937741840165e-05\\
0.00149196951360547	5.9480314746377e-05\\
0.00149257268125237	5.98244890280208e-05\\
0.00149298250414031	6.01503582395203e-05\\
0.00149312170290923	6.02768484821111e-05\\
0.00149360086647569	6.07942590234414e-05\\
0.00149388618441288	6.11328148597379e-05\\
0.00149410150178149	6.14304027310704e-05\\
0.00149413570811437	6.14788441769673e-05\\
0.00149415693008344	6.15102907975543e-05\\
0.00149438279846284	6.18880728895224e-05\\
0.00149467166868244	6.24436727957114e-05\\
0.00149478452214937	6.26982811479396e-05\\
0.0014948113061461	6.27606485455104e-05\\
0.00149495176038534	6.31259618423602e-05\\
0.00149547987929816	0.000410138142976125\\
0.001495567105556	0.000749580677176887\\
0.00149554917810926	0.000750331947706176\\
0.00149550265362032	0.000750967918632245\\
0.00149543511188738	0.000751517850727611\\
0.00149535186555685	0.000751998089470744\\
0.0014952590545945	0.000752411703974411\\
0.00149515652356265	0.000752782740275688\\
0.00149504724871517	0.000753113853451888\\
0.00149493239598821	0.000753411807439089\\
0.00149481287112702	0.00075368157241838\\
0.00149468938022518	0.00075392699028808\\
0.00149456201677617	0.000754151891844859\\
0.00149443084208897	0.000754359097760289\\
0.00149429581983774	0.000754550852231667\\
0.00149415684568532	0.000754728949529443\\
0.00149401388536234	0.000754894698859366\\
0.00149386639048453	0.000755049655210065\\
};
\addlegendentry{$L = 32$}

\addplot[only marks, mark=o, mark options={}, mark size=3.5355pt, draw=mycolor9, forget plot] table[row sep=crcr]{%
x	y\\
0.00149468938022518	0.00075392699028808\\
};
\addplot[only marks, mark=x, mark options={}, mark size=3.5355pt, draw=mycolor9, forget plot] table[row sep=crcr]{%
x	y\\
0.001495567105556	0.000749580677176887\\
};
\addplot [color=mycolor9, line width=2.0pt, forget plot]
  table[row sep=crcr]{%
0.001495567105556	0.000749580677176887\\
0.00149554917810926	0.000750331947706176\\
0.00149550265362032	0.000750967918632245\\
0.00149543511188738	0.000751517850727611\\
0.00149535186555685	0.000751998089470744\\
0.0014952590545945	0.000752411703974411\\
0.00149515652356265	0.000752782740275688\\
0.00149504724871517	0.000753113853451888\\
0.00149493239598821	0.000753411807439089\\
0.00149481287112702	0.00075368157241838\\
0.00149468938022518	0.00075392699028808\\
0.00149456201677617	0.000754151891844859\\
0.00149443084208897	0.000754359097760289\\
0.00149429581983774	0.000754550852231667\\
0.00149415684568532	0.000754728949529443\\
0.00149401388536234	0.000754894698859366\\
0.00149386639048453	0.000755049655210065\\
0.00149386639048453	0.000755049655210065\\
0.00149371321782535	0.000755195521204306\\
0.0014935545341768	0.000755332417443745\\
0.00149338938893506	0.000755461305300265\\
0.00149321694115587	0.000755582756001372\\
0.00149303620516568	0.00075569719649723\\
0.0014928457894547	0.000755805039377688\\
0.00149264498691478	0.000755906076228377\\
0.00149243167931481	0.000756000585561203\\
0.00149220410845542	0.000756088315803518\\
0.00149196000195171	0.000756168891611219\\
0.00149169722763798	0.000756241543642034\\
0.00149141111040051	0.000756305712823123\\
0.0014910979926927	0.00075635991046275\\
0.00149075487766604	0.000756402060862943\\
0.00149037404242127	0.000756429972395182\\
0.00148995046881103	0.000756440229046408\\
};

\addlegendimage{only marks, mark=o, mark options={solid}, mark size=2pt, line width=2pt, draw=mycolor11}
\addlegendentry{P-max}

\addlegendimage{only marks, mark=x, mark options={solid}, mark size=3pt, line width=2pt, draw=mycolor11}
\addlegendentry{E-max}

\addplot [draw=none, color=mycolor11, line width=2.0pt]
  table[row sep=crcr]{%
0	0\\
};
\addlegendentry{R-max}

\end{axis}
\end{tikzpicture}%

%% file: assets/simulation/pc_singular_pareto_sx64.tex
%
%
\definecolor{mycolor1}{rgb}{0.00000,0.44706,0.74118}%
\definecolor{mycolor2}{rgb}{0.92900,0.69400,0.12500}%
\definecolor{mycolor3}{rgb}{0.49400,0.18400,0.55600}%
\definecolor{mycolor4}{rgb}{0.85098,0.32549,0.09804}%
\definecolor{mycolor5}{rgb}{0.63500,0.07800,0.18400}%
\definecolor{mycolor6}{rgb}{0.00000,0.44700,0.74100}%
\definecolor{mycolor7}{rgb}{0.92941,0.69412,0.12549}%
\definecolor{mycolor8}{rgb}{0.46600,0.67400,0.18800}%
\definecolor{mycolor9}{rgb}{0.49412,0.18431,0.55686}%
\definecolor{mycolor10}{rgb}{0.85000,0.32500,0.09800}%
\definecolor{mycolor11}{rgb}{0.50196,0.50196,0.50196}%
\begin{tikzpicture}

\begin{axis}[%
width=9.509cm,
height=7.5cm,
at={(0cm,0cm)},
scale only axis,
unbounded coords=jump,
xmin=0.000665933747270548,
xmax=0.00180816416675568,
xlabel={$\sigma_1(\mathbf{H})$},
ymin=1.10708260774394e-20,
ymax=0.000931271909489233,
ylabel={$\sigma_2(\mathbf{H})$},
axis background/.style={fill=white},
xmajorgrids,
ymajorgrids,
legend style={at={(0.03,0.97)}, anchor=north west, legend cell align=left, align=left, draw=white!15!black},
every axis plot/.append style={line width=1.5pt}
]
\addplot[only marks, mark=triangle, mark options={}, mark size=2.3570pt, draw=black] table[row sep=crcr]{%
x	y\\
0.00123451554249979	0.000261671419241521\\
};
\addlegendentry{Direct}

\addplot [color=mycolor1, dotted, line width=2.0pt]
  table[row sep=crcr]{%
0.00159608285697502	0.000664114392083429\\
0.00158691664294714	0.000672837941128764\\
0.00157660050010728	0.000681731690268777\\
0.00156590259762528	0.000690085825238813\\
0.00155467994239563	0.000697976645701225\\
0.00153986973643909	0.000707372645028292\\
0.0015148645023757	0.000721449433364078\\
0.00149697676532785	0.000730394195338623\\
0.0014811063630021	0.000737471285391353\\
0.00146705261153674	0.000742888836825103\\
0.00145145312317651	0.000748034879132762\\
0.00143354032724506	0.000752981116885134\\
0.00141219815168891	0.00075775662502572\\
0.00138599656893751	0.000762282978521517\\
0.001355981991655	0.00076597328234428\\
0.00132155867810352	0.000768494850647513\\
0.00127159845792326	0.000769595351798804\\
0.00123642786506839	0.000768807421978087\\
0.00121422536014095	0.000767187754726459\\
0.00118743020096022	0.000763797105859703\\
0.00116236115681128	0.000759517337937116\\
0.00114216255507682	0.000754981192412643\\
0.00112132876332356	0.000749208839458789\\
0.0010987179355452	0.000741727590503771\\
0.001073267755154	0.000731896088791152\\
0.00104840588147425	0.00072089161304077\\
0.00103123882758884	0.000712332026035008\\
0.00102072542642717	0.00070661151546976\\
0.00101127064407598	0.000700865695981975\\
0.00100290573808909	0.000695274516366674\\
0.00099405869556404	0.000688750411001148\\
0.000985723184951612	0.000681989259163176\\
0.000977942269806648	0.000675005939163037\\
0.000883365237789003	0.000578594721583551\\
0.000875259823033198	0.000568804316998415\\
0.00086636420891398	0.000556929660283238\\
0.000858029700529914	0.000544803768004766\\
0.000849145080144272	0.00053057937604649\\
0.000840385049788102	0.00051511967478722\\
0.000831240490517266	0.000497160206301654\\
0.000823405081875965	0.000479893545133291\\
0.000814633564461672	0.000458113314889421\\
0.000806977608550476	0.000436163726292382\\
0.000800072243550595	0.000412970991852368\\
0.000793559072159092	0.000386530995276439\\
0.000788411273090975	0.000360210129693543\\
0.000784467319171457	0.000332879132259079\\
0.000781370109809877	0.000300494465576911\\
0.000779229604973693	0.000255973167522945\\
0.000778843765690314	0.000189061764117299\\
0.000780851644081482	0.00014653579410713\\
0.00078386863842555	0.000116085317999348\\
0.000787435958310619	9.41678760729649e-05\\
0.00079298844062422	6.9432611338902e-05\\
0.000797722113984264	5.34238131636791e-05\\
0.000802073424201546	4.09329153574363e-05\\
0.000806203726326662	3.06917584825308e-05\\
0.000813014934977753	2.20160208093019e-05\\
0.000824950988932121	7.10447651041556e-06\\
0.000827696888811626	3.85134078631187e-06\\
0.000830547864235395	1.21674953920938e-06\\
0.000831965275229041	2.44617826405831e-07\\
0.000832649165431366	1.79133867425279e-08\\
0.000832828850348146	1.10708260774394e-20\\
0.00161378237741308	6.65749607101577e-19\\
0.00161378237741355	5.8393880144134e-16\\
0.00161378237849435	1.32872382628457e-12\\
0.00161639227547958	3.40455067607412e-06\\
0.00163014580907626	2.28971513147335e-05\\
0.00163990293192521	3.82518884070967e-05\\
0.00164673659382605	5.02988168914937e-05\\
0.00165256368270938	6.18735109090836e-05\\
0.00165770600095059	7.34703347459089e-05\\
0.00166401148822954	8.99928952340316e-05\\
0.00167059283856931	0.000109936635373391\\
0.00167658908453108	0.00013162628400113\\
0.00168183527674099	0.000155018775367617\\
0.00168631721115301	0.000180982050508911\\
0.00169157776515169	0.000222197499567037\\
0.0016938416476388	0.000253372151620205\\
0.00169524815187732	0.000327486553761353\\
0.00169398163691352	0.000384104665758756\\
0.00169113344593976	0.000423595412164682\\
0.00168746779562959	0.00045341619951097\\
0.00168166868784352	0.000486331065994069\\
0.00167573939279825	0.000513052824819554\\
0.00166982471571825	0.000534440995798568\\
0.00166328466709128	0.000554261223027059\\
0.00165659661416354	0.000571625272087831\\
0.00164994826526868	0.00058664202350194\\
0.00164324082244803	0.000599976257139259\\
0.00163623992074054	0.000612337398149498\\
0.00162880459408019	0.00062408033199365\\
0.00162098454625344	0.000635186241435433\\
0.00161290650953524	0.00064554019696715\\
0.00160463337370114	0.000655133158519538\\
0.00159608285697502	0.000664114392083429\\
};
\addlegendentry{$L = 1$}

\addplot[only marks, mark=o, mark options={}, mark size=3.5355pt, draw=mycolor1, forget plot] table[row sep=crcr]{%
x	y\\
0.00166982471571825	0.000534440995798568\\
};
\addplot[only marks, mark=x, mark options={}, mark size=3.5355pt, draw=mycolor1, forget plot] table[row sep=crcr]{%
x	y\\
0.00169524815187732	0.000327486553761353\\
};
\addplot [color=mycolor1, line width=2.0pt, forget plot]
  table[row sep=crcr]{%
0.00169524815187732	0.000327486553761353\\
0.00169398163691352	0.000384104665758756\\
0.00169113344593976	0.000423595412164682\\
0.00168746779562959	0.00045341619951097\\
0.00168166868784352	0.000486331065994069\\
0.00167573939279825	0.000513052824819554\\
0.00166982471571825	0.000534440995798568\\
0.00166328466709128	0.000554261223027059\\
0.00165659661416354	0.000571625272087831\\
0.00164994826526868	0.00058664202350194\\
0.00164324082244803	0.000599976257139259\\
0.00163623992074054	0.000612337398149498\\
0.00162880459408019	0.00062408033199365\\
0.00162098454625344	0.000635186241435433\\
0.00161290650953524	0.00064554019696715\\
0.00160463337370114	0.000655133158519538\\
0.00159608285697502	0.000664114392083429\\
0.00159608285697502	0.000664114392083429\\
0.00158691664294714	0.000672837941128764\\
0.00157660050010728	0.000681731690268777\\
0.00156590259762528	0.000690085825238813\\
0.00155467994239563	0.000697976645701225\\
0.00153986973643909	0.000707372645028292\\
0.0015148645023757	0.000721449433364078\\
0.00149697676532785	0.000730394195338623\\
0.0014811063630021	0.000737471285391353\\
0.00146705261153674	0.000742888836825103\\
0.00145145312317651	0.000748034879132762\\
0.00143354032724506	0.000752981116885134\\
0.00141219815168891	0.00075775662502572\\
0.00138599656893751	0.000762282978521517\\
0.001355981991655	0.00076597328234428\\
0.00132155867810352	0.000768494850647513\\
0.00127159845792326	0.000769595351798804\\
};
\addplot [color=mycolor4, dotted, line width=2.0pt]
  table[row sep=crcr]{%
0.0017386594129955	0.000860048272090083\\
0.00173620516066716	0.000862385276306713\\
0.00173365845469748	0.000864582194303698\\
0.00173100324095088	0.000866654293858911\\
0.00172822270935142	0.000868612367139377\\
0.00172529641442896	0.000870465091689542\\
0.00172219949784328	0.000872218997274339\\
0.00171890327802487	0.000873877536139426\\
0.00171537346710331	0.000875441416860692\\
0.00171156313466259	0.00087691025673575\\
0.00170742153093652	0.000878277201093898\\
0.00170288236322903	0.00087953179778277\\
0.00169786056348726	0.000880657548119832\\
0.00169224883215724	0.000881628811261486\\
0.00168591886705108	0.000882406429632005\\
0.00167869706302702	0.000882935175885586\\
0.00167037980039834	0.000883134115394112\\
0.00102052209736536	0.000882801692620656\\
0.00101646047276367	0.000882503284497043\\
0.00101267238345347	0.000882037044403964\\
0.00100912045073618	0.000881421520159389\\
0.00100577415829649	0.000880670753564986\\
0.00100260461810987	0.000879794177847139\\
0.000999591303638097	0.00087879914985014\\
0.000996714501147741	0.000877689800195453\\
0.000993956695076467	0.000876467622492497\\
0.000991304163226718	0.000875132632656254\\
0.00098874401757588	0.000873682418028017\\
0.00098626581874889	0.000872113116072078\\
0.000983859443434805	0.000870418263617197\\
0.000982667940276653	0.000869508063587063\\
0.000980791270287472	0.00086797495643917\\
0.000978764853396329	0.000866174303123083\\
0.00073858103322915	0.000625587747554838\\
0.000736804562629995	0.00062341605253277\\
0.000735216509742291	0.000621195965101899\\
0.000733326042401016	0.000618190330161238\\
0.000731462763390808	0.000614635580866718\\
0.00072960197688474	0.000610261186413097\\
0.000727481471851482	0.000604194613646108\\
0.000725774643635242	0.00059698869144714\\
0.000724411832896014	0.000588243050464741\\
0.000722671507844784	0.00057242143535709\\
0.000720908493707654	0.000551316307025198\\
0.000717663020121559	0.000480820238180472\\
0.00071027146682115	1.85577797620156e-07\\
0.000710272133741029	2.04326485221653e-20\\
0.000710272133741029	1.24468938843598e-20\\
0.00176612080949286	4.10800016469333e-20\\
0.00176612080949286	1.85966565899392e-18\\
0.00176612080954112	1.05266522974633e-12\\
0.00176658705241273	0.000768319728234086\\
0.00176627636802428	0.000781399086076133\\
0.001765475267997	0.000792357337836094\\
0.00176432796955949	0.000801704108891192\\
0.00176293138974357	0.000809777996388858\\
0.00176135178909172	0.000816827834029172\\
0.00175963734811238	0.000823033257292548\\
0.00175782002453809	0.000828541533511611\\
0.00175592216240292	0.000833466677211746\\
0.00175395850137951	0.000837900348712894\\
0.00175193781323402	0.000841917498031388\\
0.00174986495888347	0.000845578643152257\\
0.0017477395215548	0.000848936553024065\\
0.00174556170527002	0.000852029858633396\\
0.00174332601024846	0.000854895377084003\\
0.00174102703287014	0.000857561023545491\\
0.0017386594129955	0.000860048272090083\\
};
\addlegendentry{$L = 4$}

\addplot[only marks, mark=o, mark options={}, mark size=3.5355pt, draw=mycolor4, forget plot] table[row sep=crcr]{%
x	y\\
0.00175395850137951	0.000837900348712894\\
};
\addplot[only marks, mark=x, mark options={}, mark size=3.5355pt, draw=mycolor4, forget plot] table[row sep=crcr]{%
x	y\\
0.00176658705241273	0.000768319728234086\\
};
\addplot [color=mycolor4, line width=2.0pt, forget plot]
  table[row sep=crcr]{%
0.00176658705241273	0.000768319728234086\\
0.00176627636802428	0.000781399086076133\\
0.001765475267997	0.000792357337836094\\
0.00176432796955949	0.000801704108891192\\
0.00176293138974357	0.000809777996388858\\
0.00176135178909172	0.000816827834029172\\
0.00175963734811238	0.000823033257292548\\
0.00175782002453809	0.000828541533511611\\
0.00175592216240292	0.000833466677211746\\
0.00175395850137951	0.000837900348712894\\
0.00175193781323402	0.000841917498031388\\
0.00174986495888347	0.000845578643152257\\
0.0017477395215548	0.000848936553024065\\
0.00174556170527002	0.000852029858633396\\
0.00174332601024846	0.000854895377084003\\
0.00174102703287014	0.000857561023545491\\
0.0017386594129955	0.000860048272090083\\
0.0017386594129955	0.000860048272090083\\
0.00173620516066716	0.000862385276306713\\
0.00173365845469748	0.000864582194303698\\
0.00173100324095088	0.000866654293858911\\
0.00172822270935142	0.000868612367139377\\
0.00172529641442896	0.000870465091689542\\
0.00172219949784328	0.000872218997274339\\
0.00171890327802487	0.000873877536139426\\
0.00171537346710331	0.000875441416860692\\
0.00171156313466259	0.00087691025673575\\
0.00170742153093652	0.000878277201093898\\
0.00170288236322903	0.00087953179778277\\
0.00169786056348726	0.000880657548119832\\
0.00169224883215724	0.000881628811261486\\
0.00168591886705108	0.000882406429632005\\
0.00167869706302702	0.000882935175885586\\
0.00167037980039834	0.000883134115394112\\
};
\addplot [color=mycolor7, dotted, line width=2.0pt]
  table[row sep=crcr]{%
0.00179380276845171	0.00091689979762741\\
0.00179330381640944	0.000917374896950445\\
0.00179278266048231	0.000917824483975737\\
0.00179223588641608	0.000918251208265826\\
0.00179166118529593	0.00091865597284368\\
0.00179105243790838	0.000919041439304825\\
0.00179040687137795	0.000919407052815547\\
0.00178971966764065	0.000919752827018602\\
0.00178898345155905	0.000920079015279063\\
0.00178819041326094	0.000920384737764874\\
0.00178733064593907	0.000920668516231091\\
0.00178639456888193	0.000920927244200767\\
0.00178536783735916	0.000921157441801043\\
0.0017842320969085	0.000921354081194788\\
0.00178296621498302	0.000921509663813811\\
0.00178154397059403	0.000921613885204562\\
0.0017799292148076	0.000921652660192432\\
0.000988173389737295	0.000921582699284711\\
0.000957736273207751	0.000921529495829503\\
0.000956956181711174	0.000921436401691889\\
0.000956262822922715	0.000921316245628706\\
0.000955621424562645	0.000921172308373418\\
0.000955026459647334	0.000921007721403984\\
0.000954469354603954	0.000920823719947534\\
0.000953946831222138	0.000920622178674876\\
0.00095345140828616	0.000920402549695003\\
0.000952983758568837	0.000920167088503862\\
0.000952538188648702	0.000919914641645634\\
0.000952111090952238	0.000919644150536898\\
0.000951701194433098	0.000919355407751773\\
0.000951305944643583	0.000919046852318652\\
0.000950925140027922	0.000918718263904203\\
0.000950848596271687	0.000918647407418934\\
0.000679717985371998	0.000646470303425183\\
0.000679378897894892	0.000645792824605016\\
0.000679064490109163	0.000645032778423549\\
0.000678722747899459	0.000644083164721458\\
0.000677914453470986	0.000641232357239559\\
0.00067750701005809	0.000638207934222391\\
0.000677322919355633	0.000635620156515246\\
0.000677047863285968	0.000630803033852765\\
0.000674975974668598	5.06904442744669e-08\\
0.000674977394216685	2.27657807105469e-20\\
0.000674977394217525	2.0914813109004e-20\\
0.00179895941878977	5.27490978356368e-16\\
0.00179895941879047	2.50903236847347e-14\\
0.00179895941879441	3.04019292668737e-13\\
0.00179895941880352	1.03827816531674e-12\\
0.00179895941881027	1.59019368378821e-12\\
0.00179895941882359	3.20186569698865e-12\\
0.00179895941884063	5.63352168589419e-12\\
0.0017991011722443	0.000900190223937653\\
0.00179904761519291	0.000902434408333122\\
0.00179890811245153	0.000904341359869694\\
0.00179870656384648	0.000905983094296771\\
0.00179845843443343	0.00090741758790198\\
0.00179817474211808	0.00090868352751566\\
0.00179786315604028	0.000909810873569764\\
0.00179752971313305	0.000910821053531535\\
0.00179717732419367	0.000911735177535424\\
0.00179680876466913	0.000912567112422704\\
0.00179642418724526	0.000913331493190543\\
0.0017960252855648	0.000914035854842987\\
0.00179561177040327	0.000914688958760955\\
0.00179518419519686	0.000915296119290815\\
0.00179474074215017	0.000915864399173621\\
0.00179428054460277	0.000916397925030764\\
0.00179380276845171	0.00091689979762741\\
};
\addlegendentry{$L = 16$}

\addplot[only marks, mark=o, mark options={}, mark size=3.5355pt, draw=mycolor7, forget plot] table[row sep=crcr]{%
x	y\\
0.00179642418724526	0.000913331493190543\\
};
\addplot[only marks, mark=x, mark options={}, mark size=3.5355pt, draw=mycolor7, forget plot] table[row sep=crcr]{%
x	y\\
0.0017991011722443	0.000900190223937653\\
};
\addplot [color=mycolor7, line width=2.0pt, forget plot]
  table[row sep=crcr]{%
0.0017991011722443	0.000900190223937653\\
0.00179904761519291	0.000902434408333122\\
0.00179890811245153	0.000904341359869694\\
0.00179870656384648	0.000905983094296771\\
0.00179845843443343	0.00090741758790198\\
0.00179817474211808	0.00090868352751566\\
0.00179786315604028	0.000909810873569764\\
0.00179752971313305	0.000910821053531535\\
0.00179717732419367	0.000911735177535424\\
0.00179680876466913	0.000912567112422704\\
0.00179642418724526	0.000913331493190543\\
0.0017960252855648	0.000914035854842987\\
0.00179561177040327	0.000914688958760955\\
0.00179518419519686	0.000915296119290815\\
0.00179474074215017	0.000915864399173621\\
0.00179428054460277	0.000916397925030764\\
0.00179380276845171	0.00091689979762741\\
0.00179380276845171	0.00091689979762741\\
0.00179330381640944	0.000917374896950445\\
0.00179278266048231	0.000917824483975737\\
0.00179223588641608	0.000918251208265826\\
0.00179166118529593	0.00091865597284368\\
0.00179105243790838	0.000919041439304825\\
0.00179040687137795	0.000919407052815547\\
0.00178971966764065	0.000919752827018602\\
0.00178898345155905	0.000920079015279063\\
0.00178819041326094	0.000920384737764874\\
0.00178733064593907	0.000920668516231091\\
0.00178639456888193	0.000920927244200767\\
0.00178536783735916	0.000921157441801043\\
0.0017842320969085	0.000921354081194788\\
0.00178296621498302	0.000921509663813811\\
0.00178154397059403	0.000921613885204562\\
0.0017799292148076	0.000921652660192432\\
};
\addplot [color=mycolor9, dotted, line width=2.0pt]
  table[row sep=crcr]{%
0.00180555185284818	0.000929399946939874\\
0.00180533514986291	0.000929606320188223\\
0.00180511241911368	0.000929798491198706\\
0.00180488301833069	0.000929977560540128\\
0.00180464543585704	0.000930144931618639\\
0.00180439887046556	0.000930301114559859\\
0.00180414061727167	0.000930447428205666\\
0.00180387197745908	0.000930582648819448\\
0.00180358777824979	0.000930708610274962\\
0.00180328851982513	0.0009308239951381\\
0.00180297041976148	0.000930929010172791\\
0.00180263140440767	0.000931022740000879\\
0.00180226747773386	0.000931104368117137\\
0.00180187450841429	0.000931172452949046\\
0.00180144635722952	0.000931225124403758\\
0.00180097831933697	0.000931259464153654\\
0.00180046348509318	0.000931271909489233\\
0.000944849900941658	0.000931211694645889\\
0.000944839913505128	0.000931209338074817\\
0.00094483470026005	0.000931207092411984\\
0.00094483124185063	0.000931205333770562\\
0.000944828626165158	0.000931203830760087\\
0.00094482570824285	0.000931201962340366\\
0.000944823320857996	0.000931200264243388\\
0.000944820797083095	0.000931198282524072\\
0.000944818436457118	0.000931196231738259\\
0.000944817065599912	0.000931194940775414\\
0.000666172865144616	0.000641827260687818\\
0.00066608480570714	0.000641225162002834\\
0.000665933747270548	5.1429931679751e-09\\
0.000665937663943505	2.6464573050457e-18\\
0.000665937663957412	2.56331862610015e-20\\
0.00179284167343554	4.12442654237233e-14\\
0.00179996213281698	0.000131597351506554\\
0.00180741613556227	0.000385954778968086\\
0.00180816416675568	0.00092051863831523\\
0.00180813232064453	0.000921859200902388\\
0.00180805141427027	0.000922966363355986\\
0.00180793735194591	0.000923896258342538\\
0.00180780005090852	0.000924690490277386\\
0.00180764709754446	0.000925373483790865\\
0.00180748259104349	0.000925969298705207\\
0.00180730920132486	0.000926495018121581\\
0.00180712977427	0.000926960625605547\\
0.00180694604044497	0.000927375443789162\\
0.00180675764720281	0.000927749950181251\\
0.00180656608987988	0.00092808825455196\\
0.00180637080873473	0.000928396761518694\\
0.00180617226146695	0.000928678778348136\\
0.00180596971317736	0.000928938408687284\\
0.00180576306092541	0.000929178045128876\\
0.00180555185284818	0.000929399946939874\\
};
\addlegendentry{$L = 64$}

\addplot[only marks, mark=o, mark options={}, mark size=3.5355pt, draw=mycolor9, forget plot] table[row sep=crcr]{%
x	y\\
0.00180675764720281	0.000927749950181251\\
};
\addplot[only marks, mark=x, mark options={}, mark size=3.5355pt, draw=mycolor9, forget plot] table[row sep=crcr]{%
x	y\\
0.00180816416675568	0.00092051863831523\\
};
\addplot [color=mycolor9, line width=2.0pt, forget plot]
  table[row sep=crcr]{%
0.00180816416675568	0.00092051863831523\\
0.00180813232064453	0.000921859200902388\\
0.00180805141427027	0.000922966363355986\\
0.00180793735194591	0.000923896258342538\\
0.00180780005090852	0.000924690490277386\\
0.00180764709754446	0.000925373483790865\\
0.00180748259104349	0.000925969298705207\\
0.00180730920132486	0.000926495018121581\\
0.00180712977427	0.000926960625605547\\
0.00180694604044497	0.000927375443789162\\
0.00180675764720281	0.000927749950181251\\
0.00180656608987988	0.00092808825455196\\
0.00180637080873473	0.000928396761518694\\
0.00180617226146695	0.000928678778348136\\
0.00180596971317736	0.000928938408687284\\
0.00180576306092541	0.000929178045128876\\
0.00180555185284818	0.000929399946939874\\
0.00180555185284818	0.000929399946939874\\
0.00180533514986291	0.000929606320188223\\
0.00180511241911368	0.000929798491198706\\
0.00180488301833069	0.000929977560540128\\
0.00180464543585704	0.000930144931618639\\
0.00180439887046556	0.000930301114559859\\
0.00180414061727167	0.000930447428205666\\
0.00180387197745908	0.000930582648819448\\
0.00180358777824979	0.000930708610274962\\
0.00180328851982513	0.0009308239951381\\
0.00180297041976148	0.000930929010172791\\
0.00180263140440767	0.000931022740000879\\
0.00180226747773386	0.000931104368117137\\
0.00180187450841429	0.000931172452949046\\
0.00180144635722952	0.000931225124403758\\
0.00180097831933697	0.000931259464153654\\
0.00180046348509318	0.000931271909489233\\
};

\addlegendimage{only marks, mark=o, mark options={solid}, mark size=2pt, line width=2pt, draw=mycolor11}
\addlegendentry{P-max}

\addlegendimage{only marks, mark=x, mark options={solid}, mark size=3pt, line width=2pt, draw=mycolor11}
\addlegendentry{E-max}

\addplot [draw=none, color=mycolor11, line width=2.0pt]
  table[row sep=crcr]{%
0	0\\
};
\addlegendentry{R-max}

\end{axis}
\end{tikzpicture}%

%% file: assets/simulation/pc_singular_pareto_sx128.tex
%
%
\definecolor{mycolor1}{rgb}{0.00000,0.44706,0.74118}%
\definecolor{mycolor2}{rgb}{0.92900,0.69400,0.12500}%
\definecolor{mycolor3}{rgb}{0.49400,0.18400,0.55600}%
\definecolor{mycolor4}{rgb}{0.85098,0.32549,0.09804}%
\definecolor{mycolor5}{rgb}{0.63500,0.07800,0.18400}%
\definecolor{mycolor6}{rgb}{0.00000,0.44700,0.74100}%
\definecolor{mycolor7}{rgb}{0.92941,0.69412,0.12549}%
\definecolor{mycolor8}{rgb}{0.46600,0.67400,0.18800}%
\definecolor{mycolor9}{rgb}{0.49412,0.18431,0.55686}%
\definecolor{mycolor10}{rgb}{0.85000,0.32500,0.09800}%
\definecolor{mycolor11}{rgb}{0.50196,0.50196,0.50196}%
\begin{tikzpicture}

\begin{axis}[%
width=9.509cm,
height=7.5cm,
at={(0cm,0cm)},
scale only axis,
unbounded coords=jump,
xmin=6.82872585627327e-18,
xmax=0.00243816160662188,
xlabel={$\sigma_1(\mathbf{H})$},
ymin=3.65485567546015e-21,
ymax=0.00161848476949908,
ylabel={$\sigma_2(\mathbf{H})$},
axis background/.style={fill=white},
xmajorgrids,
ymajorgrids,
legend style={at={(0.03,0.97)}, anchor=north west, legend cell align=left, align=left, draw=white!15!black},
every axis plot/.append style={line width=1.5pt}
]
\addplot[only marks, mark=triangle, mark options={}, mark size=2.3570pt, draw=black] table[row sep=crcr]{%
x	y\\
0.00106823197271749	0.000315419858965939\\
};
\addlegendentry{Direct}

\addplot [color=mycolor1, dotted, line width=2.0pt]
  table[row sep=crcr]{%
0.00184667263255644	0.00105283104259383\\
0.00182013858730192	0.00107811414263772\\
0.00179446438717075	0.00110027146645747\\
0.00176844990466793	0.00112057240647327\\
0.00174471598225592	0.00113735225073209\\
0.00173147018986277	0.0011457366537772\\
0.00171580206722167	0.00115459255032051\\
0.00169703638968541	0.00116404598160524\\
0.00167988551784587	0.00117165264042197\\
0.00165718151718767	0.00118033305141772\\
0.00163702803533105	0.00118700156329319\\
0.00161226798692668	0.00119381007407023\\
0.00157105208627055	0.00120298087682078\\
0.00150917948363715	0.00121387316714933\\
0.00147157873771745	0.00121848599716562\\
0.00141343882647389	0.00122268663375429\\
0.00135084809104336	0.00122423429959488\\
0.00128936609099569	0.00122271750802284\\
0.00121973809357671	0.00121765733083521\\
0.00121741694180727	0.00121741694049506\\
0.00121741694033254	0.00121741694033254\\
0.00121741694033254	0.00121741694033254\\
0.000119403095894786	0.000119403095894786\\
0.000119403095894786	0.000119403095894786\\
0.000119403095894786	0.000119403095894785\\
0.000119403095868481	0.000119403095619605\\
0.000119403020394798	0.000119402298841366\\
0.000119386835029527	0.000119220988146345\\
0.000119057519576339	0.000114537334043434\\
0.000118455254948543	9.74254158896766e-05\\
0.000118475403377689	6.05430449574898e-05\\
0.000119625289046176	3.58679417568317e-05\\
0.000121423832495273	1.68770395478336e-05\\
0.000123412343773753	3.63936799619954e-06\\
0.000126000066494017	2.66928867246076e-10\\
0.000126094125790575	7.84463000637013e-16\\
0.000126164157534685	1.20407546614755e-19\\
0.000126171645826149	1.34875575902017e-20\\
0.000464871894092027	3.07637746131283e-20\\
0.00212946194498699	1.28942837838959e-19\\
0.00212946194498699	6.05180410022353e-19\\
0.00213171673459533	1.04876633850556e-05\\
0.00213642096799013	3.76095271630723e-05\\
0.00214246905743812	8.74587445329254e-05\\
0.00214751768540749	0.000157066053816262\\
0.00215085610784843	0.000271599925869964\\
0.00214933836456014	0.00033572935755316\\
0.00214495278417729	0.000395001725460926\\
0.00213813677099238	0.000450421310742063\\
0.00212769108455725	0.000510559965583877\\
0.0021139023007005	0.000571766564918572\\
0.00209785728241506	0.000629897607970685\\
0.00208118843029345	0.000680330658337416\\
0.00206298961265576	0.000727725985644975\\
0.00204487386350509	0.000768568135802446\\
0.0020225203667261	0.000812913003034874\\
0.00199616218047092	0.000859421583248756\\
0.0019661374861644	0.000906979746165971\\
0.00193991813992618	0.000944225376108491\\
0.00191015036907694	0.000982312334340278\\
0.00187569450797794	0.00102229903484787\\
0.00184667263255644	0.00105283104259383\\
};
\addlegendentry{$L = 1$}

\addplot[only marks, mark=o, mark options={}, mark size=3.5355pt, draw=mycolor1, forget plot] table[row sep=crcr]{%
x	y\\
0.00209785728241506	0.000629897607970685\\
};
\addplot[only marks, mark=x, mark options={}, mark size=3.5355pt, draw=mycolor1, forget plot] table[row sep=crcr]{%
x	y\\
0.00215085610784843	0.000271599925869964\\
};
\addplot [color=mycolor1, line width=2.0pt, forget plot]
  table[row sep=crcr]{%
0.00215085610784843	0.000271599925869964\\
0.00214933836456014	0.00033572935755316\\
0.00214495278417729	0.000395001725460926\\
0.00213813677099238	0.000450421310742063\\
0.00212769108455725	0.000510559965583877\\
0.0021139023007005	0.000571766564918572\\
0.00209785728241506	0.000629897607970685\\
0.00208118843029345	0.000680330658337416\\
0.00206298961265576	0.000727725985644975\\
0.00204487386350509	0.000768568135802446\\
0.0020225203667261	0.000812913003034874\\
0.00199616218047092	0.000859421583248756\\
0.0019661374861644	0.000906979746165971\\
0.00193991813992618	0.000944225376108491\\
0.00191015036907694	0.000982312334340278\\
0.00187569450797794	0.00102229903484787\\
0.00184667263255644	0.00105283104259383\\
0.00184667263255644	0.00105283104259383\\
0.00182013858730192	0.00107811414263772\\
0.00179446438717075	0.00110027146645747\\
0.00176844990466793	0.00112057240647327\\
0.00174471598225592	0.00113735225073209\\
0.00173147018986277	0.0011457366537772\\
0.00171580206722167	0.00115459255032051\\
0.00169703638968541	0.00116404598160524\\
0.00167988551784587	0.00117165264042197\\
0.00165718151718767	0.00118033305141772\\
0.00163702803533105	0.00118700156329319\\
0.00161226798692668	0.00119381007407023\\
0.00157105208627055	0.00120298087682078\\
0.00150917948363715	0.00121387316714933\\
0.00147157873771745	0.00121848599716562\\
0.00141343882647389	0.00122268663375429\\
0.00135084809104336	0.00122423429959488\\
};
\addplot [color=mycolor4, dotted, line width=2.0pt]
  table[row sep=crcr]{%
0.00224810954787343	0.00142765882888639\\
0.00223752516987624	0.00143773668940813\\
0.00222645577942414	0.00144728533430297\\
0.00221483758007418	0.00145635192562367\\
0.0022025971723238	0.00146497167929375\\
0.00218964306727519	0.00147317316537416\\
0.0021758697346165	0.00148097339279209\\
0.00216113623435137	0.00148838676976667\\
0.00214526641585671	0.00149541773160406\\
0.00212805031445442	0.00150205397827131\\
0.00210921569109828	0.00150827010338685\\
0.00208840018900012	0.00151402285650014\\
0.00206515213082401	0.00151923352023128\\
0.00204653673297032	0.00152250712144237\\
0.0020319724883019	0.00152429677367295\\
0.00201538660022458	0.00152551134991323\\
0.00199610708059676	0.00152597176892137\\
0.00152539239259048	0.0015253923826303\\
0.00152539238261993	0.00152539238223587\\
0.00152539238222263	0.00152539238221814\\
0.00152539238221792	0.00152539238221792\\
2.5656190712798e-05	2.56561907127979e-05\\
2.5656190712798e-05	2.56561907127979e-05\\
2.56542030161616e-05	2.55564221517753e-05\\
2.55651332386643e-05	1.80765962966132e-05\\
2.54274242107853e-05	3.89478936932254e-08\\
2.54284933602812e-05	9.88017767029914e-20\\
2.54284933602812e-05	3.65485567546015e-21\\
0.00235185136612809	1.88970490505015e-19\\
0.00235185136612809	1.89904129212748e-19\\
0.00235185136623601	2.10937246606743e-12\\
0.00235306574672104	0.00112241590476244\\
0.00235222900963638	0.00115749997958599\\
0.00234998758195604	0.00118808242566059\\
0.00234662460486239	0.00121543787151824\\
0.00234233285079524	0.00124022235534106\\
0.00233725320120539	0.00126287502712395\\
0.00233149354254889	0.0012837093862117\\
0.00232513727045498	0.00130296539740488\\
0.00231825051346261	0.00132082977484434\\
0.00231088706115753	0.00133744986558083\\
0.00230308575053904	0.00135295487223089\\
0.00229487765359623	0.00136744900965217\\
0.00228628461857372	0.00138102243969255\\
0.00227731644465318	0.00139375886428586\\
0.00226797184875481	0.0014057346459718\\
0.00225824024199718	0.00141701733112661\\
0.00224810954787343	0.00142765882888639\\
};
\addlegendentry{$L = 4$}

\addplot[only marks, mark=o, mark options={}, mark size=3.5355pt, draw=mycolor4, forget plot] table[row sep=crcr]{%
x	y\\
0.00229487765359623	0.00136744900965217\\
};
\addplot[only marks, mark=x, mark options={}, mark size=3.5355pt, draw=mycolor4, forget plot] table[row sep=crcr]{%
x	y\\
0.00235306574672104	0.00112241590476244\\
};
\addplot [color=mycolor4, line width=2.0pt, forget plot]
  table[row sep=crcr]{%
0.00235306574672104	0.00112241590476244\\
0.00235222900963638	0.00115749997958599\\
0.00234998758195604	0.00118808242566059\\
0.00234662460486239	0.00121543787151824\\
0.00234233285079524	0.00124022235534106\\
0.00233725320120539	0.00126287502712395\\
0.00233149354254889	0.0012837093862117\\
0.00232513727045498	0.00130296539740488\\
0.00231825051346261	0.00132082977484434\\
0.00231088706115753	0.00133744986558083\\
0.00230308575053904	0.00135295487223089\\
0.00229487765359623	0.00136744900965217\\
0.00228628461857372	0.00138102243969255\\
0.00227731644465318	0.00139375886428586\\
0.00226797184875481	0.0014057346459718\\
0.00225824024199718	0.00141701733112661\\
0.00224810954787343	0.00142765882888639\\
0.00224810954787343	0.00142765882888639\\
0.00223752516987624	0.00143773668940813\\
0.00222645577942414	0.00144728533430297\\
0.00221483758007418	0.00145635192562367\\
0.0022025971723238	0.00146497167929375\\
0.00218964306727519	0.00147317316537416\\
0.0021758697346165	0.00148097339279209\\
0.00216113623435137	0.00148838676976667\\
0.00214526641585671	0.00149541773160406\\
0.00212805031445442	0.00150205397827131\\
0.00210921569109828	0.00150827010338685\\
0.00208840018900012	0.00151402285650014\\
0.00206515213082401	0.00151923352023128\\
0.00204653673297032	0.00152250712144237\\
0.0020319724883019	0.00152429677367295\\
0.00201538660022458	0.00152551134991323\\
0.00199610708059676	0.00152597176892137\\
};
\addplot [color=mycolor7, dotted, line width=2.0pt]
  table[row sep=crcr]{%
0.00238467756794191	0.00157463743802116\\
0.0023819055258318	0.00157727655066139\\
0.00237896534516425	0.00157981282300349\\
0.00237584496745458	0.00158224793554099\\
0.00237252341665759	0.00158458699447776\\
0.00236898096929058	0.00158682982530082\\
0.00236518489959242	0.00158897971725607\\
0.00236110787568849	0.00159103122832968\\
0.00235670400375528	0.0015929825465803\\
0.00235192681387114	0.00159482430984728\\
0.00234671740987606	0.00159654394608101\\
0.00234101372325919	0.00159812080273434\\
0.00233471410748323	0.00159953351436297\\
0.0023277213879351	0.00160074429522783\\
0.00231990369916102	0.00160170525428858\\
0.00231109266365087	0.00160235119567967\\
0.00230106518932941	0.00160259231647374\\
0.00160118657134884	0.00160118657130681\\
0.00160118657130541	0.00160118657130541\\
0.00160118657130541	0.00160118657130541\\
6.82872585627327e-18	6.74276093779542e-18\\
6.88479121295111e-18	6.56284337752166e-18\\
7.42057943416505e-18	5.97949898675849e-18\\
1.02652588277888e-17	2.97624528851796e-18\\
1.35225160295751e-17	9.81653158908711e-19\\
4.14744432044337e-17	2.45169228672175e-20\\
7.31026945953357e-17	1.12901305223653e-20\\
0.00240962077068218	8.49382046082521e-20\\
0.00240962077068218	7.13445152076569e-19\\
0.00240962077068218	1.82769532418166e-15\\
0.00240962077068223	3.61308205942841e-14\\
0.00241036864849235	0.000689895899012993\\
0.00241058391784397	0.001500391286991\\
0.00241038337686835	0.00150873726156622\\
0.0024098417777718	0.00151612468860024\\
0.00240902820374127	0.00152274239623327\\
0.00240799204744864	0.00152872601777691\\
0.00240676935813045	0.00153417843028343\\
0.00240538530850182	0.0015391847600932\\
0.00240385837864775	0.00154381019963698\\
0.00240220130519937	0.00154810819989861\\
0.00240042398682364	0.00155211931244551\\
0.00239852970644693	0.00155588372037528\\
0.00239652119195967	0.0015594299630115\\
0.00239439821373247	0.00156278291878547\\
0.0023921563985175	0.00156596608984207\\
0.00238979468829576	0.0015689920967447\\
0.0023873042623312	0.00157187887211322\\
0.00238467756794191	0.00157463743802116\\
};
\addlegendentry{$L = 16$}

\addplot[only marks, mark=o, mark options={}, mark size=3.5355pt, draw=mycolor7, forget plot] table[row sep=crcr]{%
x	y\\
0.00239439821373247	0.00156278291878547\\
};
\addplot[only marks, mark=x, mark options={}, mark size=3.5355pt, draw=mycolor7, forget plot] table[row sep=crcr]{%
x	y\\
0.00241058391784397	0.001500391286991\\
};
\addplot [color=mycolor7, line width=2.0pt, forget plot]
  table[row sep=crcr]{%
0.00241058391784397	0.001500391286991\\
0.00241038337686835	0.00150873726156622\\
0.0024098417777718	0.00151612468860024\\
0.00240902820374127	0.00152274239623327\\
0.00240799204744864	0.00152872601777691\\
0.00240676935813045	0.00153417843028343\\
0.00240538530850182	0.0015391847600932\\
0.00240385837864775	0.00154381019963698\\
0.00240220130519937	0.00154810819989861\\
0.00240042398682364	0.00155211931244551\\
0.00239852970644693	0.00155588372037528\\
0.00239652119195967	0.0015594299630115\\
0.00239439821373247	0.00156278291878547\\
0.0023921563985175	0.00156596608984207\\
0.00238979468829576	0.0015689920967447\\
0.0023873042623312	0.00157187887211322\\
0.00238467756794191	0.00157463743802116\\
0.00238467756794191	0.00157463743802116\\
0.0023819055258318	0.00157727655066139\\
0.00237896534516425	0.00157981282300349\\
0.00237584496745458	0.00158224793554099\\
0.00237252341665759	0.00158458699447776\\
0.00236898096929058	0.00158682982530082\\
0.00236518489959242	0.00158897971725607\\
0.00236110787568849	0.00159103122832968\\
0.00235670400375528	0.0015929825465803\\
0.00235192681387114	0.00159482430984728\\
0.00234671740987606	0.00159654394608101\\
0.00234101372325919	0.00159812080273434\\
0.00233471410748323	0.00159953351436297\\
0.0023277213879351	0.00160074429522783\\
0.00231990369916102	0.00160170525428858\\
0.00231109266365087	0.00160235119567967\\
0.00230106518932941	0.00160259231647374\\
};
\addplot [color=mycolor9, dotted, line width=2.0pt]
  table[row sep=crcr]{%
0.00243247601307818	0.00161102948034108\\
0.00243181726442887	0.00161165656477936\\
0.00243111150394324	0.0016122652276857\\
0.00243034895492804	0.00161286020159992\\
0.00242952456255834	0.00161344066443583\\
0.00242862742029526	0.001614008619057\\
0.00242764675719739	0.00161456392365131\\
0.00242657046026066	0.00161510534481505\\
0.0024253816348233	0.00161563184601292\\
0.00242406462370524	0.00161613934041111\\
0.00242258868646807	0.00161662631813364\\
0.00242092561566571	0.00161708574928453\\
0.00241903974380994	0.0016175082481848\\
0.00241688060323457	0.00161788176032039\\
0.00241438047713464	0.00161818872597477\\
0.00241145950912255	0.00161840237398835\\
0.00240800886067149	0.00161848476949908\\
0.0016125832868018	0.00161258328680179\\
0.0016125832868018	0.00161258328680179\\
4.41861269324237e-17	4.39679816668224e-17\\
4.42487311992999e-17	4.3689150444209e-17\\
4.49235979942934e-17	4.18057790230463e-17\\
4.61823527517726e-17	3.88188355014716e-17\\
4.75959736712383e-17	3.71564588856519e-17\\
4.99517629543179e-17	3.46309127460204e-17\\
5.3414622038972e-17	3.20782508438301e-17\\
1.78221588060886e-16	1.23512160090444e-19\\
2.25535334231316e-16	9.69560902112338e-20\\
0.00242826624274052	1.84359040767942e-19\\
0.00243816160662188	0.00159530193644534\\
0.00243812120540108	0.00159696998886273\\
0.00243801276205771	0.00159844584296826\\
0.00243784681893939	0.00159979588937176\\
0.00243763228419044	0.00160103384331362\\
0.00243737839939767	0.00160216490951862\\
0.00243708815663387	0.00160321401749747\\
0.00243676405035725	0.00160419536850885\\
0.00243641249255725	0.00160510705401521\\
0.00243602874408712	0.00160597328678824\\
0.00243561458791574	0.00160679618287729\\
0.00243517406320179	0.00160757377854252\\
0.00243470267688429	0.00160831823103241\\
0.00243419894595674	0.00160903347764242\\
0.00243366221677335	0.00160972119062613\\
0.00243308965934047	0.00161038494285019\\
0.00243247601307818	0.00161102948034108\\
};
\addlegendentry{$L = 128$}

\addplot[only marks, mark=o, mark options={}, mark size=3.5355pt, draw=mycolor9, forget plot] table[row sep=crcr]{%
x	y\\
0.00243470267688429	0.00160831823103241\\
};
\addplot[only marks, mark=x, mark options={}, mark size=3.5355pt, draw=mycolor9, forget plot] table[row sep=crcr]{%
x	y\\
0.00243816160662188	0.00159530193644534\\
};
\addplot [color=mycolor9, line width=2.0pt, forget plot]
  table[row sep=crcr]{%
0.00243816160662188	0.00159530193644534\\
0.00243812120540108	0.00159696998886273\\
0.00243801276205771	0.00159844584296826\\
0.00243784681893939	0.00159979588937176\\
0.00243763228419044	0.00160103384331362\\
0.00243737839939767	0.00160216490951862\\
0.00243708815663387	0.00160321401749747\\
0.00243676405035725	0.00160419536850885\\
0.00243641249255725	0.00160510705401521\\
0.00243602874408712	0.00160597328678824\\
0.00243561458791574	0.00160679618287729\\
0.00243517406320179	0.00160757377854252\\
0.00243470267688429	0.00160831823103241\\
0.00243419894595674	0.00160903347764242\\
0.00243366221677335	0.00160972119062613\\
0.00243308965934047	0.00161038494285019\\
0.00243247601307818	0.00161102948034108\\
0.00243247601307818	0.00161102948034108\\
0.00243181726442887	0.00161165656477936\\
0.00243111150394324	0.0016122652276857\\
0.00243034895492804	0.00161286020159992\\
0.00242952456255834	0.00161344066443583\\
0.00242862742029526	0.001614008619057\\
0.00242764675719739	0.00161456392365131\\
0.00242657046026066	0.00161510534481505\\
0.0024253816348233	0.00161563184601292\\
0.00242406462370524	0.00161613934041111\\
0.00242258868646807	0.00161662631813364\\
0.00242092561566571	0.00161708574928453\\
0.00241903974380994	0.0016175082481848\\
0.00241688060323457	0.00161788176032039\\
0.00241438047713464	0.00161818872597477\\
0.00241145950912255	0.00161840237398835\\
0.00240800886067149	0.00161848476949908\\
};

\addlegendimage{only marks, mark=o, mark options={solid}, mark size=2pt, line width=2pt, draw=mycolor11}
\addlegendentry{P-max}

\addlegendimage{only marks, mark=x, mark options={solid}, mark size=3pt, line width=2pt, draw=mycolor11}
\addlegendentry{E-max}

\addplot [draw=none, color=mycolor11, line width=2.0pt]
  table[row sep=crcr]{%
0	0\\
};
\addlegendentry{R-max}

\end{axis}
\end{tikzpicture}%

%% file: assets/simulation/pc_singular_bound_rank1_sx32.tex
%
%
\definecolor{mycolor1}{rgb}{0.00000,0.44700,0.74100}%
\definecolor{mycolor2}{rgb}{0.00000,0.44706,0.74118}%
\definecolor{mycolor3}{rgb}{0.85098,0.32549,0.09804}%
\definecolor{mycolor4}{rgb}{0.49400,0.18400,0.55600}%
\definecolor{mycolor5}{rgb}{0.46667,0.67451,0.18824}%
\makeatletter
\newcommand\resetstackedplots{
    \makeatletter
    \pgfplots@stacked@isfirstplottrue
    \makeatother
    \addplot [forget plot,draw=none] coordinates{(1,0) (5,0) (10,0) (15,0) (20,0)};
}
\makeatother
\begin{tikzpicture}

\begin{axis}[%
width=9.509cm,
height=7.5cm,
at={(0cm,0cm)},
scale only axis,
bar width=0.325,
xmin=-0.125,
xmax=5.125,
xtick={1,2,3,4},
xticklabels={{$\sigma_1(\mathbf{H})$},{$\sigma_2(\mathbf{H})$},{$\sigma_3(\mathbf{H})$},{$\sigma_4(\mathbf{H})$}},
ymin=0,
ymax=0.002,
ylabel={Amplitude},
axis background/.style={fill=white},
xmajorgrids,
ymajorgrids,
legend style={legend cell align=left, align=left, draw=white!15!black},
legend columns=4,
transpose legend,
legend style={/tikz/column 2/.style={column sep=5pt}}
]
\addplot[ybar stacked, fill=none, draw=none, forget plot] table[row sep=crcr] {%
0.8375	0.00172515536422714\\
1.8375	0.00100217806517558\\
2.8375	0.000352757538105714\\
3.8375	0.000123939298123254\\
};
\addplot[forget plot, color=white!15!black] table[row sep=crcr] {%
-0.125	0\\
5.125	0\\
};
\addplot[ybar stacked, fill=mycolor2, fill opacity=0.5, draw=black, area legend] table[row sep=crcr] {%
0.8375	1.77586704705756e-05\\
1.8375	0.000127158817727647\\
2.8375	0.0002203611436369\\
3.8375	0.000174937536612442\\
};
\addplot[forget plot, color=white!15!black] table[row sep=crcr] {%
-0.125	0\\
5.125	0\\
};
\addlegendentry{D-min}

\addplot[ybar stacked, fill=mycolor3, fill opacity=0.5, draw=black, area legend] table[row sep=crcr] {%
0.8375	0.000181872601837377\\
1.8375	0.000366753236169139\\
2.8375	0.000296256103311964\\
3.8375	3.0817049708412e-05\\
};
\addplot[forget plot, color=white!15!black] table[row sep=crcr] {%
-0.125	0\\
5.125	0\\
};
\addlegendentry{D-max}

\resetstackedplots

\addplot[ybar stacked, fill=none, draw=none, forget plot] table[row sep=crcr] {%
1.1625	0.00172516728828725\\
2.1625	0.00100306026989406\\
3.1625	0.000335153436324708\\
4.1625	0.000105743462680291\\
};
\addplot[forget plot, color=white!15!black] table[row sep=crcr] {%
-0.125	0\\
5.125	0\\
};
\addplot[ybar stacked, fill=mycolor2, draw=black, area legend] table[row sep=crcr] {%
1.1625	1.77467464104629e-05\\
2.1625	0.000126276613009161\\
3.1625	0.000237965245417907\\
4.1625	0.000193133372055405\\
};
\addplot[forget plot, color=white!15!black] table[row sep=crcr] {%
-0.125	0\\
5.125	0\\
};
\addlegendentry{BD-min}

\addplot[ybar stacked, fill=mycolor3, draw=black, area legend] table[row sep=crcr] {%
1.1625	0.000222406605491606\\
2.1625	0.000415962829151039\\
3.1625	0.000325674699988366\\
4.1625	3.0817049696797e-05\\
};
\addplot[forget plot, color=white!15!black] table[row sep=crcr] {%
-0.125	0\\
5.125	0\\
};
\addlegendentry{BD-max}

\addplot [color=mycolor5, line width=2.0pt]
  table[row sep=crcr]{%
-0.125	0.00172515536418547\\
5.125	0.00172515536418547\\
};
\addlegendentry{$\sigma_1(\mathbf{T})$}

\addplot [color=mycolor5, dashed, line width=2.0pt]
  table[row sep=crcr]{%
-0.125	0.00100217806517392\\
5.125	0.00100217806517392\\
};
\addlegendentry{$\sigma_2(\mathbf{T})$}

\addplot [color=mycolor5, dotted, line width=2.0pt]
  table[row sep=crcr]{%
-0.125	0.000329693884444316\\
5.125	0.000329693884444316\\
};
\addlegendentry{$\sigma_3(\mathbf{T})$}

\addplot [color=mycolor5, dashdotted, line width=2.0pt]
  table[row sep=crcr]{%
-0.125	1.01657134672358e-11\\
5.125	1.01657134672358e-11\\
};
\addlegendentry{$\sigma_4(\mathbf{T})$}

\end{axis}

\end{tikzpicture}%

%% file: assets/simulation/pc_singular_bound_rank1_sx64.tex
%
%
\definecolor{mycolor1}{rgb}{0.00000,0.44700,0.74100}%
\definecolor{mycolor2}{rgb}{0.00000,0.44706,0.74118}%
\definecolor{mycolor3}{rgb}{0.85098,0.32549,0.09804}%
\definecolor{mycolor4}{rgb}{0.49400,0.18400,0.55600}%
\definecolor{mycolor5}{rgb}{0.46667,0.67451,0.18824}%
\makeatletter
\newcommand\resetstackedplots{
    \makeatletter
    \pgfplots@stacked@isfirstplottrue
    \makeatother
    \addplot [forget plot,draw=none] coordinates{(1,0) (5,0) (10,0) (15,0) (20,0)};
}
\makeatother
\begin{tikzpicture}

\begin{axis}[%
width=9.509cm,
height=7.5cm,
at={(0cm,0cm)},
scale only axis,
bar width=0.325,
xmin=-0.125,
xmax=5.125,
xtick={1,2,3,4},
xticklabels={{$\sigma_1(\mathbf{H})$},{$\sigma_2(\mathbf{H})$},{$\sigma_3(\mathbf{H})$},{$\sigma_4(\mathbf{H})$}},
ymin=0,
ymax=0.003,
ylabel={Amplitude},
axis background/.style={fill=white},
xmajorgrids,
ymajorgrids,
legend style={legend cell align=left, align=left, draw=white!15!black},
legend columns=4,
transpose legend,
legend style={/tikz/column 2/.style={column sep=5pt}}
]
\addplot[ybar stacked, fill=none, draw=none, forget plot] table[row sep=crcr] {%
0.8375	0.00195358578718362\\
1.8375	0.000905403110314898\\
2.8375	0.000478480067807894\\
3.8375	3.24792561983075e-17\\
};
\addplot[forget plot, color=white!15!black] table[row sep=crcr] {%
-0.125	0\\
5.125	0\\
};
\addplot[ybar stacked, fill=mycolor2, fill opacity=0.5, draw=black, area legend] table[row sep=crcr] {%
0.8375	4.10673634284198e-05\\
1.8375	9.08396603062443e-05\\
2.8375	0.000280905224832324\\
3.8375	0.000316345482660046\\
};
\addplot[forget plot, color=white!15!black] table[row sep=crcr] {%
-0.125	0\\
5.125	0\\
};
\addlegendentry{D-min}

\addplot[ybar stacked, fill=mycolor3, fill opacity=0.5, draw=black, area legend] table[row sep=crcr] {%
0.8375	0.00057247498440489\\
1.8375	0.000934081190202522\\
2.8375	0.00014601781767251\\
3.8375	0.000162134585145759\\
};
\addplot[forget plot, color=white!15!black] table[row sep=crcr] {%
-0.125	0\\
5.125	0\\
};
\addlegendentry{D-max}

\resetstackedplots

\addplot[ybar stacked, fill=none, draw=none, forget plot] table[row sep=crcr] {%
1.1625	0.00195358621182515\\
2.1625	0.000905405810406567\\
3.1625	0.000478484657925968\\
4.1625	3.62951302572638e-06\\
};
\addplot[forget plot, color=white!15!black] table[row sep=crcr] {%
-0.125	0\\
5.125	0\\
};
\addplot[ybar stacked, fill=mycolor2, draw=black, area legend] table[row sep=crcr] {%
1.1625	4.10669387868874e-05\\
2.1625	9.0836960214575e-05\\
3.1625	0.000280900634714249\\
4.1625	0.000312715969634352\\
};
\addplot[forget plot, color=white!15!black] table[row sep=crcr] {%
-0.125	0\\
5.125	0\\
};
\addlegendentry{BD-min}

\addplot[ybar stacked, fill=mycolor3, draw=black, area legend] table[row sep=crcr] {%
1.1625	0.000668204539596567\\
2.1625	0.000957343016500992\\
3.1625	0.000146017817609404\\
4.1625	0.000162134585134279\\
};
\addplot[forget plot, color=white!15!black] table[row sep=crcr] {%
-0.125	0\\
5.125	0\\
};
\addlegendentry{BD-max}

\addplot [color=mycolor5, line width=2.0pt]
  table[row sep=crcr]{%
-0.125	0.00195358578716264\\
5.125	0.00195358578716264\\
};
\addlegendentry{$\sigma_1(\mathbf{T})$}

\addplot [color=mycolor5, dashed, line width=2.0pt]
  table[row sep=crcr]{%
-0.125	0.000905403110313456\\
5.125	0.000905403110313456\\
};
\addlegendentry{$\sigma_2(\mathbf{T})$}

\addplot [color=mycolor5, dotted, line width=2.0pt]
  table[row sep=crcr]{%
-0.125	0.000478480067806991\\
5.125	0.000478480067806991\\
};
\addlegendentry{$\sigma_3(\mathbf{T})$}

\addplot [color=mycolor5, dashdotted, line width=2.0pt]
  table[row sep=crcr]{%
-0.125	1.20606340074858e-11\\
5.125	1.20606340074858e-11\\
};
\addlegendentry{$\sigma_4(\mathbf{T})$}

\end{axis}
\end{tikzpicture}%

%% file: assets/simulation/pc_singular_bound_rank2_sx128.tex
%
%
\definecolor{mycolor1}{rgb}{0.00000,0.44700,0.74100}%
\definecolor{mycolor2}{rgb}{0.00000,0.44706,0.74118}%
\definecolor{mycolor3}{rgb}{0.85098,0.32549,0.09804}%
\definecolor{mycolor4}{rgb}{0.49400,0.18400,0.55600}%
\definecolor{mycolor5}{rgb}{0.46667,0.67451,0.18824}%
\makeatletter
\newcommand\resetstackedplots{
    \makeatletter
    \pgfplots@stacked@isfirstplottrue
    \makeatother
    \addplot [forget plot,draw=none] coordinates{(1,0) (5,0) (10,0) (15,0) (20,0)};
}
\makeatother
\begin{tikzpicture}

\begin{axis}[%
width=9.509cm,
height=7.5cm,
at={(0cm,0cm)},
scale only axis,
bar width=0.325,
xmin=-0.125,
xmax=5.125,
xtick={1,2,3,4},
xticklabels={{$\sigma_1(\mathbf{H})$},{$\sigma_2(\mathbf{H})$},{$\sigma_3(\mathbf{H})$},{$\sigma_4(\mathbf{H})$}},
ymin=0,
ymax=0.004,
ylabel={Amplitude},
axis background/.style={fill=white},
xmajorgrids,
ymajorgrids,
legend style={legend cell align=left, align=left, draw=white!15!black},
legend columns=4,
transpose legend,
legend style={/tikz/column 2/.style={column sep=5pt}}
]
\addplot[ybar stacked, fill=none, draw=none, forget plot] table[row sep=crcr] {%
0.8375	0.00142543267002251\\
1.8375	0.000656306214245826\\
2.8375	0.000132970460937413\\
3.8375	3.2994096547908e-20\\
};
\addplot[forget plot, color=white!15!black] table[row sep=crcr] {%
-0.125	0\\
5.125	0\\
};
\addplot[ybar stacked, fill=mycolor2, fill opacity=0.5, draw=black, area legend] table[row sep=crcr] {%
0.8375	0.000316188517375316\\
1.8375	0.000636356951614819\\
2.8375	0.000675093172956787\\
3.8375	0.000311028141103694\\
};
\addplot[forget plot, color=white!15!black] table[row sep=crcr] {%
-0.125	0\\
5.125	0\\
};
\addlegendentry{D-min}

\addplot[ybar stacked, fill=mycolor3, fill opacity=0.5, draw=black, area legend] table[row sep=crcr] {%
0.8375	0.00143103448582135\\
1.8375	0.000911508088276409\\
2.8375	0.000557699498413843\\
3.8375	0.000345278073129164\\
};
\addplot[forget plot, color=white!15!black] table[row sep=crcr] {%
-0.125	0\\
5.125	0\\
};
\addlegendentry{D-max}

\resetstackedplots

\addplot[ybar stacked, fill=none, draw=none, forget plot] table[row sep=crcr] {%
1.1625	0.00142543579700608\\
2.1625	0.00065631304524136\\
3.1625	5.27477395632166e-06\\
4.1625	4.21048388214027e-08\\
};
\addplot[forget plot, color=white!15!black] table[row sep=crcr] {%
-0.125	0\\
5.125	0\\
};
\addplot[ybar stacked, fill=mycolor2, draw=black, area legend] table[row sep=crcr] {%
1.1625	0.000316185390391755\\
2.1625	0.000636350120619285\\
3.1625	0.000802788859937878\\
4.1625	0.000310986036264873\\
};
\addplot[forget plot, color=white!15!black] table[row sep=crcr] {%
-0.125	0\\
5.125	0\\
};
\addlegendentry{BD-min}

\addplot[ybar stacked, fill=mycolor3, draw=black, area legend] table[row sep=crcr] {%
1.1625	0.00179647322615157\\
2.1625	0.00141980780340078\\
3.1625	0.000617369035963568\\
4.1625	0.000345278072973618\\
};
\addplot[forget plot, color=white!15!black] table[row sep=crcr] {%
-0.125	0\\
5.125	0\\
};
\addlegendentry{BD-max}

\addplot [color=mycolor5, line width=2.0pt]
  table[row sep=crcr]{%
-0.125	0.00142543266997377\\
5.125	0.00142543266997377\\
};
\addlegendentry{$\sigma_1(\mathbf{T})$}

\addplot [color=mycolor5, dashed, line width=2.0pt]
  table[row sep=crcr]{%
-0.125	0.000656306214237067\\
5.125	0.000656306214237067\\
};
\addlegendentry{$\sigma_2(\mathbf{T})$}

\addplot [color=mycolor5, dotted, line width=2.0pt]
  table[row sep=crcr]{%
-0.125	1.83134593569563e-11\\
5.125	1.83134593569563e-11\\
};
\addlegendentry{$\sigma_3(\mathbf{T})$}

\addplot [color=mycolor5, dashdotted, line width=2.0pt]
  table[row sep=crcr]{%
-0.125	8.07634856586464e-12\\
5.125	8.07634856586464e-12\\
};
\addlegendentry{$\sigma_4(\mathbf{T})$}

\end{axis}
\end{tikzpicture}%

%% file: assets/simulation/pc_singular_bound_rank4_sx256.tex
%
%
\definecolor{mycolor1}{rgb}{0.00000,0.44700,0.74100}%
\definecolor{mycolor2}{rgb}{0.00000,0.44706,0.74118}%
\definecolor{mycolor3}{rgb}{0.85098,0.32549,0.09804}%
\definecolor{mycolor4}{rgb}{0.49400,0.18400,0.55600}%
\definecolor{mycolor5}{rgb}{0.46667,0.67451,0.18824}%
\makeatletter
\newcommand\resetstackedplots{
    \makeatletter
    \pgfplots@stacked@isfirstplottrue
    \makeatother
    \addplot [forget plot,draw=none] coordinates{(1,0) (5,0) (10,0) (15,0) (20,0)};
}
\makeatother
\begin{tikzpicture}

\begin{axis}[%
width=9.509cm,
height=7.5cm,
at={(0cm,0cm)},
scale only axis,
bar width=0.325,
xmin=-0.125,
xmax=5.125,
xtick={1,2,3,4},
xticklabels={{$\sigma_1(\mathbf{H})$},{$\sigma_2(\mathbf{H})$},{$\sigma_3(\mathbf{H})$},{$\sigma_4(\mathbf{H})$}},
ymin=0,
ymax=0.005,
ylabel={Amplitude},
axis background/.style={fill=white},
xmajorgrids,
ymajorgrids,
legend style={legend cell align=left, align=left, draw=white!15!black},
legend columns=4,
transpose legend,
legend style={/tikz/column 2/.style={column sep=5pt}}
]
\addplot[ybar stacked, fill=none, draw=none, forget plot] table[row sep=crcr] {%
0.8375	0.000590158480426682\\
1.8375	0.000129241592659798\\
2.8375	2.58214028253461e-05\\
3.8375	6.14844865636593e-20\\
};
\addplot[forget plot, color=white!15!black] table[row sep=crcr] {%
-0.125	0\\
5.125	0\\
};
\addplot[ybar stacked, fill=mycolor2, fill opacity=0.5, draw=black, area legend] table[row sep=crcr] {%
0.8375	0.00131538553898149\\
1.8375	0.000939081061667071\\
2.8375	0.00068542783385897\\
3.8375	0.000169352904449523\\
};
\addplot[forget plot, color=white!15!black] table[row sep=crcr] {%
-0.125	0\\
5.125	0\\
};
\addlegendentry{D-min}

\addplot[ybar stacked, fill=mycolor3, fill opacity=0.5, draw=black, area legend] table[row sep=crcr] {%
0.8375	0.0023951194667282\\
1.8375	0.00184094568953638\\
2.8375	0.00115975659942478\\
3.8375	0.000827148080235149\\
};
\addplot[forget plot, color=white!15!black] table[row sep=crcr] {%
-0.125	0\\
5.125	0\\
};
\addlegendentry{D-max}

\resetstackedplots

\addplot[ybar stacked, fill=none, draw=none, forget plot] table[row sep=crcr] {%
1.1625	0.000346385584583973\\
2.1625	7.78556027220081e-08\\
3.1625	3.09245924695688e-10\\
4.1625	3.0971797863086e-11\\
};
\addplot[forget plot, color=white!15!black] table[row sep=crcr] {%
-0.125	0\\
5.125	0\\
};
\addplot[ybar stacked, fill=mycolor2, draw=black, area legend] table[row sep=crcr] {%
1.1625	0.00155915843482419\\
2.1625	0.00106824479872415\\
3.1625	0.000711248927438391\\
4.1625	0.000169352873477725\\
};
\addplot[forget plot, color=white!15!black] table[row sep=crcr] {%
-0.125	0\\
5.125	0\\
};
\addlegendentry{BD-min}

\addplot[ybar stacked, fill=mycolor3, draw=black, area legend] table[row sep=crcr] {%
1.1625	0.002953753985213\\
2.1625	0.002904805544911\\
3.1625	0.00250318921193009\\
4.1625	0.00168645690620485\\
};
\addplot[forget plot, color=white!15!black] table[row sep=crcr] {%
-0.125	0\\
5.125	0\\
};
\addlegendentry{BD-max}

\addplot [color=mycolor5, line width=2.0pt]
  table[row sep=crcr]{%
-0.125	3.45656762152605e-11\\
5.125	3.45656762152605e-11\\
};
\addlegendentry{$\sigma_1(\mathbf{T})$}

\addplot [color=mycolor5, dashed, line width=2.0pt]
  table[row sep=crcr]{%
-0.125	1.40250773872945e-11\\
5.125	1.40250773872945e-11\\
};
\addlegendentry{$\sigma_2(\mathbf{T})$}

\addplot [color=mycolor5, dotted, line width=2.0pt]
  table[row sep=crcr]{%
-0.125	9.62701814871279e-12\\
5.125	9.62701814871279e-12\\
};
\addlegendentry{$\sigma_3(\mathbf{T})$}

\addplot [color=mycolor5, dashdotted, line width=2.0pt]
  table[row sep=crcr]{%
-0.125	3.69974075950399e-12\\
5.125	3.69974075950399e-12\\
};
\addlegendentry{$\sigma_4(\mathbf{T})$}

\end{axis}
\end{tikzpicture}%

%% file: assets/simulation/pc_power_bond_txrx1_nd.tex
%
%
\definecolor{mycolor1}{rgb}{0.46667,0.67451,0.18824}%
\definecolor{mycolor2}{rgb}{0.92941,0.50588,0.59608}%
\begin{tikzpicture}

\begin{axis}[%
width=9.509cm,
height=7.5cm,
at={(0cm,0cm)},
scale only axis,
bar shift auto,
xmin=0,
xmax=6,
xtick={1,2,3,4,5},
xticklabels={{1},{4},{16},{64},{256}},
xlabel={RIS Group Size},
ymin=0,
ymax=7.17504949306389e-06,
ylabel={Channel Power [W]},
axis background/.style={fill=white},
xmajorgrids,
ymajorgrids,
legend style={at={(0.97,0.03)}, anchor=south east, legend cell align=left, align=left, draw=white!15!black}
]
\addplot[ybar, bar width=0.8, fill=mycolor1, draw=black, area legend] table[row sep=crcr] {%
1	4.05032179628278e-06\\
2	5.75669800225934e-06\\
3	6.32486926088864e-06\\
4	6.48461843519388e-06\\
5	6.52277226642172e-06\\
};
\addplot[forget plot, color=white!15!black] table[row sep=crcr] {%
0	0\\
6	0\\
};
\addlegendentry{Equivalent}

\addplot [color=mycolor2, dashed, line width=2.0pt]
  table[row sep=crcr]{%
0	6.52277226642171e-06\\
6	6.52277226642171e-06\\
};
\addlegendentry{Cascaded}

\end{axis}
\end{tikzpicture}%

%% file: assets/simulation/pc_power_bond_txrx4_nd.tex
%
%
\definecolor{mycolor1}{rgb}{0.46667,0.67451,0.18824}%
\definecolor{mycolor2}{rgb}{0.92941,0.50588,0.59608}%
\begin{tikzpicture}

\begin{axis}[%
width=9.509cm,
height=7.5cm,
at={(0cm,0cm)},
scale only axis,
bar shift auto,
xmin=0,
xmax=6,
xtick={1,2,3,4,5},
xticklabels={{1},{4},{16},{64},{256}},
xlabel={RIS Group Size},
ymin=0,
ymax=2.9105133741022e-05,
ylabel={Channel Power [W]},
axis background/.style={fill=white},
xmajorgrids,
ymajorgrids,
legend style={at={(0.97,0.03)}, anchor=south east, legend cell align=left, align=left, draw=white!15!black}
]
\addplot[ybar, bar width=0.8, fill=mycolor1, draw=black, area legend] table[row sep=crcr] {%
1	7.17506228344783e-06\\
2	1.70104957555835e-05\\
3	2.42867157787661e-05\\
4	2.6078917547989e-05\\
5	2.64592124918381e-05\\
};
\addplot[forget plot, color=white!15!black] table[row sep=crcr] {%
0	0\\
6	0\\
};
\addlegendentry{Equivalent}

\addplot [color=mycolor2, dashed, line width=2.0pt]
  table[row sep=crcr]{%
0	2.65527006901199e-05\\
6	2.65527006901199e-05\\
};
\addlegendentry{Cascaded}

\end{axis}
\end{tikzpicture}%

%% file: assets/simulation/pc_power_sx_txrx16_nd.tex
%
%
\definecolor{mycolor1}{rgb}{0.00000,0.44700,0.74100}%
\definecolor{mycolor2}{rgb}{0.00000,0.44706,0.74118}%
\definecolor{mycolor3}{rgb}{0.85000,0.32500,0.09800}%
\definecolor{mycolor4}{rgb}{0.85098,0.32549,0.09804}%
\definecolor{mycolor5}{rgb}{0.92900,0.69400,0.12500}%
\definecolor{mycolor6}{rgb}{0.92941,0.69412,0.12549}%
\definecolor{mycolor7}{rgb}{0.49400,0.18400,0.55600}%
\definecolor{mycolor8}{rgb}{0.49412,0.18431,0.55686}%
\definecolor{mycolor9}{rgb}{0.46600,0.67400,0.18800}%
\definecolor{mycolor10}{rgb}{0.46667,0.67451,0.18824}%
\definecolor{mycolor11}{rgb}{0.30100,0.74500,0.93300}%
\definecolor{mycolor12}{rgb}{0.50196,0.50196,0.50196}%
\begin{tikzpicture}

\begin{axis}[%
width=9.509cm,
height=7.5cm,
at={(0cm,0cm)},
scale only axis,
xmin=1,
xmax=9,
xtick={1,2,3,4,5,6,7,8,9},
xticklabels={{$2^0$},{$2^1$},{$2^2$},{$2^3$},{$2^4$},{$2^5$},{$2^6$},{$2^7$},{$2^8$}},
xlabel={RIS Group Size},
ymin=0,
ymax=0.000116439758271769,
ylabel={Channel Power [W]},
axis background/.style={fill=white},
xmajorgrids,
ymajorgrids,
legend style={at={(0.03,0.97)}, anchor=north west, legend cell align=left, align=left, draw=white!15!black}
]
\addplot[only marks, mark=o, mark options={solid}, mark size=2pt, line width=2pt, draw=mycolor2, forget plot] table[row sep=crcr]{%
x	y\\
5	8.18421179472296e-07\\
};
\addplot[only marks, mark=o, mark options={solid}, mark size=2pt, line width=2pt, draw=mycolor4, forget plot] table[row sep=crcr]{%
x	y\\
6	2.47795062804422e-06\\
};
\addplot[only marks, mark=o, mark options={solid}, mark size=2pt, line width=2pt, draw=mycolor6, forget plot] table[row sep=crcr]{%
x	y\\
7	8.20088397806638e-06\\
};
\addplot[only marks, mark=o, mark options={solid}, mark size=2pt, line width=2pt, draw=mycolor8, forget plot] table[row sep=crcr]{%
x	y\\
8	2.94451091720999e-05\\
};
\addplot[only marks, mark=o, mark options={solid}, mark size=2pt, line width=2pt, draw=mycolor10, forget plot] table[row sep=crcr]{%
x	y\\
9	0.000111552731718914\\
};
\addplot[draw=none, only marks, mark=o, mark options={solid}, mark size=2pt, line width=2pt, draw=mycolor12] table[row sep=crcr]{%
x	y\\
0	0\\
};
\addlegendentry{Explicit}

\addplot [color=mycolor2, line width=2.0pt]
  table[row sep=crcr]{%
1	5.66270957115429e-07\\
2	6.32214073230813e-07\\
3	7.11789599032912e-07\\
4	7.83501502399535e-07\\
5	8.1582630412623e-07\\
};
\addlegendentry{$N_\mathrm{S} = 16$}

\addplot [color=mycolor4, dashed, line width=2.0pt]
  table[row sep=crcr]{%
1	1.31853309573688e-06\\
2	1.53662356650387e-06\\
3	1.82774899163227e-06\\
4	2.1412614461017e-06\\
5	2.38070294433103e-06\\
6	2.46743024417076e-06\\
};
\addlegendentry{$N_\mathrm{S} = 32$}

\addplot [color=mycolor6, dotted, line width=2.0pt]
  table[row sep=crcr]{%
1	3.16991712195604e-06\\
2	3.88139289638628e-06\\
3	4.87508772014936e-06\\
4	6.11547438007159e-06\\
5	7.2872459799269e-06\\
6	7.95147687920036e-06\\
7	8.1598841971569e-06\\
};
\addlegendentry{$N_\mathrm{S} = 64$}

\addplot [color=mycolor8, dashdotted, line width=2.0pt]
  table[row sep=crcr]{%
1	8.08840132190028e-06\\
2	1.03513925278208e-05\\
3	1.38370561575957e-05\\
4	1.85667340754753e-05\\
5	2.37528268575109e-05\\
6	2.72839029873604e-05\\
7	2.87929093680483e-05\\
8	2.92781883450922e-05\\
};
\addlegendentry{$N_\mathrm{S} = 128$}

\addplot [color=mycolor10, line width=2.0pt]
  table[row sep=crcr]{%
1	2.18463432534312e-05\\
2	2.96551238699234e-05\\
3	4.20172201243997e-05\\
4	5.99730466582087e-05\\
5	8.20696998678401e-05\\
6	9.8504151463339e-05\\
7	0.00010650814145777\\
8	0.000109789012745988\\
9	0.000110895007877875\\
};
\addlegendentry{$N_\mathrm{S} = 256$}

\end{axis}
\end{tikzpicture}%

%% file: assets/simulation/pc_power_sx_txrx16.tex
%
%
\definecolor{mycolor1}{rgb}{0.85000,0.32500,0.09800}%
\definecolor{mycolor2}{rgb}{0.00000,0.44706,0.74118}%
\definecolor{mycolor3}{rgb}{0.92900,0.69400,0.12500}%
\definecolor{mycolor4}{rgb}{0.49400,0.18400,0.55600}%
\definecolor{mycolor5}{rgb}{0.85098,0.32549,0.09804}%
\definecolor{mycolor6}{rgb}{0.46600,0.67400,0.18800}%
\definecolor{mycolor7}{rgb}{0.30100,0.74500,0.93300}%
\definecolor{mycolor8}{rgb}{0.92941,0.69412,0.12549}%
\definecolor{mycolor9}{rgb}{0.63500,0.07800,0.18400}%
\definecolor{mycolor10}{rgb}{0.00000,0.44700,0.74100}%
\definecolor{mycolor11}{rgb}{0.49412,0.18431,0.55686}%
\definecolor{mycolor12}{rgb}{0.46667,0.67451,0.18824}%
\definecolor{mycolor13}{rgb}{0.50196,0.50196,0.50196}%
\begin{tikzpicture}

\begin{axis}[%
width=9.509cm,
height=7.5cm,
at={(0cm,0cm)},
scale only axis,
xmin=1,
xmax=9,
xtick={1,2,3,4,5,6,7,8,9},
xticklabels={{$2^0$},{$2^1$},{$2^2$},{$2^3$},{$2^4$},{$2^5$},{$2^6$},{$2^7$},{$2^8$}},
xlabel={RIS Group Size},
ymode=log,
ymin=7.71952537202106e-05,
ymax=0.000358145391898644,
yminorticks=true,
ylabel={Channel Power [W]},
axis background/.style={fill=white},
xmajorgrids,
ymajorgrids,
yminorgrids,
legend style={at={(0.03,0.97)}, anchor=north west, legend cell align=left, align=left, draw=white!15!black}
]
\addplot [color=black, line width=2.0pt]
  table[row sep=crcr]{%
1	8.12581618107481e-05\\
2	8.12581618107481e-05\\
3	8.12581618107481e-05\\
4	8.12581618107481e-05\\
5	8.12581618107481e-05\\
6	8.12581618107481e-05\\
7	8.12581618107481e-05\\
8	8.12581618107481e-05\\
9	8.12581618107481e-05\\
};
\addlegendentry{No RIS}

\addplot[only marks, mark=triangle, mark options={rotate=90}, mark size=2pt, line width=2pt, draw=mycolor2, forget plot] table[row sep=crcr]{%
x	y\\
5	8.74153313091917e-05\\
};
\addplot[only marks, mark=triangle, mark options={rotate=270}, mark size=2pt, line width=2pt, draw=mycolor2, forget plot] table[row sep=crcr]{%
x	y\\
5	8.74291033971713e-05\\
};
\addplot[only marks, mark=triangle, mark options={rotate=90}, mark size=2pt, line width=2pt, draw=mycolor5, forget plot] table[row sep=crcr]{%
x	y\\
6	9.91843205545709e-05\\
};
\addplot[only marks, mark=triangle, mark options={rotate=270}, mark size=2pt, line width=2pt, draw=mycolor5, forget plot] table[row sep=crcr]{%
x	y\\
6	9.91780210311253e-05\\
};
\addplot[only marks, mark=triangle, mark options={rotate=90}, mark size=2pt, line width=2pt, draw=mycolor8, forget plot] table[row sep=crcr]{%
x	y\\
7	0.00012405412223482\\
};
\addplot[only marks, mark=triangle, mark options={rotate=270}, mark size=2pt, line width=2pt, draw=mycolor8, forget plot] table[row sep=crcr]{%
x	y\\
7	0.000124044149009442\\
};
\addplot[only marks, mark=triangle, mark options={rotate=90}, mark size=2pt, line width=2pt, draw=mycolor11, forget plot] table[row sep=crcr]{%
x	y\\
8	0.000182844868229739\\
};
\addplot[only marks, mark=triangle, mark options={rotate=270}, mark size=2pt, line width=2pt, draw=mycolor11, forget plot] table[row sep=crcr]{%
x	y\\
8	0.000182852647982853\\
};
\addplot[only marks, mark=triangle, mark options={rotate=90}, mark size=2pt, line width=2pt, draw=mycolor12, forget plot] table[row sep=crcr]{%
x	y\\
9	0.000339584399340245\\
};
\addplot[only marks, mark=triangle, mark options={rotate=270}, mark size=2pt, line width=2pt, draw=mycolor12, forget plot] table[row sep=crcr]{%
x	y\\
9	0.000339587956470715\\
};

\addplot[draw=none, only marks, mark=triangle, mark options={rotate=90}, mark size=2pt, line width=2pt, draw=mycolor13] table[row sep=crcr]{%
x	y\\
1	1\\
};
\addlegendentry{OP-left}

\addplot[draw=none, only marks, mark=triangle, mark options={rotate=270}, mark size=2pt, line width=2pt, draw=mycolor13] table[row sep=crcr]{%
x	y\\
1	1\\
};
\addlegendentry{OP-right}

\addplot [color=mycolor2, line width=2.0pt]
  table[row sep=crcr]{%
1	8.41788427928869e-05\\
2	8.50641788557691e-05\\
3	8.63696852707132e-05\\
4	8.79903454326538e-05\\
5	8.96339735645668e-05\\
};
\addlegendentry{$N_\mathrm{S} = 16$}

\addplot [color=mycolor5, dashed, line width=2.0pt]
  table[row sep=crcr]{%
1	8.72783634122255e-05\\
2	8.91184079692796e-05\\
3	9.17640563659579e-05\\
4	9.50977289249081e-05\\
5	9.85503564245316e-05\\
6	0.00010076742152366\\
};
\addlegendentry{$N_\mathrm{S} = 32$}

\addplot [color=mycolor8, dotted, line width=2.0pt]
  table[row sep=crcr]{%
1	9.36120443845964e-05\\
2	9.75056089045287e-05\\
3	0.000102998044584118\\
4	0.000110120638452998\\
5	0.0001177772588008\\
6	0.000122854470762259\\
7	0.000125497666976946\\
};
\addlegendentry{$N_\mathrm{S} = 64$}

\addplot [color=mycolor11, dashdotted, line width=2.0pt]
  table[row sep=crcr]{%
1	0.000106289915229579\\
2	0.000114680256721358\\
3	0.000126826397397359\\
4	0.000142994695102657\\
5	0.00016107230730003\\
6	0.000173889895707073\\
7	0.000180752926606848\\
8	0.000184267352458571\\
};
\addlegendentry{$N_\mathrm{S} = 128$}

\addplot [color=mycolor12, line width=2.0pt]
  table[row sep=crcr]{%
1	0.000134334856638143\\
2	0.000153364621257308\\
3	0.000182125698990213\\
4	0.000221951821498391\\
5	0.000269865717893532\\
6	0.000305861046930524\\
7	0.00032571637890533\\
8	0.000335868097136638\\
9	0.00034109084942728\\
};
\addlegendentry{$N_\mathrm{S} = 256$}

\end{axis}

\end{tikzpicture}%

%% file: assets/simulation/pc_rate_beamforming.tex
%
%
\definecolor{mycolor1}{rgb}{0.00000,0.44706,0.74118}%
\definecolor{mycolor2}{rgb}{0.85098,0.32549,0.09804}%
\definecolor{mycolor3}{rgb}{0.92941,0.69412,0.12549}%
\begin{tikzpicture}

\begin{axis}[%
width=9.509cm,
height=7.5cm,
at={(0cm,0cm)},
scale only axis,
xmin=-20,
xmax=20,
xlabel={Transmit Power [dB]},
ymin=0,
ymax=50,
ylabel={Achievable Rate [bit/s/Hz]},
axis background/.style={fill=white},
xmajorgrids,
ymajorgrids,
legend style={at={(0.03,0.97)}, anchor=north west, legend cell align=left, align=left, draw=white!15!black}
]
\addplot [color=black, line width=2.0pt]
  table[row sep=crcr]{%
-20	0.966447889650251\\
-15	2.13054139165545\\
-10	4.18094583776722\\
-5	7.28839897274867\\
0	11.4412630405996\\
5	16.4565340424558\\
10	22.1955896885685\\
15	28.4236694030567\\
20	34.8930026415245\\
};
\addlegendentry{No RIS}

\addplot [color=mycolor1, line width=2.0pt, mark=o, mark options={solid, mycolor1}]
  table[row sep=crcr]{%
-20	1.78458164361096\\
-15	3.561866909033\\
-10	6.41094164720174\\
-5	10.6060494439492\\
0	15.8685341243575\\
5	22.3729979194697\\
10	29.1590147839292\\
15	35.8614481865284\\
20	42.5460340979809\\
};
\addlegendentry{Alternate: $L = 1$}

\addplot [color=mycolor1, dashed, line width=2.0pt, mark=o, mark options={solid, mycolor1}]
  table[row sep=crcr]{%
-20	1.68142608264932\\
-15	3.47229181555908\\
-10	6.32764150769308\\
-5	10.2843682032806\\
0	15.1638664406851\\
5	20.8059527875408\\
10	26.9552004673932\\
15	33.3756719661617\\
20	39.9334611692689\\
};
\addlegendentry{Decouple: $L = 1$}

\addplot [color=mycolor2, line width=2.0pt, mark=+, mark options={solid, mycolor2}]
  table[row sep=crcr]{%
-20	2.04085797293085\\
-15	4.23025423333992\\
-10	7.66035321117132\\
-5	12.3873778126605\\
0	18.3041387392747\\
5	25.1416115903226\\
10	32.1208560143212\\
15	38.8917948319834\\
20	45.5317843299722\\
};
\addlegendentry{Alternate: $L = 4$}

\addplot [color=mycolor2, dashed, line width=2.0pt, mark=+, mark options={solid, mycolor2}]
  table[row sep=crcr]{%
-20	2.02933277886478\\
-15	4.22295410662709\\
-10	7.76659807581125\\
-5	12.7308814491693\\
0	18.701625254284\\
5	25.1136038999092\\
10	31.6819574777726\\
15	38.3017094469001\\
20	44.9379200377717\\
};
\addlegendentry{Decouple: $L = 4$}

\addplot [color=mycolor3, line width=2.0pt, mark=square, mark options={solid, mycolor3}]
  table[row sep=crcr]{%
-20	2.28407687242564\\
-15	4.81826080186169\\
-10	9.03773074941209\\
-5	14.6404440346335\\
0	20.9116094666672\\
5	27.4324929209647\\
10	34.0375775236663\\
15	40.668902236646\\
20	47.3087929158649\\
};
\addlegendentry{Alternate: $L = 128$}

\addplot [color=mycolor3, dashed, line width=2.0pt, mark=square, mark options={solid, mycolor3}]
  table[row sep=crcr]{%
-20	2.28308476166498\\
-15	4.81938588993846\\
-10	9.03056751903296\\
-5	14.6240066357304\\
0	20.8908816734224\\
5	27.4101252589219\\
10	34.0139992275384\\
15	40.6451530688563\\
20	47.2849864973432\\
};
\addlegendentry{Decouple: $L = 128$}

\end{axis}
\end{tikzpicture}%

%% file: assets/simulation/pc_rate_txrx.tex
%
%
\definecolor{mycolor1}{rgb}{0.00000,0.44706,0.74118}%
\definecolor{mycolor2}{rgb}{0.85098,0.32549,0.09804}%
\definecolor{mycolor3}{rgb}{0.92941,0.69412,0.12549}%
\begin{tikzpicture}

\begin{axis}[%
width=9.509cm,
height=7.5cm,
at={(0cm,0cm)},
scale only axis,
xmin=-20,
xmax=20,
xlabel={Transmit Power [dB]},
ymin=0,
ymax=180,
ylabel={Achievable Rate [bit/s/Hz]},
axis background/.style={fill=white},
xmajorgrids,
ymajorgrids,
legend style={at={(0.03,0.97)}, anchor=north west, legend cell align=left, align=left, draw=white!15!black}
]
\addplot [color=mycolor1, line width=2.0pt, mark=o, mark options={solid, mycolor1}]
  table[row sep=crcr]{%
-20	0.80645914151727\\
-15	1.73412802108339\\
-10	3.04800729466691\\
-5	4.57707843574454\\
0	6.19336035613507\\
5	7.83986275409894\\
10	9.49621942970358\\
15	11.1557232230498\\
20	12.8162253009448\\
};
\addlegendentry{D: $N_\mathrm{T}=N_\mathrm{R} = 1$}

\addplot [color=mycolor1, dashed, line width=2.0pt, mark=o, mark options={solid, mycolor1}]
  table[row sep=crcr]{%
-20	1.02738485492902\\
-15	2.08553552813062\\
-10	3.48388115902665\\
-5	5.04985209432047\\
0	6.67930178784317\\
5	8.33014111798399\\
10	9.98788714882126\\
15	11.6478319270646\\
20	13.3084734885675\\
};
\addlegendentry{BD: $N_\mathrm{T}=N_\mathrm{R} = 1$}

\addplot [color=mycolor2, line width=2.0pt, mark=+, mark options={solid, mycolor2}]
  table[row sep=crcr]{%
-20	1.80216712847512\\
-15	3.61095912930842\\
-10	6.46727118154321\\
-5	10.6728381717045\\
0	15.8663694653186\\
5	22.3025811138217\\
10	29.0071824375286\\
15	35.8429220927395\\
20	42.5257892938348\\
};
\addlegendentry{D: $N_\mathrm{T}=N_\mathrm{R} = 4$}

\addplot [color=mycolor2, dashed, line width=2.0pt, mark=+, mark options={solid, mycolor2}]
  table[row sep=crcr]{%
-20	2.25658747812802\\
-15	4.74132020552887\\
-10	8.90509622893698\\
-5	14.6240914814496\\
0	20.8875622193268\\
5	27.4056906119124\\
10	34.0092394399095\\
15	40.6403427407311\\
20	47.2808661774305\\
};
\addlegendentry{BD: $N_\mathrm{T}=N_\mathrm{R} = 4$}

\addplot [color=mycolor3, line width=2.0pt, mark=square, mark options={solid, mycolor3}]
  table[row sep=crcr]{%
-20	5.00809118408572\\
-15	10.4621038645415\\
-10	19.7374730744499\\
-5	33.550570634643\\
0	51.8679825318363\\
5	73.8730754193531\\
10	98.4593995342972\\
15	124.330060387556\\
20	151.362123996259\\
};
\addlegendentry{D: $N_\mathrm{T}=N_\mathrm{R} = 16$}

\addplot [color=mycolor3, dashed, line width=2.0pt, mark=square, mark options={solid, mycolor3}]
  table[row sep=crcr]{%
-20	6.35841719175107\\
-15	13.4657499296479\\
-10	25.5048236983859\\
-5	43.2319146419582\\
0	66.1428004903052\\
5	91.9243453059882\\
10	118.034135979311\\
15	144.396858442714\\
20	171.015174863974\\
};
\addlegendentry{BD: $N_\mathrm{T}=N_\mathrm{R} = 16$}

\end{axis}
\end{tikzpicture}%

%% file: assets/simulation/pc_rate_sx.tex
%
%
\definecolor{mycolor1}{rgb}{0.00000,0.44706,0.74118}%
\definecolor{mycolor2}{rgb}{0.85098,0.32549,0.09804}%
\definecolor{mycolor3}{rgb}{0.92941,0.69412,0.12549}%
\begin{tikzpicture}

\begin{axis}[%
width=9.509cm,
height=7.5cm,
at={(0cm,0cm)},
scale only axis,
xmin=-20,
xmax=20,
xlabel={Transmit Power [dB]},
ymin=0,
ymax=60,
ylabel={Achievable Rate [bit/s/Hz]},
axis background/.style={fill=white},
xmajorgrids,
ymajorgrids,
legend style={at={(0.03,0.97)}, anchor=north west, legend cell align=left, align=left, draw=white!15!black}
]
\addplot [color=black, line width=2.0pt]
  table[row sep=crcr]{%
-20	0.974924221815307\\
-15	2.155737929667\\
-10	4.22511615579774\\
-5	7.37578165885651\\
0	11.5659715565206\\
5	16.5899473057444\\
10	22.348709082376\\
15	28.5788877981851\\
20	35.0376846379771\\
};
\addlegendentry{No RIS}

\addplot [color=mycolor1, line width=2.0pt, mark=o, mark options={solid, mycolor1}]
  table[row sep=crcr]{%
-20	1.0775951909933\\
-15	2.33797192421853\\
-10	4.5244691403991\\
-5	7.83325616187767\\
0	12.2043252106765\\
5	17.502783862643\\
10	23.5906958194149\\
15	30.0475873965181\\
20	36.6178405883567\\
};
\addlegendentry{D: $N_\mathrm{S} = 16$}

\addplot [color=mycolor1, dashed, line width=2.0pt, mark=o, mark options={solid, mycolor1}]
  table[row sep=crcr]{%
-20	1.10923579115774\\
-15	2.442709029227\\
-10	4.74476500230175\\
-5	8.23321469591226\\
0	12.8056756289994\\
5	18.4128594155627\\
10	24.6430285568875\\
15	31.1476882752096\\
20	37.7327815763973\\
};
\addlegendentry{BD: $N_\mathrm{S} = 16$}

\addplot [color=mycolor2, line width=2.0pt, mark=+, mark options={solid, mycolor2}]
  table[row sep=crcr]{%
-20	1.38911856962193\\
-15	2.88770421854911\\
-10	5.39989878022609\\
-5	9.15288473496818\\
0	13.9621829772168\\
5	19.9297422219992\\
10	26.4702320328737\\
15	33.1585838779561\\
20	39.78349389952\\
};
\addlegendentry{D: $N_\mathrm{S} = 64$}

\addplot [color=mycolor2, dashed, line width=2.0pt, mark=+, mark options={solid, mycolor2}]
  table[row sep=crcr]{%
-20	1.5620206135513\\
-15	3.37793863109278\\
-10	6.49562082648052\\
-5	11.0483565427754\\
0	16.8960888847718\\
5	23.2423074308869\\
10	29.7886558296274\\
15	36.4013262364968\\
20	43.0352862870955\\
};
\addlegendentry{BD: $N_\mathrm{S} = 64$}

\addplot [color=mycolor3, line width=2.0pt, mark=square, mark options={solid, mycolor3}]
  table[row sep=crcr]{%
-20	2.54346268106081\\
-15	4.91381388788593\\
-10	8.42373765644029\\
-5	13.3216188303469\\
0	19.1393982680901\\
5	26.3802267396723\\
10	33.2591092328213\\
15	39.9877368645945\\
20	46.6278113854872\\
};
\addlegendentry{D: $N_\mathrm{S} = 256$}

\addplot [color=mycolor3, dashed, line width=2.0pt, mark=square, mark options={solid, mycolor3}]
  table[row sep=crcr]{%
-20	3.8102088263969\\
-15	7.91422759898922\\
-10	13.5945990292298\\
-5	19.8260573275662\\
0	26.3330631578426\\
5	32.9330300712558\\
10	39.5629517769786\\
15	46.2024228690652\\
20	52.8449051313244\\
};
\addlegendentry{BD: $N_\mathrm{S} = 256$}

\end{axis}
\end{tikzpicture}%

%% file: assets/simulation/pc_rate_kfactor.tex
%
%
\definecolor{mycolor1}{rgb}{0.00000,0.44706,0.74118}%
\definecolor{mycolor2}{rgb}{0.85098,0.32549,0.09804}%
\definecolor{mycolor3}{rgb}{0.92941,0.69412,0.12549}%
\begin{tikzpicture}

\begin{axis}[%
width=9.509cm,
height=7.5cm,
at={(0cm,0cm)},
scale only axis,
xmin=-20,
xmax=20,
xlabel={Transmit Power [dB]},
ymin=0,
ymax=50,
ylabel={Achievable Rate [bit/s/Hz]},
axis background/.style={fill=white},
xmajorgrids,
ymajorgrids,
legend style={at={(0.03,0.97)}, anchor=north west, legend cell align=left, align=left, draw=white!15!black}
]
\addplot [color=mycolor1, line width=2.0pt, mark=o, mark options={solid, mycolor1}]
  table[row sep=crcr]{%
-20	1.80821876007788\\
-15	3.63754065418823\\
-10	6.52397177287891\\
-5	10.7375410087883\\
0	16.0801140220357\\
5	22.5773384000589\\
10	29.3334425394153\\
15	36.0021963018504\\
20	42.6917330184181\\
};
\addlegendentry{D: $\kappa_\mathrm{F} = \kappa_\mathrm{B} = 0$}

\addplot [color=mycolor1, dashed, line width=2.0pt, mark=o, mark options={solid, mycolor1}]
  table[row sep=crcr]{%
-20	2.26507172145531\\
-15	4.75770801019453\\
-10	8.99003046555732\\
-5	14.7118193453058\\
0	20.9839343618143\\
5	27.5050240775648\\
10	34.1095125803414\\
15	40.7408675077114\\
20	47.3807680556974\\
};
\addlegendentry{BD: $\kappa_\mathrm{F} = \kappa_\mathrm{B} = 0$}

\addplot [color=mycolor2, line width=2.0pt, mark=+, mark options={solid, mycolor2}]
  table[row sep=crcr]{%
-20	3.35042591902828\\
-15	5.13668237212932\\
-10	7.72882256094737\\
-5	11.3034247828302\\
0	15.8185759560096\\
5	21.0986699036873\\
10	27.1640911636977\\
15	33.6386180785618\\
20	40.2353537199864\\
};
\addlegendentry{D: $\kappa_\mathrm{F} = \kappa_\mathrm{B} = 10$}

\addplot [color=mycolor2, dashed, line width=2.0pt, mark=+, mark options={solid, mycolor2}]
  table[row sep=crcr]{%
-20	3.38008248429671\\
-15	5.24886942440445\\
-10	7.98495739756983\\
-5	11.7904129482671\\
0	16.5603601097506\\
5	22.2885132742307\\
10	28.6366348668615\\
15	35.1975720349819\\
20	41.8146872560824\\
};
\addlegendentry{BD: $\kappa_\mathrm{F} = \kappa_\mathrm{B} = 10$}

\addplot [color=mycolor3, line width=2.0pt, mark=square, mark options={solid, mycolor3}]
  table[row sep=crcr]{%
-20	3.5341539120381\\
-15	5.32663027748093\\
-10	7.91783240648881\\
-5	11.5025021870777\\
0	16.0242740767472\\
5	21.3232680679002\\
10	27.2833878226791\\
15	33.6787670451286\\
20	40.2329094885052\\
};
\addlegendentry{{D: $\kappa_\mathrm{F} = \kappa_\mathrm{B} \to \infty$}}

\addplot [color=mycolor3, dashed, line width=2.0pt, mark=square, mark options={solid, mycolor3}]
  table[row sep=crcr]{%
-20	3.53634968366243\\
-15	5.33057914284338\\
-10	7.92426624828917\\
-5	11.5073024887957\\
0	16.0274673892275\\
5	21.3283494881659\\
10	27.3020235801385\\
15	33.6812074348956\\
20	40.2364830785231\\
};
\addlegendentry{{BD: $\kappa_\mathrm{F} = \kappa_\mathrm{B} \to \infty$}}

\end{axis}
\end{tikzpicture}%

%% file: assets/simulation/pc_power_symmetry.tex
%
%
\definecolor{mycolor1}{rgb}{0.00000,0.44706,0.74118}%
\definecolor{mycolor2}{rgb}{0.85098,0.32549,0.09804}%
\begin{tikzpicture}

\begin{axis}[%
width=9.509cm,
height=7.5cm,
at={(0cm,0cm)},
scale only axis,
xmin=1,
xmax=7,
xtick={1,2,3,4,5,6,7},
xticklabels={{$2^0$},{$2^1$},{$2^2$},{$2^3$},{$2^4$},{$2^5$},{$2^6$}},
xlabel={RIS Group Size},
ymin=0,
ymax=9e-07,
ylabel={Channel Power [W]},
axis background/.style={fill=white},
xmajorgrids,
ymajorgrids,
legend style={at={(0.48,0.8)}, anchor=north west, legend cell align=left, align=left, draw=white!15!black}
]
\addplot [color=mycolor1, line width=2.0pt, mark=o, mark options={solid, mycolor1}]
  table[row sep=crcr]{%
1	3.07226159314351e-08\\
2	4.11492399332381e-08\\
3	4.88148530601955e-08\\
4	5.26094766034717e-08\\
5	5.43186358023285e-08\\
};
\addlegendentry{Asymmetric: $N_\mathrm{S} = 16$}

\addplot [color=mycolor1, dashed, line width=2.0pt, mark=o, mark options={solid, mycolor1}]
  table[row sep=crcr]{%
1	3.07226159314351e-08\\
2	2.48110921829919e-08\\
3	2.09999716084868e-08\\
4	1.85392359072275e-08\\
5	1.69392170352418e-08\\
};
\addlegendentry{Enforced: $N_\mathrm{S} = 16$}

\addplot [color=mycolor1, dotted, line width=2.0pt, mark=o, mark options={solid, mycolor1}]
  table[row sep=crcr]{%
1	2.18823167401622e-08\\
2	2.58667252589526e-08\\
3	3.6393122841988e-08\\
4	4.41807731672663e-08\\
5	4.8483723780213e-08\\
};
\addlegendentry{Legacy: $N_\mathrm{S} = 16$}

\addplot [color=mycolor1, dashdotted, line width=2.0pt, mark=o, mark options={solid, mycolor1}]
  table[row sep=crcr]{%
1	3.07226159314351e-08\\
2	3.78824565046636e-08\\
3	4.6008649901077e-08\\
4	5.13929112888468e-08\\
5	5.39731769119744e-08\\
};
\addlegendentry{Takagi: $N_\mathrm{S} = 16$}

\addplot [color=mycolor2, line width=2.0pt, mark=+, mark options={solid, mycolor2}]
  table[row sep=crcr]{%
1	3.73674394477017e-07\\
2	5.420346500246e-07\\
3	6.85389992643681e-07\\
4	7.6227451113248e-07\\
5	8.0069480829595e-07\\
6	8.18206267524393e-07\\
7	8.25088677505729e-07\\
};
\addlegendentry{Asymmetric: $N_\mathrm{S} = 64$}

\addplot [color=mycolor2, dashed, line width=2.0pt, mark=+, mark options={solid, mycolor2}]
  table[row sep=crcr]{%
1	3.73674394477017e-07\\
2	3.06957520044884e-07\\
3	2.68618316241873e-07\\
4	2.45748953801729e-07\\
5	2.27807576071471e-07\\
6	2.22821968007804e-07\\
7	2.16970499497674e-07\\
};
\addlegendentry{Enforced: $N_\mathrm{S} = 64$}

\addplot [color=mycolor2, dotted, line width=2.0pt, mark=+, mark options={solid, mycolor2}]
  table[row sep=crcr]{%
1	3.02768340133104e-07\\
2	3.67874038919462e-07\\
3	5.47929160259681e-07\\
4	6.80081694767098e-07\\
5	7.48298918253076e-07\\
6	7.85561233493806e-07\\
7	8.04051669342744e-07\\
};
\addlegendentry{Legacy: $N_\mathrm{S} = 64$}

\addplot [color=mycolor2, dashdotted, line width=2.0pt, mark=+, mark options={solid, mycolor2}]
  table[row sep=crcr]{%
1	3.73674394477017e-07\\
2	4.73667276858297e-07\\
3	6.21291346950978e-07\\
4	7.28822615730252e-07\\
5	7.83509460162992e-07\\
6	8.10199796739281e-07\\
7	8.22445224822532e-07\\
};
\addlegendentry{Takagi: $N_\mathrm{S} = 64$}

\end{axis}
\end{tikzpicture}%

%% file: assets/simulation/pc_rate_symmetry.tex
%
%
\definecolor{mycolor1}{rgb}{0.00000,0.44706,0.74118}%
\definecolor{mycolor2}{rgb}{0.85098,0.32549,0.09804}%
\begin{tikzpicture}

\begin{axis}[%
width=9.509cm,
height=7.5cm,
at={(0cm,0cm)},
scale only axis,
xmin=-20,
xmax=20,
xlabel={Transmit Power [dB]},
ymin=0,
ymax=60,
ylabel={Achievable Rate [bit/s/Hz]},
axis background/.style={fill=white},
xmajorgrids,
ymajorgrids,
legend style={at={(0.03,0.97)}, anchor=north west, legend cell align=left, align=left, draw=white!15!black}
]
\addplot [color=black, line width=2.0pt]
  table[row sep=crcr]{%
-20	0.980890669970127\\
-15	2.17631704465586\\
-10	4.26602616349385\\
-5	7.40450506397734\\
0	11.5417650232559\\
5	16.5267353801558\\
10	22.2389737269901\\
15	28.45808707201\\
20	34.9262681793615\\
};
\addlegendentry{No RIS}

\addplot [color=mycolor1, line width=2.0pt, mark=o, mark options={solid, mycolor1}]
  table[row sep=crcr]{%
-20	1.12063253355975\\
-15	2.46573489743792\\
-10	4.77754000084507\\
-5	8.23988306278143\\
0	12.742009858353\\
5	18.2478590438694\\
10	24.5238055280933\\
15	31.066076634542\\
20	37.6934502155395\\
};
\addlegendentry{Asymmetric BD: $N_\mathrm{S} = 16$}

\addplot [color=mycolor1, dashed, line width=2.0pt, mark=o, mark options={solid, mycolor1}]
  table[row sep=crcr]{%
-20	1.05580750715398\\
-15	2.32934800427963\\
-10	4.53857611475026\\
-5	7.85161754589999\\
0	12.1875414322358\\
5	17.4725755956921\\
10	23.5215773006871\\
15	29.9576641834048\\
20	36.545658607323\\
};
\addlegendentry{Enforced BD: $N_\mathrm{S} = 16$}

\addplot [color=mycolor1, dotted, line width=2.0pt, mark=o, mark options={solid, mycolor1}]
  table[row sep=crcr]{%
-20	1.05604829532615\\
-15	2.32528757764103\\
-10	4.51652360379757\\
-5	7.81092268694229\\
0	12.1322770314351\\
5	17.3640485846909\\
10	23.4102621948832\\
15	29.8328777883331\\
20	36.3827795168908\\
};
\addlegendentry{Projected BD: $N_\mathrm{S} = 16$}

\addplot [color=mycolor2, line width=2.0pt, mark=+, mark options={solid, mycolor2}]
  table[row sep=crcr]{%
-20	3.21307234456805\\
-15	6.44489264840814\\
-10	11.1055110729233\\
-5	16.9451867599123\\
0	24.3067729620249\\
5	32.0556428907967\\
10	39.3125340285854\\
15	46.0961194117957\\
20	52.8338188375336\\
};
\addlegendentry{Asymmetric BD: $N_\mathrm{S} = 256$}

\addplot [color=mycolor2, dashed, line width=2.0pt, mark=+, mark options={solid, mycolor2}]
  table[row sep=crcr]{%
-20	2.13313601839314\\
-15	4.51556782846934\\
-10	8.24543732693642\\
-5	13.2870464928264\\
0	19.7620253796895\\
5	26.9290403889184\\
10	33.9580771550036\\
15	40.7022106761348\\
20	47.4165131438527\\
};
\addlegendentry{Enforced BD: $N_\mathrm{S} = 256$}

\addplot [color=mycolor2, dotted, line width=2.0pt, mark=+, mark options={solid, mycolor2}]
  table[row sep=crcr]{%
-20	2.15500488613771\\
-15	4.96728432099124\\
-10	8.26883664746099\\
-5	14.3003861999247\\
0	20.4023452509746\\
5	28.2965717556944\\
10	35.5863371396393\\
15	42.393566131641\\
20	49.1026041293653\\
};
\addlegendentry{Projected BD: $N_\mathrm{S} = 256$}

\end{axis}
\end{tikzpicture}%

%% file: assets/simulation/pc_singular_csi.tex
%
%
\definecolor{mycolor1}{rgb}{0.00000,0.44700,0.74100}%
\definecolor{mycolor2}{rgb}{0.92900,0.69400,0.12500}%
\definecolor{mycolor3}{rgb}{0.49400,0.18400,0.55600}%
\definecolor{mycolor4}{rgb}{0.63500,0.07800,0.18400}%
\definecolor{mycolor5}{rgb}{0.46600,0.67400,0.18800}%
\definecolor{mycolor6}{rgb}{0.85000,0.32500,0.09800}%
\definecolor{mycolor7}{rgb}{0.30100,0.74500,0.93300}%
\definecolor{mycolor8}{rgb}{0.50196,0.50196,0.50196}%
\begin{tikzpicture}

\begin{axis}[%
width=9.509cm,
height=7.5cm,
at={(0cm,0cm)},
scale only axis,
unbounded coords=jump,
xmin=0.000292357955879774,
xmax=0.00146931598006191,
xlabel={$\sigma_1(\mathbf{H})$},
ymin=5.00353872452725e-07,
ymax=0.000855139152254738,
ylabel={$\sigma_2(\mathbf{H})$},
axis background/.style={fill=white},
xmajorgrids,
ymajorgrids,
legend style={at={(0.03,0.97)}, anchor=north west, legend cell align=left, align=left, draw=white!15!black},
every axis plot/.append style={line width=1.5pt}
]
\addplot[only marks, mark=triangle, mark options={}, mark size=4pt, draw=black] table[row sep=crcr]{%
x	y\\
0.000839165048501169	0.000214263316833017\\
};
\addlegendentry{Direct}

\addplot [color=mycolor1, dashed, line width=3.0pt]
  table[row sep=crcr]{%
0.00120932633718953	0.000583119017598987\\
0.00119575201720946	0.000596081239213197\\
0.00117726551646851	0.000611861101174529\\
0.00116606299082248	0.000620580102997392\\
0.00115424592284094	0.000628870538828447\\
0.00113886771388772	0.000638539772491746\\
0.00111898950928878	0.000649774010006312\\
0.00110183012413076	0.000658373581049715\\
0.00108325437257282	0.000666566528121604\\
0.001063768616969	0.000673393058199169\\
0.00104157703419366	0.000680312749731737\\
0.00102940704848264	0.000683655627075032\\
0.00101534590489402	0.000686767298208718\\
0.0010026917822706	0.000688974641564217\\
0.00099097088499458	0.000690411158545693\\
0.000975869339018042	0.000691490294174907\\
0.00096019933715482	0.000691860001599399\\
0.000944270055421307	0.000691427925266965\\
0.000923381543260725	0.000689880173702533\\
0.000906338430168291	0.000687778553566788\\
0.000890771657779231	0.000685055092600335\\
0.000874967807227278	0.000681494832282678\\
0.000860651531189936	0.000677538166163577\\
0.00084263466302954	0.000671530981556631\\
0.000833179943934603	0.000667891247037841\\
0.000825096636366874	0.00066430463113554\\
0.000817622286439315	0.000660533145039521\\
0.000810229737161579	0.000656334038054984\\
0.000802689225471437	0.000651563607280594\\
0.000797565212369227	0.000648027775864695\\
0.00079153891940423	0.00064351105664721\\
0.000784321813851128	0.000637592054903147\\
0.000777733164379927	0.000631712159747391\\
0.000441524298685212	0.000296037981911999\\
0.000434469809806145	0.00028814322809853\\
0.000428092651889828	0.000280170072109908\\
0.000422631513802745	0.000272581400590931\\
0.000417648726184689	0.000264924313108306\\
0.000412471420119124	0.000256109633332169\\
0.000407648973867364	0.000246885310746728\\
0.000403028531630567	0.000236849224421848\\
0.000398822264577218	0.000226476405700697\\
0.000394233101548728	0.000213511948147064\\
0.000390225450359751	0.000200091453867316\\
0.000386614485758307	0.000185348027264666\\
0.000383675555943931	0.000170337808639181\\
0.000380968519648806	0.000151900607665391\\
0.000379059275924778	0.000132622462208425\\
0.000378194792996519	0.000114943948967252\\
0.000378226418735324	9.89468700138412e-05\\
0.000378884085713094	8.62508167801333e-05\\
0.000380569640365825	7.03420220286831e-05\\
0.00038618679753343	3.2704930777993e-05\\
0.000392139965165464	1.76545987441939e-05\\
0.000398143101171136	4.42592319687909e-06\\
0.000400015585409934	7.84639312588013e-07\\
0.000400890956764154	5.36547173317211e-07\\
0.000401760254197925	5.01004440513917e-07\\
0.000401781864355714	5.00353872452725e-07\\
0.00130733817970372	1.04917083605326e-06\\
0.00130733817970373	1.04917083605349e-06\\
0.00130733824343437	1.04934572103263e-06\\
0.00131566404537624	2.75924910873365e-05\\
0.00132002132522115	4.31868231132689e-05\\
0.00132362706051804	5.93309503755032e-05\\
0.00132807382278624	8.4683899850467e-05\\
0.00133195793694021	0.000116308252543656\\
0.00133439025321267	0.000150391916544675\\
0.001335396584083	0.000193411167999653\\
0.00133390666371673	0.000252240584333186\\
0.00133056020206507	0.000297825772154749\\
0.00132669877840738	0.000329091629685837\\
0.00132209572955983	0.000355429694817544\\
0.00131557906744292	0.000384322067807563\\
0.00130856007339145	0.000409670350270726\\
0.00130092669810398	0.000432619438275947\\
0.00128961290453414	0.000461695061770904\\
0.00127993906730125	0.000483570717832979\\
0.00127155738904107	0.000500166508701635\\
0.0012630999277564	0.000515008980512333\\
0.00125414139085769	0.000529062388896712\\
0.00124447591477809	0.000542698501702921\\
0.00123400340183044	0.000556039848291126\\
0.001222422875848	0.000569395636875488\\
0.00120932633718953	0.000583119017598987\\
};
\addlegendentry{D: $\epsilon = 0.01$}

\addplot[only marks, mark=o, mark options={}, mark size=5pt, draw=mycolor1, forget plot] table[row sep=crcr]{%
x	y\\
0.00130856007339145	0.000409670350270726\\
};
\addplot[only marks, mark=x, mark options={}, mark size=5pt, draw=mycolor1, forget plot] table[row sep=crcr]{%
x	y\\
0.001335396584083	0.000193411167999653\\
};
\addplot [color=mycolor1, line width=3.0pt, forget plot]
  table[row sep=crcr]{%
0.001335396584083	0.000193411167999653\\
0.00133390666371673	0.000252240584333186\\
0.00133056020206507	0.000297825772154749\\
0.00132669877840738	0.000329091629685837\\
0.00132209572955983	0.000355429694817544\\
0.00131557906744292	0.000384322067807563\\
0.00130856007339145	0.000409670350270726\\
0.00130092669810398	0.000432619438275947\\
0.00128961290453414	0.000461695061770904\\
0.00127993906730125	0.000483570717832979\\
0.00127155738904107	0.000500166508701635\\
0.0012630999277564	0.000515008980512333\\
0.00125414139085769	0.000529062388896712\\
0.00124447591477809	0.000542698501702921\\
0.00123400340183044	0.000556039848291126\\
0.001222422875848	0.000569395636875488\\
0.00120932633718953	0.000583119017598987\\
0.00120932633718953	0.000583119017598987\\
0.00119575201720946	0.000596081239213197\\
0.00117726551646851	0.000611861101174529\\
0.00116606299082248	0.000620580102997392\\
0.00115424592284094	0.000628870538828447\\
0.00113886771388772	0.000638539772491746\\
0.00111898950928878	0.000649774010006312\\
0.00110183012413076	0.000658373581049715\\
0.00108325437257282	0.000666566528121604\\
0.001063768616969	0.000673393058199169\\
0.00104157703419366	0.000680312749731737\\
0.00102940704848264	0.000683655627075032\\
0.00101534590489402	0.000686767298208718\\
0.0010026917822706	0.000688974641564217\\
0.00099097088499458	0.000690411158545693\\
0.000975869339018042	0.000691490294174907\\
0.00096019933715482	0.000691860001599399\\
};
\addplot [opacity=0.667, color=mycolor1, dashed, line width=3.0pt]
  table[row sep=crcr]{%
0.00120723012416074	0.000580649381500933\\
0.00119572729949016	0.000591745547210129\\
0.00118323402865883	0.000602798267829881\\
0.00117091537169704	0.000612839661119823\\
0.00115871195379311	0.000621853409700489\\
0.00114557630666833	0.000630548168328466\\
0.00112381067899623	0.000643396496125846\\
0.00110490865706361	0.00065352458175076\\
0.00108323062533875	0.00066324307023993\\
0.00105217157633789	0.000674306046892957\\
0.00103852102435755	0.00067826603848561\\
0.00102662790265939	0.000681104564886569\\
0.00101556058282178	0.00068328724779443\\
0.00100435117234784	0.000685068812689737\\
0.000992864015068341	0.000686406907038555\\
0.000979441788478403	0.000687342785378085\\
0.000963855285803246	0.000687717102169256\\
0.000948058830893159	0.000687280169509713\\
0.00092940877843491	0.000685885826937215\\
0.000911571441685017	0.000683711468008641\\
0.000894387573610684	0.000680699586104162\\
0.000877435937001939	0.000676840186873352\\
0.000857390104973841	0.000670820436095137\\
0.000846247012304084	0.000667074354965531\\
0.00083702695229735	0.000663404368656974\\
0.000828405618219941	0.000659373359582414\\
0.000820307528666392	0.000655058369952753\\
0.000813154779696064	0.000650847743907278\\
0.000806694137532323	0.000646665605428428\\
0.000802863701419803	0.00064400843587972\\
0.000797548426269461	0.00064002686775209\\
0.000791726365677255	0.000635281026831634\\
0.000785241970194614	0.000629484141544239\\
0.000450253072269392	0.000295130453124787\\
0.000443925306443406	0.000288248687422809\\
0.000437864912868001	0.000280935921502853\\
0.000432334917526517	0.000273476743656745\\
0.0004261872727151	0.000264332926337528\\
0.000420776171850977	0.000255147287321726\\
0.000415063280504114	0.000244224441864073\\
0.000411023454093519	0.000235210593730386\\
0.000406472971950957	0.000222764293180299\\
0.000401019635746126	0.000207189977302398\\
0.00039697662692898	0.000193114474755604\\
0.000393733215238289	0.000179239481115576\\
0.000391058369075111	0.000165112158252769\\
0.000388549264458002	0.000147975470967991\\
0.000386780242306876	0.000129368296321085\\
0.000386065682272045	0.00011437207647525\\
0.000386038743317747	0.000101193805734313\\
0.000386569086633482	8.84784721748764e-05\\
0.00038764069524001	7.63670297314887e-05\\
0.000392923494859776	3.45641459698829e-05\\
0.000396450893267522	1.65206667761283e-05\\
0.000398597614915704	1.00099888062855e-05\\
0.000400336315289976	7.24206559109498e-06\\
0.000402653082589453	6.26178688962721e-06\\
0.000402717796617414	6.23904507595133e-06\\
0.000402717796617414	6.2390450759513e-06\\
0.000583868021978105	5.43299949955658e-06\\
0.00130375716626679	6.79156794051234e-06\\
0.0013041937899189	7.84000214723081e-06\\
0.00131228500208719	3.40132985725704e-05\\
0.00131624094647624	4.86094735783491e-05\\
0.00131969258780504	6.43480463589392e-05\\
0.00132317871828661	8.44959404070106e-05\\
0.00132773748084538	0.000115724237911933\\
0.00133001746507462	0.000142934636115753\\
0.00133088767352632	0.000186670530727628\\
0.00133017374050803	0.000247743221422876\\
0.00132773744729883	0.000287931387759637\\
0.00132422561125768	0.000320204358037492\\
0.00131943435417675	0.000349507358963146\\
0.00131329985055898	0.000378000486951702\\
0.00130676080816937	0.000402926442904334\\
0.00129953974025189	0.000425922991478535\\
0.00129055003033978	0.000450226741951118\\
0.00128057806813479	0.000473892662606403\\
0.00127142153795021	0.000493134650385841\\
0.00126268033911301	0.000509168626837649\\
0.00125350499759553	0.000524077981137843\\
0.00124378686109559	0.000538195569220947\\
0.00123312164972066	0.000552094566569883\\
0.00121990237114162	0.000567381051703403\\
0.00120723012416074	0.000580649381500933\\
};
\addlegendentry{D: $\epsilon = 0.1$}

\addplot[opacity=0.667, only marks, mark=o, mark options={}, mark size=5pt, draw=mycolor1, forget plot] table[row sep=crcr]{%
x	y\\
0.00129953974025189	0.000425922991478535\\
};
\addplot[opacity=0.667, only marks, mark=x, mark options={}, mark size=5pt, draw=mycolor1, forget plot] table[row sep=crcr]{%
x	y\\
0.00133088767352632	0.000186670530727628\\
};
\addplot [opacity=0.667, color=mycolor1, line width=3.0pt, forget plot]
  table[row sep=crcr]{%
0.00133088767352632	0.000186670530727628\\
0.00133017374050803	0.000247743221422876\\
0.00132773744729883	0.000287931387759637\\
0.00132422561125768	0.000320204358037492\\
0.00131943435417675	0.000349507358963146\\
0.00131329985055898	0.000378000486951702\\
0.00130676080816937	0.000402926442904334\\
0.00129953974025189	0.000425922991478535\\
0.00129055003033978	0.000450226741951118\\
0.00128057806813479	0.000473892662606403\\
0.00127142153795021	0.000493134650385841\\
0.00126268033911301	0.000509168626837649\\
0.00125350499759553	0.000524077981137843\\
0.00124378686109559	0.000538195569220947\\
0.00123312164972066	0.000552094566569883\\
0.00121990237114162	0.000567381051703403\\
0.00120723012416074	0.000580649381500933\\
0.00120723012416074	0.000580649381500933\\
0.00119572729949016	0.000591745547210129\\
0.00118323402865883	0.000602798267829881\\
0.00117091537169704	0.000612839661119823\\
0.00115871195379311	0.000621853409700489\\
0.00114557630666833	0.000630548168328466\\
0.00112381067899623	0.000643396496125846\\
0.00110490865706361	0.00065352458175076\\
0.00108323062533875	0.00066324307023993\\
0.00105217157633789	0.000674306046892957\\
0.00103852102435755	0.00067826603848561\\
0.00102662790265939	0.000681104564886569\\
0.00101556058282178	0.00068328724779443\\
0.00100435117234784	0.000685068812689737\\
0.000992864015068341	0.000686406907038555\\
0.000979441788478403	0.000687342785378085\\
0.000963855285803246	0.000687717102169256\\
};
\addplot [opacity=0.333, color=mycolor1, dashed, line width=3.0pt]
  table[row sep=crcr]{%
0.00110722127930191	0.00055231862569117\\
0.0010950603254174	0.00055964768047213\\
0.00108547410172965	0.00056526768014669\\
0.00108090453727743	0.000567824281317283\\
0.00107660910789735	0.000569850547760194\\
0.00107229143829399	0.000571622340327614\\
0.00106773267965282	0.000573279032414328\\
0.00106256311230578	0.000574969360498219\\
0.00104784919952727	0.000579469887650121\\
0.00103874039558323	0.000581995247615636\\
0.00102883599392078	0.00058367868645446\\
0.00101819521148904	0.000584080812793491\\
0.000978855415163161	0.000583380406652782\\
0.000928523427871111	0.000580022438230062\\
0.000917984747042123	0.000578929961136174\\
0.000908304406902503	0.000577464606670763\\
0.000899104704584145	0.000575636018707832\\
0.000890493529967113	0.00057352492657437\\
0.000882601773118892	0.000571207663249651\\
0.00087535130087042	0.000568686253554518\\
0.000868517872535676	0.000565888024055649\\
0.000861873785990023	0.000562710633228406\\
0.000855280934632879	0.000559088933603478\\
0.000848716167899615	0.000555044826053949\\
0.000838296885173201	0.000548097883356892\\
0.00083284725164066	0.000544123720530958\\
0.000827369422440948	0.000539643693075147\\
0.000522389860125143	0.000288044812067282\\
0.000516516721686327	0.000281285041617264\\
0.000510684009963046	0.000273400693690491\\
0.000505000957031184	0.000264800118145228\\
0.000499359074801272	0.000255419148656063\\
0.000494439072797046	0.000245427354342094\\
0.000489286065066436	0.000234699824553598\\
0.000483329731374156	0.000221903685132653\\
0.000477640370699273	0.000207987936565\\
0.000471857300041654	0.000191090776375776\\
0.000468005947815214	0.000175452431539676\\
0.000466258026820786	0.000162539849289817\\
0.000464878489647416	0.00015071779912929\\
0.000463932838780443	0.000138354327627574\\
0.000463209986977303	0.000126711524489337\\
0.000463045849092032	0.00011421695454841\\
0.000464076147870395	9.72894347089132e-05\\
0.000466126418233955	7.278474756215e-05\\
0.000467747387169956	6.79953909667504e-05\\
0.000619547044543305	6.21542263318127e-05\\
0.00124921683817881	9.27965760077673e-05\\
0.00124921683817882	9.27965760077696e-05\\
0.00125061928271792	9.74490799027226e-05\\
0.00126064627294896	0.000137492920950342\\
0.00126207053113663	0.000148633952727836\\
0.00126268994822494	0.000157880740079037\\
0.001262775195004	0.000181063828595574\\
0.00125962161650267	0.000229932308940695\\
0.00125507144106471	0.000273174115566032\\
0.00125051881218205	0.000311623680125553\\
0.00124051948090657	0.000347659681552022\\
0.00123385608693869	0.000371317295327494\\
0.00122748631434108	0.000393853984872361\\
0.00121944327786394	0.000415934092633447\\
0.00121040506725937	0.000435820712044041\\
0.00118962404774373	0.000476528980532702\\
0.00118085669557352	0.000487935282336488\\
0.00117183625322355	0.000498471436893667\\
0.00116256638850969	0.000508203035476104\\
0.00115302695345452	0.000517222822417311\\
0.00114299471484386	0.000525809792659205\\
0.0011317939784689	0.000534311032065687\\
0.00110722127930191	0.00055231862569117\\
};
\addlegendentry{D: $\epsilon = 0.5$}

\addplot[opacity=0.333, only marks, mark=o, mark options={}, mark size=5pt, draw=mycolor1, forget plot] table[row sep=crcr]{%
x	y\\
0.00122748631434108	0.000393853984872361\\
};
\addplot[opacity=0.333, only marks, mark=x, mark options={}, mark size=5pt, draw=mycolor1, forget plot] table[row sep=crcr]{%
x	y\\
0.001262775195004	0.000181063828595574\\
};
\addplot [opacity=0.333, color=mycolor1, line width=3.0pt, forget plot]
  table[row sep=crcr]{%
0.001262775195004	0.000181063828595574\\
0.00125962161650267	0.000229932308940695\\
0.00125507144106471	0.000273174115566032\\
0.00125051881218205	0.000311623680125553\\
0.00124051948090657	0.000347659681552022\\
0.00123385608693869	0.000371317295327494\\
0.00122748631434108	0.000393853984872361\\
0.00121944327786394	0.000415934092633447\\
0.00121040506725937	0.000435820712044041\\
0.00118962404774373	0.000476528980532702\\
0.00118085669557352	0.000487935282336488\\
0.00117183625322355	0.000498471436893667\\
0.00116256638850969	0.000508203035476104\\
0.00115302695345452	0.000517222822417311\\
0.00114299471484386	0.000525809792659205\\
0.0011317939784689	0.000534311032065687\\
0.00110722127930191	0.00055231862569117\\
0.00110722127930191	0.00055231862569117\\
0.0010950603254174	0.00055964768047213\\
0.00108547410172965	0.00056526768014669\\
0.00108090453727743	0.000567824281317283\\
0.00107660910789735	0.000569850547760194\\
0.00107229143829399	0.000571622340327614\\
0.00106773267965282	0.000573279032414328\\
0.00106256311230578	0.000574969360498219\\
0.00104784919952727	0.000579469887650121\\
0.00103874039558323	0.000581995247615636\\
0.00102883599392078	0.00058367868645446\\
0.00101819521148904	0.000584080812793491\\
};
\addplot [color=mycolor6, dashed, line width=3.0pt]
  table[row sep=crcr]{%
0.00146681015044238	0.000852691160218086\\
0.00146655353540688	0.00085293323578946\\
0.0014662843547074	0.000853162890053388\\
0.0014659999443247	0.000853382058906823\\
0.00146569873492164	0.000853591142986456\\
0.00146537786977228	0.000853790987445657\\
0.00146503583744612	0.000853981103932059\\
0.00146466901723951	0.000854161789887561\\
0.00146427408817534	0.000854332555068325\\
0.00146384588816425	0.000854493022963016\\
0.00146338001330438	0.000854641789717379\\
0.00146286837965225	0.000854777647926877\\
0.00146230537852938	0.00085489776865824\\
0.0014616804308002	0.000854999222340005\\
0.00146098104619341	0.000855077662321666\\
0.00146019189894725	0.000855127076148748\\
0.00145929133993952	0.000855139152254738\\
0.000856795082982419	0.000854198016477824\\
0.000856109717751413	0.000853761150395599\\
0.000292531578686958	0.000291275356291254\\
0.0002924996535061	0.000289286117898126\\
0.000292468450754768	0.000276170063242794\\
0.000292438505603605	0.000227020923484573\\
0.000292357955879774	9.57693921944774e-07\\
0.00029235864001388	9.52452028573311e-07\\
0.000292359025324758	9.51243212844082e-07\\
0.00146795919069224	9.55327674380309e-07\\
0.00146795919069226	9.55327676142196e-07\\
0.00146931598006191	0.0008451691026365\\
0.00146929838454492	0.0008460739359627\\
0.00146924392219391	0.000846865155314428\\
0.00146916024584968	0.000847570100281123\\
0.00146905402164446	0.000848197473863468\\
0.00146892945901967	0.000848761721462629\\
0.00146878894868572	0.000849275354436091\\
0.0014686366705857	0.000849740048392178\\
0.00146847258403339	0.000850167769688045\\
0.00146829824275695	0.000850562344129948\\
0.0014681144525415	0.000850927976096036\\
0.00146792126731962	0.000851268890999304\\
0.00146771877116984	0.000851588066033354\\
0.0014675065123188	0.000851888384961548\\
0.00146728443410782	0.000852171497358568\\
0.00146705226525215	0.000852438856716109\\
0.00146681015044238	0.000852691160218086\\
};
\addlegendentry{BD: $\epsilon = 0.01$}

\addplot[only marks, mark=o, mark options={}, mark size=5pt, draw=mycolor6, forget plot] table[row sep=crcr]{%
x	y\\
0.00146792126731962	0.000851268890999304\\
};
\addplot[only marks, mark=x, mark options={}, mark size=5pt, draw=mycolor6, forget plot] table[row sep=crcr]{%
x	y\\
0.00146931598006191	0.0008451691026365\\
};
\addplot [color=mycolor6, line width=3.0pt, forget plot]
  table[row sep=crcr]{%
0.00146931598006191	0.0008451691026365\\
0.00146929838454492	0.0008460739359627\\
0.00146924392219391	0.000846865155314428\\
0.00146916024584968	0.000847570100281123\\
0.00146905402164446	0.000848197473863468\\
0.00146892945901967	0.000848761721462629\\
0.00146878894868572	0.000849275354436091\\
0.0014686366705857	0.000849740048392178\\
0.00146847258403339	0.000850167769688045\\
0.00146829824275695	0.000850562344129948\\
0.0014681144525415	0.000850927976096036\\
0.00146792126731962	0.000851268890999304\\
0.00146771877116984	0.000851588066033354\\
0.0014675065123188	0.000851888384961548\\
0.00146728443410782	0.000852171497358568\\
0.00146705226525215	0.000852438856716109\\
0.00146681015044238	0.000852691160218086\\
0.00146681015044238	0.000852691160218086\\
0.00146655353540688	0.00085293323578946\\
0.0014662843547074	0.000853162890053388\\
0.0014659999443247	0.000853382058906823\\
0.00146569873492164	0.000853591142986456\\
0.00146537786977228	0.000853790987445657\\
0.00146503583744612	0.000853981103932059\\
0.00146466901723951	0.000854161789887561\\
0.00146427408817534	0.000854332555068325\\
0.00146384588816425	0.000854493022963016\\
0.00146338001330438	0.000854641789717379\\
0.00146286837965225	0.000854777647926877\\
0.00146230537852938	0.00085489776865824\\
0.0014616804308002	0.000854999222340005\\
0.00146098104619341	0.000855077662321666\\
0.00146019189894725	0.000855127076148748\\
0.00145929133993952	0.000855139152254738\\
};
\addplot [opacity=0.667, color=mycolor6, dashed, line width=3.0pt]
  table[row sep=crcr]{%
0.00146121897207338	0.000846827535506497\\
0.00146097765394714	0.000847093950333192\\
0.00146072046558346	0.000847349487981516\\
0.00146044705057326	0.000847593672832021\\
0.00146015295046338	0.000847829260305024\\
0.00145983984942792	0.00084805325977882\\
0.0014595016939962	0.000848268353081352\\
0.00145913578163398	0.000848473809334482\\
0.00145873840121968	0.000848668898667813\\
0.00145830335024085	0.000848853209029696\\
0.00145782476309251	0.000849025106697699\\
0.00145729589640998	0.000849182402819355\\
0.00145670791857563	0.000849322132240472\\
0.0014560479376686	0.000849440501948768\\
0.00145530203038965	0.0008495318649295\\
0.00145445094325694	0.000849588684928226\\
0.00145347019769869	0.000849599956138923\\
0.00133535335023838	0.000848240585528493\\
0.000868558167512083	0.000842095233310384\\
0.000866006660522284	0.000840906431925238\\
0.000311204945422966	0.00028812611563894\\
0.000311163087059558	0.000288044007182338\\
0.00031111082767997	0.000287833840332238\\
0.000311022244068098	0.000287458393185244\\
0.000310178176636716	0.000281089953418288\\
0.000309037951403639	0.000259394766574548\\
0.000308547163468215	0.000218529816751153\\
0.000308135644918699	1.27562202441757e-05\\
0.000308142198405789	1.27543574073415e-05\\
0.0014514285412979	5.4934140812153e-06\\
0.00146253957895156	0.000392187410298758\\
0.00146316630554445	0.000637569357669725\\
0.00146329060060681	0.000839879224266291\\
0.00146327594387066	0.000840688293975485\\
0.00146323112408763	0.000841405744344246\\
0.00146316188308265	0.000842049058368922\\
0.00146307186276867	0.000842635311505379\\
0.00146296449834791	0.00084316902514609\\
0.00146284315951209	0.000843656449710219\\
0.0014627089825723	0.000844105394094888\\
0.00146256111196691	0.000844525977027221\\
0.00146240300583003	0.000844915655295614\\
0.00146223384011476	0.000845281008593198\\
0.00146205372032496	0.000845624844070102\\
0.00146186223345718	0.000845950105072241\\
0.0014616595928268	0.000846257880594908\\
0.00146144513405722	0.000846550216389214\\
0.00146121897207338	0.000846827535506497\\
};
\addlegendentry{BD: $\epsilon = 0.1$}

\addplot[opacity=0.667, only marks, mark=o, mark options={}, mark size=5pt, draw=mycolor6, forget plot] table[row sep=crcr]{%
x	y\\
0.00146205372032496	0.000845624844070102\\
};
\addplot[opacity=0.667, only marks, mark=x, mark options={}, mark size=5pt, draw=mycolor6, forget plot] table[row sep=crcr]{%
x	y\\
0.00146329060060681	0.000839879224266291\\
};
\addplot [opacity=0.667, color=mycolor6, line width=3.0pt, forget plot]
  table[row sep=crcr]{%
0.00146329060060681	0.000839879224266291\\
0.00146327594387066	0.000840688293975485\\
0.00146323112408763	0.000841405744344246\\
0.00146316188308265	0.000842049058368922\\
0.00146307186276867	0.000842635311505379\\
0.00146296449834791	0.00084316902514609\\
0.00146284315951209	0.000843656449710219\\
0.0014627089825723	0.000844105394094888\\
0.00146256111196691	0.000844525977027221\\
0.00146240300583003	0.000844915655295614\\
0.00146223384011476	0.000845281008593198\\
0.00146205372032496	0.000845624844070102\\
0.00146186223345718	0.000845950105072241\\
0.0014616595928268	0.000846257880594908\\
0.00146144513405722	0.000846550216389214\\
0.00146121897207338	0.000846827535506497\\
0.00146121897207338	0.000846827535506497\\
0.00146097765394714	0.000847093950333192\\
0.00146072046558346	0.000847349487981516\\
0.00146044705057326	0.000847593672832021\\
0.00146015295046338	0.000847829260305024\\
0.00145983984942792	0.00084805325977882\\
0.0014595016939962	0.000848268353081352\\
0.00145913578163398	0.000848473809334482\\
0.00145873840121968	0.000848668898667813\\
0.00145830335024085	0.000848853209029696\\
0.00145782476309251	0.000849025106697699\\
0.00145729589640998	0.000849182402819355\\
0.00145670791857563	0.000849322132240472\\
0.0014560479376686	0.000849440501948768\\
0.00145530203038965	0.0008495318649295\\
0.00145445094325694	0.000849588684928226\\
0.00145347019769869	0.000849599956138923\\
};
\addplot [opacity=0.333, color=mycolor6, dashed, line width=3.0pt]
  table[row sep=crcr]{%
0.00135974631155389	0.000724823769998848\\
0.00135953770287775	0.000724959097772451\\
0.00135930841105839	0.000725075859772852\\
0.00135905251731161	0.000725174061106641\\
0.00135876645041413	0.000725251582635129\\
0.00135844504190773	0.000725305738521597\\
0.00135808096420241	0.000725332966641797\\
0.00135766611149215	0.000725327838887627\\
0.000998009069685272	0.000715961083342091\\
0.000914368040713072	0.000708894576215582\\
0.000914367771984209	0.000708894495120505\\
0.000439462218988646	0.000242386283695094\\
0.000439429985030211	0.000242303816007562\\
0.000439422487311764	0.00024227668607942\\
0.000439422063599652	0.000242248314080196\\
0.000438844206686368	0.000199055503463383\\
0.000438801030726679	6.62467235231518e-05\\
0.000438801164092316	6.62467159828962e-05\\
0.00131528080903739	7.30592638172014e-05\\
0.00131699699587284	9.11530380275911e-05\\
0.00135553540283406	0.000619171704935501\\
0.00136121384433032	0.000718871337418467\\
0.00136124298802143	0.0007196171277229\\
0.00136124483413245	0.000720261942588815\\
0.00136122800179109	0.000720822696100897\\
0.00136118995232776	0.000721324317599137\\
0.00136114355284806	0.000721768145531832\\
0.0013610860059918	0.000722165399638182\\
0.00136101459266812	0.000722527344639531\\
0.00136093744679864	0.000722854438802073\\
0.00136085071548151	0.000723153187057033\\
0.00136075328367922	0.000723428042900895\\
0.00136064721010743	0.000723679627767221\\
0.00136052942406619	0.000723912888527262\\
0.00136040155899497	0.000724127127483274\\
0.0013602594943061	0.000724325767615328\\
0.0013601042977718	0.000724507841104727\\
0.00135993288654369	0.000724674256545998\\
0.00135974631155389	0.000724823769998848\\
};
\addlegendentry{BD: $\epsilon = 0.5$}

\addplot[opacity=0.333, only marks, mark=o, mark options={}, mark size=5pt, draw=mycolor6, forget plot] table[row sep=crcr]{%
x	y\\
0.00136052942406619	0.000723912888527262\\
};
\addplot[opacity=0.333, only marks, mark=x, mark options={}, mark size=5pt, draw=mycolor6, forget plot] table[row sep=crcr]{%
x	y\\
0.00136124483413245	0.000720261942588815\\
};
\addplot [opacity=0.333, color=mycolor6, line width=3.0pt, forget plot]
  table[row sep=crcr]{%
0.00136124483413245	0.000720261942588815\\
0.00136122800179109	0.000720822696100897\\
0.00136118995232776	0.000721324317599137\\
0.00136114355284806	0.000721768145531832\\
0.0013610860059918	0.000722165399638182\\
0.00136101459266812	0.000722527344639531\\
0.00136093744679864	0.000722854438802073\\
0.00136085071548151	0.000723153187057033\\
0.00136075328367922	0.000723428042900895\\
0.00136064721010743	0.000723679627767221\\
0.00136052942406619	0.000723912888527262\\
0.00136040155899497	0.000724127127483274\\
0.0013602594943061	0.000724325767615328\\
0.0013601042977718	0.000724507841104727\\
0.00135993288654369	0.000724674256545998\\
0.00135974631155389	0.000724823769998848\\
0.00135974631155389	0.000724823769998848\\
0.00135953770287775	0.000724959097772451\\
0.00135930841105839	0.000725075859772852\\
0.00135905251731161	0.000725174061106641\\
0.00135876645041413	0.000725251582635129\\
0.00135844504190773	0.000725305738521597\\
0.00135808096420241	0.000725332966641797\\
};

\addlegendimage{only marks, mark=o, mark options={solid}, mark size=4pt, line width=3pt, draw=mycolor8}
\addlegendentry{P-max}

\addlegendimage{only marks, mark=x, mark options={solid}, mark size=5pt, line width=3pt, draw=mycolor8}
\addlegendentry{E-max}

\addplot [draw=none, color=mycolor8, line width=3.0pt]
  table[row sep=crcr]{%
0	0\\
};
\addlegendentry{R-max}

\end{axis}
\end{tikzpicture}%

%% file: assets/simulation/pc_rate_csi.tex
%
%
\definecolor{mycolor1}{rgb}{0.00000,0.44706,0.74118}%
\definecolor{mycolor2}{rgb}{0.85098,0.32549,0.09804}%
\begin{tikzpicture}

\begin{axis}[%
width=9.509cm,
height=7.5cm,
at={(0cm,0cm)},
scale only axis,
xmin=-20,
xmax=20,
xlabel={Transmit Power [dB]},
ymin=0.872791961172002,
ymax=47.8680403270704,
ylabel={Achievable Rate [bit/s/Hz]},
axis background/.style={fill=white},
xmajorgrids,
ymajorgrids,
legend style={at={(0.03,0.97)}, anchor=north west, legend cell align=left, align=left, draw=white!15!black}
]
\addplot [color=black, line width=2.0pt]
  table[row sep=crcr]{%
-20	0.872791961172002\\
-15	2.07904737782427\\
-10	4.27044866001181\\
-5	7.57596761019195\\
0	12.0581965631964\\
5	17.7466867911886\\
10	24.0493143932381\\
15	30.5805423901838\\
20	37.1882780145721\\
};
\addlegendentry{No RIS}

\addplot [color=mycolor1, line width=2.0pt, mark=o, mark options={solid, mycolor1}]
  table[row sep=crcr]{%
-20	1.72199515744992\\
-15	3.59697161695806\\
-10	6.5895331682259\\
-5	10.9723869941346\\
0	17.3323710991557\\
5	23.7503231602655\\
10	30.3210481368116\\
15	36.941604234274\\
20	43.578074408526\\
};
\addlegendentry{D: $\epsilon = 0.01$}

\addplot [color=mycolor1, dashed, line width=2.0pt, mark=o, mark options={solid, mycolor1}]
  table[row sep=crcr]{%
-20	1.71536945206226\\
-15	3.58520666728276\\
-10	6.57366233471145\\
-5	10.9490393137998\\
0	17.2930208713544\\
5	23.7096575476636\\
10	30.2798063010204\\
15	36.9001857304648\\
20	43.5366039849325\\
};
\addlegendentry{D: $\epsilon = 0.1$}

\addplot [color=mycolor1, dotted, line width=2.0pt, mark=o, mark options={solid, mycolor1}]
  table[row sep=crcr]{%
-20	1.55441512858157\\
-15	3.33220429724211\\
-10	6.23113234284968\\
-5	10.420387175913\\
0	16.538499287718\\
5	22.9165434284263\\
10	29.4734527514333\\
15	36.0894704932352\\
20	42.7244977470308\\
};
\addlegendentry{D: $\epsilon = 0.5$}

\addplot [color=mycolor2, line width=2.0pt, mark=+, mark options={solid, mycolor2}]
  table[row sep=crcr]{%
-20	2.26316055471769\\
-15	4.77307899043111\\
-10	9.39639994943951\\
-5	15.124974996935\\
0	21.4464746896447\\
5	27.984503388636\\
10	34.5944862634857\\
15	41.2275888731308\\
20	47.8680403270704\\
};
\addlegendentry{BD: $\epsilon = 0.01$}

\addplot [color=mycolor2, dashed, line width=2.0pt, mark=+, mark options={solid, mycolor2}]
  table[row sep=crcr]{%
-20	2.24904712851749\\
-15	4.74754914484069\\
-10	9.34470319031593\\
-5	15.0631073403706\\
0	21.380523013882\\
5	27.9171868966866\\
10	34.5267244403906\\
15	41.1596859123352\\
20	47.8000925325886\\
};
\addlegendentry{BD: $\epsilon = 0.1$}

\addplot [color=mycolor2, dotted, line width=2.0pt, mark=+, mark options={solid, mycolor2}]
  table[row sep=crcr]{%
-20	1.93925626618002\\
-15	4.25529914900791\\
-10	8.29391410054909\\
-5	13.8015257433391\\
0	20.029137829252\\
5	26.534917281773\\
10	33.1342433437324\\
15	39.7639478491383\\
20	46.4033172121794\\
};
\addlegendentry{BD: $\epsilon = 0.5$}

\end{axis}
\end{tikzpicture}%

%% file: main.bbl
\begin{thebibliography}{10}
\providecommand{\url}[1]{#1}
\csname url@samestyle\endcsname
\providecommand{\newblock}{\relax}
\providecommand{\bibinfo}[2]{#2}
\providecommand{\BIBentrySTDinterwordspacing}{\spaceskip=0pt\relax}
\providecommand{\BIBentryALTinterwordstretchfactor}{4}
\providecommand{\BIBentryALTinterwordspacing}{\spaceskip=\fontdimen2\font plus
\BIBentryALTinterwordstretchfactor\fontdimen3\font minus \fontdimen4\font\relax}
\providecommand{\BIBforeignlanguage}[2]{{%
\expandafter\ifx\csname l@#1\endcsname\relax
\typeout{** WARNING: IEEEtran.bst: No hyphenation pattern has been}%
\typeout{** loaded for the language `#1'. Using the pattern for}%
\typeout{** the default language instead.}%
\else
\language=\csname l@#1\endcsname
\fi
#2}}
\providecommand{\BIBdecl}{\relax}
\BIBdecl

\bibitem{Basar2019}
E.~Basar, M.~D. Renzo, J.~D. Rosny, M.~Debbah, M.-S. Alouini, and R.~Zhang, ``Wireless communications through reconfigurable intelligent surfaces,'' \emph{IEEE Access}, vol.~7, pp. 116\,753--116\,773, 2019.

\bibitem{Wu2019}
Q.~Wu and R.~Zhang, ``Intelligent reflecting surface enhanced wireless network via joint active and passive beamforming,'' \emph{IEEE Transactions on Wireless Communications}, vol.~18, pp. 5394--5409, Nov 2019.

\bibitem{Guo2020}
H.~Guo, Y.-C. Liang, J.~Chen, and E.~G. Larsson, ``Weighted sum-rate maximization for reconfigurable intelligent surface aided wireless networks,'' \emph{IEEE Transactions on Wireless Communications}, vol.~19, pp. 3064--3076, May 2020.

\bibitem{Liu2022}
Y.~Liu, Y.~Zhang, X.~Zhao, S.~Geng, P.~Qin, and Z.~Zhou, ``Dynamic-controlled {RIS} assisted multi-user {MISO} downlink system: Joint beamforming design,'' \emph{IEEE Transactions on Green Communications and Networking}, vol.~6, pp. 1069--1081, Jun 2022.

\bibitem{He2022}
Y.~He, Y.~Cai, H.~Mao, and G.~Yu, ``{RIS}-assisted communication radar coexistence: Joint beamforming design and analysis,'' \emph{IEEE Journal on Selected Areas in Communications}, vol.~40, pp. 2131--2145, Jul 2022.

\bibitem{Luo2022}
H.~Luo, R.~Liu, M.~Li, Y.~Liu, and Q.~Liu, ``Joint beamforming design for {RIS}-assisted integrated sensing and communication systems,'' \emph{IEEE Transactions on Vehicular Technology}, vol.~71, pp. 13\,393--13\,397, Dec 2022.

\bibitem{Hua2023}
M.~Hua, Q.~Wu, C.~He, S.~Ma, and W.~Chen, ``Joint active and passive beamforming design for {IRS}-aided radar-communication,'' \emph{IEEE Transactions on Wireless Communications}, vol.~22, pp. 2278--2294, Apr 2023.

\bibitem{Wu2020a}
Q.~Wu and R.~Zhang, ``Joint active and passive beamforming optimization for intelligent reflecting surface assisted {SWIPT} under {QoS} constraints,'' \emph{IEEE Journal on Selected Areas in Communications}, vol.~38, no.~8, pp. 1735--1748, Aug 2020.

\bibitem{Feng2022}
Z.~Feng, B.~Clerckx, and Y.~Zhao, ``Waveform and beamforming design for intelligent reflecting surface aided wireless power transfer: Single-user and multi-user solutions,'' \emph{IEEE Transactions on Wireless Communications}, 2022.

\bibitem{Zhao2022}
Y.~Zhao, B.~Clerckx, and Z.~Feng, ``{IRS}-aided {SWIPT}: Joint waveform, active and passive beamforming design under nonlinear harvester model,'' \emph{IEEE Transactions on Communications}, vol.~70, pp. 1345--1359, 2022.

\bibitem{Karasik2020}
R.~Karasik, O.~Simeone, M.~D. Renzo, and S.~S. Shitz, ``Beyond max-{SNR}: Joint encoding for reconfigurable intelligent surfaces,'' in \emph{2020 IEEE International Symposium on Information Theory (ISIT)}, Jun 2020, pp. 2965--2970.

\bibitem{Ye2022}
J.~Ye, S.~Guo, S.~Dang, B.~Shihada, and M.-S. Alouini, ``On the capacity of reconfigurable intelligent surface assisted {MIMO} symbiotic communications,'' \emph{IEEE Transactions on Wireless Communications}, vol.~21, pp. 1943--1959, Mar 2022.

\bibitem{Liang2020}
Y.-C. Liang, Q.~Zhang, E.~G. Larsson, and G.~Y. Li, ``Symbiotic radio: Cognitive backscattering communications for future wireless networks,'' \emph{IEEE Transactions on Cognitive Communications and Networking}, vol.~6, pp. 1242--1255, Dec 2020.

\bibitem{Zhao2024}
Y.~Zhao and B.~Clerckx, ``{RIScatter}: Unifying backscatter communication and reconfigurable intelligent surface,'' \emph{IEEE Journal on Selected Areas in Communications}, pp. 1--1, Dec 2024.

\bibitem{Basar2021}
E.~Basar, ``Reconfigurable intelligent surfaces for doppler effect and multipath fading mitigation,'' \emph{Frontiers in Communications and Networks}, vol.~2, May 2021.

\bibitem{Arslan2022}
E.~Arslan, I.~Yildirim, F.~Kilinc, and E.~Basar, ``Over-the-air equalization with reconfigurable intelligent surfaces,'' \emph{IET Communications}, vol.~16, pp. 1486--1497, Aug 2022.

\bibitem{Ozdogan2020a}
O.~Ozdogan, E.~Bjornson, and E.~G. Larsson, ``Using intelligent reflecting surfaces for rank improvement in {MIMO} communications,'' in \emph{ICASSP 2020 - 2020 IEEE International Conference on Acoustics, Speech and Signal Processing (ICASSP)}, May 2020, pp. 9160--9164.

\bibitem{Yang2019}
Y.~Yang, B.~Zheng, S.~Zhang, and R.~Zhang, ``Intelligent reflecting surface meets {OFDM}: Protocol design and rate maximization,'' \emph{IEEE Transactions on Communications}, vol.~68, pp. 4522--4535, Jul 2020.

\bibitem{Chen2023}
G.~Chen and Q.~Wu, ``Fundamental limits of intelligent reflecting surface aided multiuser broadcast channel,'' \emph{IEEE Transactions on Communications}, vol.~71, pp. 5904--5919, Oct 2023.

\bibitem{ElMossallamy2021}
M.~A. ElMossallamy, H.~Zhang, R.~Sultan, K.~G. Seddik, L.~Song, G.~Y. Li, and Z.~Han, ``On spatial multiplexing using reconfigurable intelligent surfaces,'' \emph{IEEE Wireless Communications Letters}, vol.~10, pp. 226--230, Feb 2021.

\bibitem{Meng2023}
S.~Meng, W.~Tang, W.~Chen, J.~Lan, Q.~Y. Zhou, Y.~Han, X.~Li, and S.~Jin, ``Rank optimization for {MIMO} channel with {RIS}: Simulation and measurement,'' \emph{IEEE Wireless Communications Letters}, vol.~13, pp. 437--441, Feb 2024.

\bibitem{Huang2023}
W.~Huang, B.~Lei, S.~He, C.~Kai, and C.~Li, ``Condition number improvement of {IRS}-aided near-field {MIMO} channels,'' in \emph{2023 IEEE International Conference on Communications Workshops (ICC Workshops)}, May 2023, pp. 1210--1215.

\bibitem{Chae2023}
S.~H. Chae and K.~Lee, ``Cooperative communication for the rank-deficient {MIMO} interference channel with a reconfigurable intelligent surface,'' \emph{IEEE Transactions on Wireless Communications}, vol.~22, pp. 2099--2112, Mar 2023.

\bibitem{Shen2020a}
S.~Shen, B.~Clerckx, and R.~Murch, ``Modeling and architecture design of reconfigurable intelligent surfaces using scattering parameter network analysis,'' \emph{IEEE Transactions on Wireless Communications}, vol.~21, pp. 1229--1243, Feb 2022.

\bibitem{Nerini2023}
M.~Nerini, S.~Shen, and B.~Clerckx, ``Closed-form global optimization of beyond diagonal reconfigurable intelligent surfaces,'' \emph{IEEE Transactions on Wireless Communications}, vol.~23, pp. 1037--1051, Feb 2024.

\bibitem{Nerini2024}
M.~Nerini, S.~Shen, H.~Li, and B.~Clerckx, ``Beyond diagonal reconfigurable intelligent surfaces utilizing graph theory: Modeling, architecture design, and optimization,'' \emph{IEEE Transactions on Wireless Communications}, pp. 1--1, May 2024.

\bibitem{Santamaria2023}
I.~Santamaria, M.~Soleymani, E.~Jorswieck, and J.~Gutiérrez, ``{SNR} maximization in beyond diagonal {RIS}-assisted single and multiple antenna links,'' \emph{IEEE Signal Processing Letters}, vol.~30, pp. 923--926, 2023.

\bibitem{Santamaria2023a}
------, ``Interference leakage minimization in {RIS}-assisted {MIMO} interference channels,'' in \emph{ICASSP 2023 - 2023 IEEE International Conference on Acoustics, Speech and Signal Processing (ICASSP)}, vol.~39, Jun 2023, pp. 1--5.

\bibitem{Ahn2006}
H.-R. Ahn, \emph{Asymmetric Passive Components in Microwave Integrated Circuits}.\hskip 1em plus 0.5em minus 0.4em\relax Hoboken, NJ, USA: Wiley, 2006.

\bibitem{Li2024}
H.~Li, S.~Shen, Y.~Zhang, and B.~Clerckx, ``Channel estimation and beamforming for beyond diagonal reconfigurable intelligent surfaces,'' \emph{IEEE Transactions on Signal Processing}, vol.~72, pp. 3318--3332, Jan 2024.

\bibitem{Li2023f}
H.~Li, S.~Shen, M.~Nerini, M.~D. Renzo, and B.~Clerckx, ``Beyond diagonal reconfigurable intelligent surfaces with mutual coupling: Modeling and optimization,'' \emph{IEEE Communications Letters}, pp. 1--1, Oct 2024.

\bibitem{Li2024a}
H.~Li, M.~Nerini, S.~Shen, and B.~Clerckx, ``Wideband modeling and beamforming for beyond diagonal reconfigurable intelligent surfaces,'' in \emph{2024 IEEE 25th International Workshop on Signal Processing Advances in Wireless Communications (SPAWC)}, Jun 2024, pp. 926--930.

\bibitem{Li2023c}
H.~Li, S.~Shen, and B.~Clerckx, ``Beyond diagonal reconfigurable intelligent surfaces: A multi-sector mode enabling highly directional full-space wireless coverage,'' \emph{IEEE Journal on Selected Areas in Communications}, vol.~41, pp. 2446--2460, Aug 2023.

\bibitem{Tapie25e}
J.~Tapie, M.~Nerini, B.~Clerckx, and P.~del Hougne, ``{Beyond-Diagonal RIS Prototype and Performance Evaluation},'' \emph{arXiv [eess.SP]}, 19~May 2025.

\bibitem{Fang2023}
T.~Fang and Y.~Mao, ``A low-complexity beamforming design for beyond-diagonal {RIS} aided multi-user networks,'' \emph{IEEE Communications Letters}, pp. 1--1, Jul 2023.

\bibitem{Zhou2023}
Y.~Zhou, Y.~Liu, H.~Li, Q.~Wu, S.~Shen, and B.~Clerckx, ``Optimizing power consumption, energy efficiency and sum-rate using beyond diagonal {RIS} — a unified approach,'' \emph{IEEE Transactions on Wireless Communications}, pp. 1--1, 2023.

\bibitem{Bartoli2023}
G.~Bartoli, A.~Abrardo, N.~Decarli, D.~Dardari, and M.~D. Renzo, ``Spatial multiplexing in near field {MIMO} channels with reconfigurable intelligent surfaces,'' \emph{IET Signal Processing}, vol.~17, Mar 2023.

\bibitem{Semmler2023}
D.~Semmler, M.~Joham, and W.~Utschick, ``High {SNR} analysis of {RIS}-aided {MIMO} broadcast channels,'' in \emph{2023 IEEE 24th International Workshop on Signal Processing Advances in Wireless Communications (SPAWC)}, Sep 2023, pp. 221--225.

\bibitem{Hogben2013}
L.~Hogben, Ed., \emph{Handbook of Linear Algebra}.\hskip 1em plus 0.5em minus 0.4em\relax Boca Raton, FL, USA: CRC press, 2013.

\bibitem{Horn1994}
R.~A. Horn and C.~R. Johnson, \emph{Topics in Matrix Analysis}.\hskip 1em plus 0.5em minus 0.4em\relax Cambridge, UK: Cambridge University Press, Jun 1994.

\bibitem{Fulton2000}
W.~Fulton, ``Eigenvalues, invariant factors, highest weights, and schubert calculus,'' \emph{Bulletin of the American Mathematical Society}, vol.~37, pp. 209--249, Apr 2000.

\bibitem{Bhatia2001}
R.~Bhatia, ``Linear algebra to quantum cohomology: The story of alfred horn's inequalities,'' \emph{The American Mathematical Monthly}, vol. 108, pp. 289--318, Apr 2001.

\bibitem{Shen2021}
S.~Shen and B.~Clerckx, ``Beamforming optimization for {MIMO} wireless power transfer with nonlinear energy harvesting: {RF} combining versus {DC} combining,'' \emph{IEEE Transactions on Wireless Communications}, vol.~20, pp. 199--213, Jan 2021.

\bibitem{Rong2009a}
Y.~Rong and Y.~Hua, ``Optimality of diagonalization of multi-hop {MIMO} relays,'' \emph{IEEE Transactions on Wireless Communications}, vol.~8, no.~12, pp. 6068--6077, Dec 2009.

\bibitem{Zanella2009}
A.~Zanella, M.~Chiani, and M.~Win, ``On the marginal distribution of the eigenvalues of wishart matrices,'' \emph{IEEE Transactions on Communications}, vol.~57, pp. 1050--1060, Apr 2009.

\bibitem{Clerckx2013}
B.~Clerckx and C.~Oestges, \emph{{MIMO} Wireless Networks: Channels, Techniques and Standards for Multi-Antenna, Multi-User and Multi-Cell Systems}.\hskip 1em plus 0.5em minus 0.4em\relax Waltham, MA, USA: Academic Press, 2013.

\bibitem{Abrudan2008}
T.~E. Abrudan, J.~Eriksson, and V.~Koivunen, ``Steepest descent algorithms for optimization under unitary matrix constraint,'' \emph{IEEE Transactions on Signal Processing}, vol.~56, pp. 1134--1147, Mar 2008.

\bibitem{Abrudan2009}
T.~Abrudan, J.~Eriksson, and V.~Koivunen, ``Conjugate gradient algorithm for optimization under unitary matrix constraint,'' \emph{Signal Processing}, vol.~89, pp. 1704--1714, Sep 2009.

\bibitem{Hager2006}
W.~W. Hager and H.~Zhang, ``A survey of nonlinear conjugate gradient methods,'' \emph{Pacific journal of Optimization}, vol.~2, 2006.

\bibitem{Armijo1966}
L.~Armijo, ``Minimization of functions having lipschitz continuous first partial derivatives,'' \emph{Pacific Journal of Mathematics}, vol.~16, pp. 1--3, Jan 1966.

\bibitem{Moler2003}
C.~Moler and C.~V. Loan, ``Nineteen dubious ways to compute the exponential of a matrix, twenty-five years later,'' \emph{SIAM Review}, vol.~45, pp. 3--49, Jan 2003.

\bibitem{Liu2022c}
F.~Liu, Y.~Cui, C.~Masouros, J.~Xu, T.~X. Han, Y.~C. Eldar, and S.~Buzzi, ``Integrated sensing and communications: Toward dual-functional wireless networks for {6G} and beyond,'' \emph{IEEE Journal on Selected Areas in Communications}, vol.~40, pp. 1728--1767, Jun 2022.

\bibitem{Nie2017}
F.~Nie, R.~Zhang, and X.~Li, ``A generalized power iteration method for solving quadratic problem on the {Stiefel} manifold,'' \emph{Science China Information Sciences}, vol.~60, p. 112101, Nov 2017.

\bibitem{Manton2002}
J.~H. Manton, ``Optimization algorithms exploiting unitary constraints,'' \emph{IEEE Transactions on Signal Processing}, vol.~50, pp. 635--650, Mar 2002.

\bibitem{Viklands2006}
T.~Viklands, ``Algorithms for the weighted orthogonal procrustes problem and other least squares problems,'' PhD dissertation, Datavetenskap, Ume{\aa} University, 2006.

\bibitem{Bell2003}
T.~Bell, ``Global positioning system-based attitude determination and the orthogonal procrustes problem,'' \emph{Journal of Guidance, Control, and Dynamics}, vol.~26, pp. 820--822, Sep 2003.

\bibitem{Golub2013}
G.~H. Golub and C.~F.~V. Loan, \emph{Matrix Computations}.\hskip 1em plus 0.5em minus 0.4em\relax Baltimore, MD, USA: Johns Hopkins University Press, 2013.

\bibitem{Boumal14a}
\BIBentryALTinterwordspacing
N.~Boumal, B.~Mishra, P.-A. Absil, and R.~Sepulchre, ``{M}anopt, a {M}atlab toolbox for optimization on manifolds,'' \emph{Journal of Machine Learning Research}, vol.~15, no.~42, pp. 1455--1459, 2014. [Online]. Available: \url{https://www.manopt.org}
\BIBentrySTDinterwordspacing

\bibitem{Ikramov2012}
K.~D. Ikramov, ``Takagi's decomposition of a symmetric unitary matrix as a finite algorithm,'' \emph{Computational Mathematics and Mathematical Physics}, vol.~52, pp. 1--3, Jan 2012.

\bibitem{An23b}
J.~An, C.~Xu, D.~W.~K. Ng, G.~C. Alexandropoulos, C.~Huang, C.~Yuen, and L.~Hanzo, ``{Stacked Intelligent Metasurfaces for Efficient Holographic MIMO Communications in {6G}},'' \emph{IEEE Journal on Selected Areas in Communications}, vol.~41, no.~8, pp. 2380--2396, Aug. 2023.

\bibitem{Zhang2005}
F.~Zhang, Ed., \emph{\BIBforeignlanguage{en}{The Schur Complement and Its Applications}}, ser. Numerical Methods and Algorithms.\hskip 1em plus 0.5em minus 0.4em\relax New York, NY, USA: Springer, Apr 2005.

\bibitem{Marshall2010}
A.~W. Marshall, I.~Olkin, and B.~C. Arnold, \emph{\BIBforeignlanguage{en}{Inequalities: Theory of Majorization and Its Applications}}, 2nd~ed., ser. Springer Series in Statistics.\hskip 1em plus 0.5em minus 0.4em\relax New York, NY, USA: Springer, Dec 2010.

\bibitem{Watson1992}
G.~A. Watson, ``Characterization of the subdifferential of some matrix norms,'' \emph{Linear Algebra and its Applications}, vol. 170, no.~1, pp. 33--45, 1992.

\bibitem{Hjorungnes2007}
A.~Hjorungnes and D.~Gesbert, ``Complex-valued matrix differentiation: Techniques and key results,'' \emph{IEEE Transactions on Signal Processing}, vol.~55, pp. 2740--2746, Jun 2007.

\end{thebibliography}
